\documentclass{amsproc}
\usepackage[utf8]{inputenc}
\usepackage{amsmath,amsthm,amsfonts,amscd,amssymb,amstext}
\usepackage{mathtools}

\numberwithin{equation}{section}
\usepackage{booktabs}
\usepackage{esint}
\usepackage{relsize}
\usepackage{graphicx,color}
\usepackage{tabularx}
\usepackage{titletoc}
\usepackage{microtype}
\usepackage{caption}
\usepackage{subcaption}
\usepackage[colorlinks=true,linkcolor=black,citecolor=black,urlcolor=blue]{hyperref}
\usepackage{listings}
\usepackage{cite}
\usepackage{url}
\usepackage{alltt}
\usepackage{tikz}
\usepackage{bbm}
\usetikzlibrary{arrows,cd,decorations.markings,positioning}
\usepackage{tkz-graph}
\tikzstyle{vecArrow} = [thick, decoration={markings,mark=at position
   1 with {\arrow[semithick]{open triangle 60}}},
   double distance=1.4pt, shorten >= 5.5pt,
   preaction = {decorate},
   postaction = {draw,line width=1.4pt, white,shorten >= 4.5pt}]
\tikzstyle{innerWhite} = [semithick, white,line width=1.4pt, shorten >= 4.5pt]

\newtheorem{theorem}{Theorem}
\newenvironment{thmbis}[1]
  {\renewcommand{\thetheorem}{\ref{#1}$'$}%
   \addtocounter{theorem}{-1}%
   \begin{theorem}}
  {\end{theorem}}
\numberwithin{theorem}{section}

\newtheorem{lem}[theorem]{Lemma}
\newtheorem{prop}[theorem]{Proposition}
\newtheorem{cor}[theorem]{Corollary}

\newtheorem*{conj}{Conjecture}

\newtheorem*{rmk}{Remark}
\theoremstyle{remark}

\newcommand{\barr}{\begin{eqnarray}}
\newcommand{\earr}{\end{eqnarray}}
\newcommand{\be}{\begin{equation}}
\newcommand{\ee}{\end{equation}}

\newcommand{\de}{\operatorname{d}}

\newcommand{\numberset}{\mathbb}
\newcommand{\N}{\numberset{N}}
\newcommand{\Z}{\numberset{Z}}
\newcommand{\R}{\numberset{R}}
\newcommand{\C}{\numberset{C}}
\newcommand{\Q}{\numberset{Q}}
\newcommand{\Hh}{\numberset{H}}

\newcommand{\tr}{\operatorname{Tr}}
\newcommand{\Tr}{\operatorname{Tr}}

\newcommand{\E}{\mathbb{E}}
\newcommand{\MM}{\mathcal{M}}

\newcommand{\GUE}{\mathrm{GUE}}
\newcommand{\LUE}{\mathrm{LUE}}
\newcommand{\JUE}{\mathrm{JUE}}
\newcommand{\GSE}{\mathrm{GSE}}
\newcommand{\LSE}{\mathrm{LSE}}
\newcommand{\JSE}{\mathrm{JSE}}
\newcommand{\GOE}{\mathrm{GOE}}
\newcommand{\LOE}{\mathrm{LOE}}
\newcommand{\JOE}{\mathrm{JOE}}

\begin{document}

\title[Moments of random matrices and hypergeometric OP's]
{Moments of random matrices and\\ hypergeometric orthogonal polynomials}

\author[Cunden]{Fabio Deelan Cunden}
\address[F. D. Cunden]{\newline
School of Mathematics and Statistics, University College Dublin, Belfield, Dublin 4, Ireland}

\author[Mezzadri]{Francesco Mezzadri}
\address[F. Mezzadri]{\newline School of Mathematics, University of Bristol, University Walk, Bristol BS8 1TW, United Kingdom}

\author[O'Connell]{Neil O'Connell}
\address[N. O'Connell]{\newline
School of Mathematics and Statistics, University College Dublin, Belfield, Dublin 4, Ireland}

\author[Simm]{Nick Simm}
\address[N. Simm]{\newline Mathematics Department, University of Sussex, Brighton, BN1 9RH, United Kingdom}

\date{\today}
\begin{abstract}
We establish a new connection between moments of $n \times n$ random matrices $X_n$ and hypergeometric orthogonal polynomials. Specifically, we consider moments $\mathbb{E}\Tr X_n^{-s}$ as a function of the complex variable $s \in \mathbb{C}$, whose analytic structure we describe completely. We discover several remarkable features, including a reflection symmetry (or functional equation),  zeros on a critical line in the complex plane, and orthogonality relations. An application of the theory resolves part of an integrality conjecture of Cunden \textit{et al.}~[F. D. Cunden, F. Mezzadri, N. J. Simm and P. Vivo, J. Math. Phys. 57 (2016)] on the time-delay matrix of chaotic cavities. In each of the classical ensembles of random matrix theory (Gaussian, Laguerre, Jacobi) we characterise the moments in terms of the Askey scheme of hypergeometric orthogonal polynomials. We also calculate the leading order $n\to\infty$ asymptotics of the moments and discuss their symmetries and zeroes.  We discuss aspects of these phenomena beyond the random matrix setting, including the Mellin transform of products and Wronskians of pairs of classical orthogonal polynomials. When the random matrix model has orthogonal or symplectic symmetry, we obtain a new duality formula relating their moments to hypergeometric orthogonal polynomials.
\end{abstract}

\maketitle
\tableofcontents

\section{Introduction}
\label{sec:intro}
In this paper we present a novel approach to the moments of the classical ensembles of random matrices.  Much of random matrix theory is devoted to moments $\E\Tr X_n^k$ ($k\in\N$) of random matrices of finite or asymptotically large size $n$. The Gaussian, Laguerre and Jacobi unitary ensembles have been extensively studied and virtually everything is known about the moments as functions of the matrix size $n$. In particular, for the GUE,  $\E\Tr X_n^k$ is a polynomial in $n$. This fact is a consequence of Wick's theorem, it is usually called `genus expansion', and it is at the heart of several successful theories such as the topological recursion~\cite{Ambjorn90,Eynard15}. For example, the $4$-th moment of GUE matrices of size $n$ is
\[
\frac{1}{n}\E\Tr X_n^8=14n^4+70n^2+21.
\]
\par
In contrast to the wealth of results on moments as functions of the size $n$, less attention has been devoted to them as functions of the order $k$.
One of the consequences is that some remarkable properties have been somehow missed.
The theory described in this paper is intended to fill this gap.
By looking at the moments as functions of $k$, we gain access to additional structure. 
Several results contained in this paper are in fact facets of the same phenomenon, which appears to be a new observation: \emph{moments $\E\Tr X_n^k$ of classical matrix ensembles, if properly normalized, are hypergeometric orthogonal polynomials as functions of $k$.}  For example, for a GUE random matrix of size $n=4$
\[
\frac{1}{(2k-1)!!}\E\Tr X_4^{2k}=\frac{4}{3}k^3+4k^2+\frac{20}{3}k+4,
\]
and this polynomial is actually a Meixner polynomial. In fact, the moments are essentially Meixner polynomials as functions of $(n-1)$, too.
\par
During our investigation it became natural to consider \emph{complex} moments $\E\Tr X_n^k$ ($k\in\C$) or, equivalently, averages of spectral zeta functions of random matrices.

\subsection{Spectral zeta functions of random matrices}
\label{sub:zeta}
There exist various generalizations of the Riemann  $\zeta$-function, associated with operator spectra and which are generically called \emph{spectral zeta functions}. 
\par
Consider a compact operator $A$ on a separable Hilbert space. Then $AA^*$ is a nonnegative operator, so that $|A|= AA^*$ makes sense. The singular values of $A$ are defined as the (nonzero) eigenvalues of $|A|$. If $A$ is self-adjoint with discrete spectrum $\lambda_1,\lambda_2,\dots$, the singular values are $|\lambda_1|,|\lambda_2|,\dots$. The  Dirichlet series representation of the Riemann zeta function $\zeta(s)$ suggests to define the spectral zeta function $\zeta_{A}(s)$ of the operator $A$ as the maximal analytic continuation of the series
\be
\sum_{j\geq 1}|\lambda_j|^{-s}\nonumber
\ee
(this is also called Minakshisundaram–Pleijel~\cite{MP49} zeta function of $A$.)
In this sense, the Riemann $\zeta(s)$ is the spectral zeta of the integer spectrum $\lambda_j=j$. Several authors have posed the question of how the `spectral' properties of Riemann's zeta function carry over (or not) to various spectral zeta functions~\cite{Voros87}; classical properties of the Riemann zeta function are:
\begin{enumerate}
\item Functional equation: the function $\xi(s)=\pi^{-s/2}\Gamma(s/2)\zeta(s)$ satisfies $\xi(1-s)=\xi(s)$;
\item Meromorphic structure: $\zeta(s)$ is analytic in $\C\setminus\{1\}$;
\item Special values: trivial zeros $\zeta(-2j)=0$ for $j=1,2,\dots$;
\item Complex zeros of $\zeta(s)$ and the Riemann hypothesis (RH): the complex Riemann zeros are in the critical strip $0<\operatorname{Re}(s)<1$ and enjoy the reflection symmetries along the real axis and the line $\operatorname{Re}(s)=1/2$. It is  conjectured (RH) that the nontrivial zeros all lie on the critical line $\operatorname{Re}(s)=1/2$.
\end{enumerate}
\par
Let $X_n$ be a $n\small\times n$ random Hermitian matrix, and denote by $\lambda_1,\dots,\lambda_n$ its eigenvalues. Assume that, with full probability, $0$ is not in the spectrum of $X_n$ (this is certainly true for the classical ensembles of random matrices). We proceed to define the averaged spectral zeta function $\zeta_{X_n}(s)$  as the maximal analytic continuation of 
\be
\E\,\zeta_{X_n}(s)=\E\Tr |X_n|^{-s}=\E\sum_{j=1}^n|\lambda_j|^{-s}.\nonumber
\ee
Note that $\E\,\zeta_{X_n}(s)$ is not a random function. Much of the paper is devoted to pointing out the analytic structure of $\E\,\zeta_{X_n}(s)$  when $X_n$ comes from the Gaussian, Laguerre or Jacobi ensembles.

\subsection{Time-delay matrix of chaotic cavities}
\label{sub:RMT_QT}
Random matrix theory provides a mathematical framework to develop a statistical theory of quantum transport. This theory is believed to apply in particular to mesoscopic conductors confined in space, often referred as \emph{quantum dots}, connected to the environment through ideal leads. For these systems, Brouwer, Frahm and  Beenakker~\cite{BroFraBee97_99}, showed that the proper delay times are distributed as the inverse of the eigenvalues of matrices $X_n$ in the Laguerre ensemble (the size $n$ being the number of scattering channels) with Dyson index $\beta\in\{1,2,4\}$ labelling the classical symmetry classes, and parameter $\alpha=n$. 
\par
The moments of the proper delay times have been studied using both random matrix theory~\cite{Cunden14,Cunden16,Cunden16b,Grabsch15,Livan11,Simm11,Savin01,Savin01b} and semiclassical scattering orbit theory~\cite{Novaes15I,BerkKuip11,Sieber14}. Of course, the subject of moments on rather general matrix ensembles has been extensively studied. There is, however, one important complication here: the moments of the time-delay matrix are \emph{singular} spectral linear statistic on the Laguerre ensemble. (Invariant random matrix ensembles with singular potentials received considerable interest in recent years in mathematical physics, see, e.g.\ Refs.~\cite{Akemann14,Atkin16,Berry08,Chen10,Mezzadri09,Brightmore15}.)
\par
In~\cite{Cunden14,Cunden16,Cunden16b}, it was conjectured that the $1/n$-expansion of the cumulants of power traces for the time-delay matrix of quantum dots has positive integer coefficients. In this paper we prove that the conjecture is true for the first order cumulants, i.e. the moments, when $\beta=2$ (systems without broken time reversal symmetry).

\subsection{Mellin transform of orthogonal polynomials} 
\label{sub:Mellin} 
The averaged zeta function is related to the Mellin transform of the one-point correlation function\footnote{See Section \ref{sec:def} for the definition of the one-point function and the important identity \eqref{tracedens}.}. In the classical unitary invariant ensembles, by using the well-known determinantal formulae and Christoffel-Darboux formula, the one-point correlation function can be written as a Wronskian 
\be
\rho^{(2)}_n(x)=\sum_{j=0}^{n-1}\psi_j^2(x)=\frac{k_{n-1}}{k_n}\operatorname{Wr}(\psi_{n-1}(x),\psi_{n}(x)),
\nonumber
\ee
where $\psi_j(x)$ are the \emph{normalized} Hermite, Laguerre or Jacobi wavefunctions with leading coefficient $k_j$. (The superscript stands for $\beta=2$.)
The Mellin transform of a function $f(x)$ is defined by the integral
\begin{equation}
\MM\left[f(x);s\right] = \int_0^\infty f(x)x^{s-1}dx,
\nonumber
\end{equation}
when it exists. We set $f^*(s) = \MM\left[f(x);s\right]$. Of course, $\E\,\zeta_{X_n}(s)=\rho^{(2), *}_n(1-s)$. In all instances in this paper, the Mellin transforms have meromorphic extensions to all of $\C$ with simple poles (see Appendix~\ref{eq:app_Mellin}).
\par
Bump and Ng~\cite{Bump86} and Bump, Choi, Kurlberg and Vaaler~\cite{Bump00} made the remarkable observation that the Mellin transforms of Hermite and Laguerre functions $\psi_j(x)$ form families of orthogonal polynomials (OP's) and have zeros on the critical line $\operatorname{Re}(s)=1/2$. (Jacobi functions were not considered by them.) 
A few years later, Coffey~\cite{Coffey07} and Coffey and Lettington~\cite{Coffey15} pointed out that the polynomials described by Bump \textit{et al.} were hypergeometric OP's and investigated other families.
\par
Indeed, we show that for the classical matrix ensembles, the Mellin transform $\rho^{(2), *}_n(s)$ of a Wronskian of two adjacent wavefunctions is a hypergeometric OP (up to a factor containing ratios of Gamma functions). We stress that the proof does not go along the lines of the method of Bump \textit{et al.}. They started from the orthogonality of the classical wavefunctions which is preserved by the Mellin transform (a unitary operator in $L^2$). In our case, by explicit computations, we identify a discrete Sturm-Liouville (S-L) problem satisfied by $\rho^{(2), *}_n(s)$ (as a function of $s$) and this turns out to be the same S-L of the classical hypergeometric OP's.
We remark that the Mellin transforms $\rho^{(2), *}_n(s)$ of $\rho^{(2)}_n(x)$ do have a probabilistic meaning (moments of random matrices). The Mellin transforms $\psi_j^*(s)$ studied in~\cite{Bump86,Bump00}, while not unmotivated, do not have an obvious probabilistic interpretation.
\par
Once the analytic structure of the Mellin transform of $\operatorname{Wr}(\psi_n(x),\psi_{n+1}(x))$ was established, it became natural for us to look for similar polynomial properties for Wronskians of nonadjacent wavefunctions $\operatorname{Wr}(\psi_n(x),\psi_{n+k}(x))$, $k>1$. Such Wronskians do not have a random matrix interpretation. Nevertheless, they have a certain interest in mathematical physics as they appear when applying Darboux-Crum~\cite{Crum55} transformations on a Schr\"odinger operator to generate families of \emph{exceptional orthogonal polynomials}~\cite{Garcia-Ferrero15,Kuijlaars15,Odake11,Gomez-Ullate14}.

\subsection{Orthogonal and symplectic ensembles}
\label{sub:OS}
The theory developed for the classical ensembles of complex random matrices suggested to look for similar polynomial properties in the real and quaternionic cases (orthogonal and symplectic ensembles). 
\par
Now a fundamental insight came from recursion formulae satisfied by orthogonal / symplectic moments coupled with moments of the corresponding unitary ensembles (see~\cite{Ledoux09} in the Gaussian case and~\cite{Cunden16b} in the Laguerre ensemble).
It turns out that for the classical ensembles of random matrices with orthogonal and symplectic symmetries, certain combinations of moments satisfy three term recursion formulae which, again, correspond to the S-L equations defining families of hypergeometric OP's. Therefore, this combination of moments plays the role of the single moments in the unitary case: they satisfy three term recursions, have hypergeometric OP factors, reflection symmetries, zeros on a vertical line, etc. 
\par
The (single) moments of the symplectic ensembles do have polynomial factors, but these do not belong to the Askey scheme. In the orthogonal cases, we use a novel duality formula (based on the results by Adler \emph{et al.}~\cite{AFNvM}) to write the moments of real random matrices of odd dimension as quaternionic moments plus a remainder containing an orthogonal polynomial factor. 
\par
Coupling this result with a classical duality between orthogonal and symplectic ensembles, we discover a functional equation for moments of real random matrices.

\subsection{Outline} The paper has the following structure:
\begin{itemize}
\item In Section~\ref{sec:app} the physics motivations and application to quantum transport in chaotic cavities are presented;
\item In Section~\ref{sec:def} we set some notation and we recall the definition of the classical ensembles of random matrices and hypergeometric OP's;
\item In Section~\ref{sec:results} we present the main results along with their proofs for the Gaussian, Laguerre and Jacobi unitary ensembles;
\item In Section~\ref{sec:asymptotics} we discuss the large-$n$ asymptotics of the spectral zeta functions; 
\item In Section~\ref{sec:extension} we discuss the relation of our findings with earlier works on the Mellin transform of classical orthogonal polynomials; then, we extend our results beyond random matrix theory by considering Mellin transforms of products and Wronskians of generic pairs of orthogonal polynomials;
\item In Section~\ref{sec:further} we discuss the extension of duality formulae between moments of random matrices to higher order cumulants; 
\item In Section~\ref{sec:Orth_Sympl} we consider the classical orthogonal and symplectic ensembles and, in particular, present
a new duality formula relating their moments to hypergeometric orthogonal polynomials.
\end{itemize}
\section{Motivation and applications}
\label{sec:app}
One of the original motivations for this work was to make some progress on an integrality conjecture for the $1/n$-expansion of the Laguerre ensemble put forward in~\cite{Cunden14,Cunden16,Cunden16b}. The problem originated from a random matrix approach to quantum transport in chaotic cavities. 
\par
In this section, we will first present some of our findings on the Laguerre ensemble. Then we will briefly review the connection between the Laguerre ensemble and the time-delay matrix in chaotic cavities, and explain the applicability of our results to the integrality conjecture.

\subsection{The Laguerre ensemble: reciprocity law and spectral zeta function}
\label{sub:zeta2}
Let $X_n$ be a random matrix from the Laguerre Unitary Ensemble (LUE) with parameter $m\geq n$. That is, $X_n$ is distributed according to 
\be
d\mathbb{P}(X)=\frac{1}{Z}(\det X^{\alpha})\exp(-\Tr X)dX
\label{eq:LUEmeasure1}
\ee
on the space $\mathcal{P}_n$ of nonnegative Hermitian matrices, where $dX$ is Lebesgue measure on $\mathcal{P}_n\simeq\R^{n^2}$, $Z$ the
normalizing constant, and $\alpha=m-n$. 
\par
Consider for integer $k\in\N$, the (inverse) moments
\be
\E\Tr X_n^{-k}=\E\sum_{j=1}^n\lambda_j^{-k},
\ee
where $\lambda_1,\dots,\lambda_n$ are the eigenvalues of $X_n$. The above moments are finite if and only if  $k\leq\alpha$~\cite{Kumari17}. We will prove the following remarkable property of these moments.
\begin{prop}[Reciprocity law for LUE]
\be
\E\Tr X_n^{-(k+1)}=\left(\prod_{j=-k}^k\frac{1}{\alpha+j}\right)\E\Tr X_n^{k}.
\label{eq:dualityLUE}
\ee
\end{prop}
The above identity can be verified using the recurrence relation for  moments proved by Haagerup and Thorbj\o rnsen~\cite{Haagerup03} and extended to inverse moments in~\cite{Cunden16b}. This gives an explicit formula for the inverse moments given known identities for positive moments.
For instance,
\begin{align*}
\E\Tr X_n^{0}&=n &\E\Tr X_n^{-1}&=\frac{n}{\alpha} \\
\E\Tr X_n^{1}&=n^2+\alpha n &\E\Tr X_n^{-2}&=\frac{n^2+\alpha n}{(\alpha -1) \alpha  (\alpha +1)} \\
\E\Tr X_n^{2}&=2 n^3+3 \alpha  n^2+\alpha ^2 n &\E\Tr X_n^{-3}&=\frac{2 n^3+3 \alpha  n^2+\alpha ^2 n }{(\alpha -2) (\alpha -1) \alpha  (\alpha +1) (\alpha +2)} .
\end{align*}
\par
It is natural to consider complex moments or, equivalently, the averaged LUE spectral zeta function defined as  
\be
\E\zeta_{X_n}(s)=\E\Tr X_n^{-s}=\E\sum_{j}\lambda_j^{-s},\quad \text{for $\operatorname{Re}(s)\leq\alpha$},
\ee
and by analytic continuation for other values of $s$. We list below a few remarkable properties of the averaged LUE spectral zeta function.
\begin{enumerate}
\item Functional equation: the reciprocity law~\eqref{eq:dualityLUE} suggests to consider the function
\be
\xi_{n}(s)=\frac{1}{\Gamma(1+\alpha-s)}\E\,\zeta_{X_n}(s),
\ee
so that~\eqref{eq:dualityLUE} becomes the functional equation
\be
\xi_{n}(1-s)=\xi_{n}(s).
\ee
\item Analytic structure: it turns out that $\E\,\zeta_{X_n}(s)$ can be analytically extended to the whole complex plane; In particular, $\xi_{n}(s)$  is a polynomial of degree  $2(n-1)$. 
\item Special values: trivial zeros $\E\,\zeta_{X_n}(1+\alpha+j)=0$ for $j=1,2,\dots$;
\item Complex zeros and Riemann hypothesis: as for the Riemann zeta function, the set of complex zeros is symmetric with respect to reflections along the real axis and the critical line $\operatorname{Re}(s)=1/2$. It is tempting to ask whether a RH holds true for the averaged LUE zeta function. Amusingly, the answer is `Yes': \emph{the nontrivial zeros of $\E\,\zeta_{X_n}(s)$ all lie on the critical line $\operatorname{Re}(s)=1/2$}.
\end{enumerate}
These facts are an immediate consequences of the main results presented in Section~\ref{sec:results}. See Fig.~\ref{fig:LUEzeta} for a plot of the averaged LUE zeta function and Fig.~\ref{fig:LUEzeta2} for an illustration of the zeros.
\begin{rmk}
For any fixed $n$, the function $\zeta_{X_n}(s)$ is a finite sum of exponentials. Therefore, without taking the average, $\zeta_{X_n}(s)$ is a random analytic function in $\C$ and never vanishes.
\end{rmk}
\begin{figure}
\centering
\includegraphics[width=.8\columnwidth]{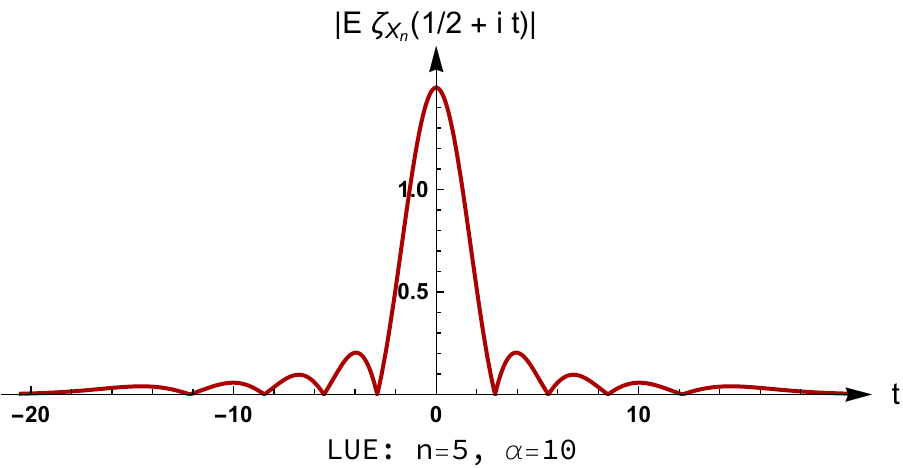}
\caption{Expected LUE spectral zeta function (modulus) on the critical line $s=1/2+it$. }
\label{fig:LUEzeta}
\end{figure}

\subsection{Application to quantum transport in chaotic cavities}
\label{sub:QT}
In~\cite{Cunden14,Cunden16}, it was proposed that, for $\beta\in\{1,2\}$, the cumulants of the time-delay matrix of a ballistic chaotic cavity have a $1/n$-expansion \emph{with positive integer coefficients} (similar to the genus expansion of Gaussian matrices). The precise conjectural statement is as follows.
\par
Consider the measure~\eqref{eq:LUEmeasure1} with $\alpha=n$, and the rescaled inverse power traces 
\begin{equation*}
\tau_k(n)=n^{k-1}\Tr X_n^{-k}\quad (k\geq0). 
\end{equation*}
It is known that the expectation of  $\tau_k(n)$ has a $1/n$-expansion
\begin{equation*}
\E\tau_{k}(n)=\sum_{g=0}^{\infty}\kappa_{g}(k)n^{-g}.
\end{equation*}
\begin{conj}[\cite{Cunden16b}]  
$\kappa_{g}(k)\in\N$.
\end{conj}
\par
The conjecture was supported by a systematic computation of certain generating functions, and it is in agreement with the diagrammatic expansions of scattering orbit theory. The integrality of the  coefficients in the large-$n$ expansion has been also conjectured in the real case (LOE)~\cite{Cunden16b}, and for higher order cumulants~\cite{Cunden14,Cunden16}.
\par
The results reported in this paper resolve the conjecture in the complex case.  
\begin{theorem} The above conjecture is true. 
\end{theorem}
\begin{proof} To prove the Theorem we take advantage of the reciprocity law to use known results for positive moments of the Laguerre ensemble. 
\par
 Let $X_n$ be in the LUE with parameter $\alpha=m-n$. For $k\geq0$, from~\cite[Corollary 2.4]{Hanlon92} (see also~\cite[Exercise 12]{Mingo17}) we read the formula
\be
\frac{1}{n^{k+1}}\E\Tr X_n^k= \sum_{\sigma\in S_k}n^{\#(\sigma)+\#(\gamma_k\sigma^{-1})-(k+1)}\left(\frac{m}{n}\right)^{\#(\sigma)},
\label{eq:moments_+ve}
\ee 
where $S_{k}$ is the symmetric group, and for a permutation $\sigma \in S_{k}$, $\#(\sigma)$ denotes the number of cycles in $\sigma$. By $\gamma_{k}$ we denote the $k$-cycle $(1\; 2\; 3\; \ldots\; k)$. 
If $m=cn$ with $c>0$, the above formula shows that $\frac{1}{n^{k+1}}\E\Tr X_n^k$ is a polynomial in $n^{-2}$ with positive coefficients (see Lemma~\ref{lem:Wishart} below). 

By the reciprocity law~\eqref{eq:dualityLUE} with $\alpha=n$ (or, equivalently, $c=2$)
\be
\E\tau_{k+1}(n)=\left(\prod_{j=1}^k\frac{1}{1-\frac{j^2}{n^2}}\right)\frac{1}{n^{k+1}}\E\Tr X_n^k.
\label{eq:tau}
\ee
The factor $\prod_{j=1}^k\left(1-\frac{j^2}{n^2}\right)^{-1}$ in~\eqref{eq:tau} is a product of geometric series. Therefore, we have
\be
\E\tau_{k+1}(n)=\sum_{\sigma\in S_k}\sum_{i_1,\dots,i_k=0}^{\infty}\left(2^{\#(\sigma)}\prod_{j=1}^kj^{2i_j}\right)n^{\#(\sigma)+\#(\gamma_k\sigma^{-1})-(k+1)-2(i_1+\cdots+i_k)},
\ee 
and this readily prove that $\E\tau_k$ has an expansion in $n^{-2}$ with positive integer coefficients.
\end{proof}
\begin{rmk} From~\eqref{eq:moments_+ve} we see that, if $c\in\N$, then $\frac{1}{n^{k+1}}\E\Tr X_n^k$ ($k\geq0$ integer) has a $1/n$-expansion with positive integer coefficients. The computation above shows that the integrality of the coefficients for the LUE negative moments also holds whenever $c/(c-1)\in\N$. Therefore, $c=2$ is the only case when all moments $\frac{1}{n^{k+1}}\E\Tr X_n^k$ ($k\in\Z$) have integer coefficients in their large-$n$ expansion. 
\end{rmk}

\section{Notation and definitions}
\label{sec:def}
\subsection{Classical ensembles of random matrices}
We will consider expectations with respect to the measures 
\begin{equation}
\label{eq:integrals}
\frac{1}{C_{n,\beta}}
\prod_{j=1}^{n}
w_{\beta}(x_{j})\chi_I(x_{j})\prod_{1 \leq j < k \leq n}|x_{k}-x_{j}|^{\beta}dx_{1}
\dotsm dx_{n}
\end{equation}
for finite $n$ and for any value of $\beta \in \{1,2,4\}$.
The value of $\beta$ corresponds to ensembles of real symmetric
($\beta=1$), complex hermitian ($\beta=2$) or quaternion self-dual
matrices ($\beta=4$).  The function $w_\beta(x)$ is the weight of the ensemble:
\begin{equation}
\label{eq:weights}
w_{\beta}(x) = \begin{cases}
 e^{-(\beta/2) x^{2}/2}, & I=\R,  \\
 x^{(\beta/2) (m-n+1)-1}\;e^{-(\beta/2) x},& I=\R_+,
   \\
(1-x)^{(\beta/2) (m_1-n+1)-1}\; x^{(\beta/2) (m_2-n+1)-1}, &
 I=[0,1], 
   \end{cases}
\end{equation}
for Gaussian, Laguerre, and Jacobi, respectively. $C_{n,\beta}$ is a normalization constant which depends on the ensemble and is known explicitly~\cite{Forrester_book}.
For convenience we set $\alpha = m - n>0$ in the Laguerre ensemble and $\alpha_1=m_1 -n>0,\; \alpha_2 =m_2 -n >0$ in the Jacobi
ensemble. 

We define the \textit{one-point eigenvalue density} $\rho^{(\beta)}_{n}(x)$ corresponding to \eqref{eq:integrals} by
\begin{equation}
\rho^{(\beta)}_{n}(x)= \mathbb{E}\left(\sum_{j=1}^{n}\delta(x-x_{j})\right) \label{dens}
\end{equation}
We will call \eqref{dens} the eigenvalue density corresponding to the weight $w_{\beta}(x)$ defining the expectation over \eqref{eq:integrals}. The following identity easily follows from the definitions \eqref{dens} and \eqref{eq:integrals}:
\begin{equation}
\mathbb{E}\tr X_{n}^{k} = \int_{I}\,x^{k} \rho^{(\beta)}_{n}(x)\,dx. \label{tracedens}
\end{equation}

\subsection{Hypergeometric orthogonal polynomials}
We use the standard notation for hypergeometric functions
\be
{}_p F_q\left(\begin{matrix}a_1,\dots, a_p\\b_1,\dots, b_q\end{matrix}; z\right)=
\sum_{j=0}^{\infty}\frac{(a_1)_j\cdots (a_p)_j}{(b_1)_j\cdots (b_q)_j}\frac{z^j}{j!},
\nonumber
\ee
where $(q)_n=\Gamma(q+n)/\Gamma(q)$.
We need to introduce some families of hypergeometric OP's~\cite{Koekoek10}. Recall that there are three types of hypergeometric OP's:
\begin{enumerate}
\item Polynomials of the first type are solutions of usual S-L problems for second order differential operators: Jacobi $P^{(\alpha_1,\alpha_2)}_n(x)$ and its degenerations, Laguerre $L^{(\alpha)}_n(x)$ and Hermite $H_n(x)$. They have a hypergeometric representation, but they are perhaps better known by the Rodrigues-type formulae
\begin{align}
H_n(x)&=(-1)^n\,e^{x^2}\frac{d^n}{d x^n}e^{-x^2},
\label{eq:Hermite}\\
L_n^{(\alpha)}(x)&=\frac{1}{n!}\,x^{-\alpha}e^{x}\frac{d^n}{d x^n}x^{n+\alpha}e^{-x},
\label{eq:Laguerre}\\
P_n^{(\alpha_1,\alpha_2)}(x)&=\frac{1}{n!}\frac{(-1)^n}{(1-x)^{\alpha_1}x^{\alpha_2}}\frac{d^n}{d x^n}(1-x)^{\alpha_1+n} x^{\alpha_2+n}.
\label{eq:Jacobi}
\end{align}
\par
Note that the Jacobi polynomials considered in this paper are orthogonal with respect to the measure $(1-x)^{\alpha_1}x^{\alpha_2}dx$ on the unit interval $[0,1]$;
\item Polynomials of the second type are solutions of discrete S-L problems (three-terms recurrence relations) with real coefficients:   Racah $R_n(\lambda(x);\alpha,\beta,\gamma,\delta)$, including its degenerations Hahn $Q_n(x;\alpha,\beta,N)$, dual Hahn $R_n(\lambda(x);\gamma,\delta,N)$, Meixner $M_n(x;\beta,c)$, etc. They can be represented as finite hypergeometric  series
\begin{align}
R_n(\lambda(x);\alpha,\beta,\gamma,\delta)&=
\setlength\arraycolsep{1pt}
\; {}_4 F_3\left(\begin{matrix}-n,n+\alpha+\beta+1,-x,x+\gamma+\delta+1\\\alpha+1,\beta+\delta+1,\gamma+1 \end{matrix}; 1\right)\\
Q_n(x;\alpha,\beta,N)&=
{}_3 F_2\left(\begin{matrix}-n,n+\alpha+\beta+1,-x \\\alpha+1,-N \end{matrix}; 1\right)
\label{eq:Hahn}\\
R_n(\lambda(x);\gamma,\delta,N)&=
{}_3 F_2\left(\begin{matrix}-n,-x,x+\gamma+\delta+1 \\ \gamma+1,-N \end{matrix}; 1\right)
\label{eq:DualHahn}\\
M_n(x;\beta,c)&=
{}_2 F_1\left(\begin{matrix}-n,-x \\ \beta \end{matrix}; 1-\frac{1}{c}\right),
\label{eq:Meixner}
\end{align}
where $\lambda(x)=x(x+\gamma+\delta+1)$. Note that some authors define the Meixner polynomials as $m_n(x;\beta,c)=(\beta)_nM_n(x;\beta,c)$;
\item Polynomials of the third type are solutions of discrete S-L problems with complex coefficients:  Wilson $W_n(x^2;a,b,c,d)$ including its degenerations, continuous dual Hahn $S_n(x^2;\;a,b,c)$, continuous Hahn $p_n(x;a,b,c,d)$, Meixner-Pollaczek $P_n^{(\lambda)}(x;\phi)$, etc.
They have the following hypergeometric representations:
\begin{align}
W_n(x^2;a,b,c,d)&=(a+b)_n(a+c)_n(a+d)_n\nonumber\\
&\small\times\setlength\arraycolsep{1pt}
\; {}_4 F_3\left(\begin{matrix}-n,n+a+b+c+d-1,a+ix,a-ix\\a+b,a+c,a+d \end{matrix}; 1\right)
\label{eq:Wilson}\\
S_n(x^2;a,b,c)&=(a+b)_n(a+c)_n\;
{}_3 F_2\left(\begin{matrix}-n,a+ix,a-ix \\a+b,a+c \end{matrix}; 1\right)
\label{eq:ContinuousDualHahn}\\
p_n(x;a,b,c,d)&= \frac{i^n}{n!}(a+c)_n(a+d)_n\nonumber\\
&\small\times{}_3 F_2\left(\begin{matrix}-n,n+a+b+c+d-1,a+ix \\a+c,a+d \end{matrix}; 1\right)
\label{eq:ContinuousHahn}\\
P_n^{(\lambda)}(x;\phi)&=(2\lambda)_n\frac{e^{in\phi}}{n!}
{}_2 F_1\left(\begin{matrix}-n,\lambda+i x \\2\lambda \end{matrix}; 1-e^{2i\phi}\right).
\label{eq:Meixner-Pollaczek}
\end{align}
\end{enumerate}
At the top of the hierarchy of hypergeometric orthogonal polynomials are the Wilson and Racah polynomials.
In this paper, the parameter ranges are such that it is most natural to consider the polynomials which appear 
as Wilson polynomials and their degenerations, specifically continuous dual Hahn and Meixner-Pollaczek polynomials.
The reader can find the common notation and the most important properties of hypergeometric OP's in~\cite[Section 9]{Koekoek10}.

\section{Unitary ensembles}
\label{sec:results}

It is known that the $k$-th moments of the classical unitary invariant ensembles of random matrices $X_n$ of dimension $n$ are polynomials in $n$ (or in $1/n$ after rescaling). Here we show that the (completed) moments can also be seen as polynomials in the parameter $k$.  These polynomials are hypergeometric orthogonal polynomials OP's belonging to the Askey scheme~\cite{Askey85}. The polynomial property suggests to consider complex moments or, equivalently, averages of spectral zeta functions. 
For the three classical ensembles, define
\begin{equation*}
\zeta_{X_n}(s)=\Tr |X_n|^{-s},\quad X_n\sim\{\GUE,\LUE,\JUE\},
\end{equation*}
and
\begin{equation*}
\xi_{n}(s)=
\left\{  \begin{array}{l@{\quad}cr} 
\displaystyle\frac{2^{2s}}{ \Gamma\left(1/2-2s\right)}\;\E\,\zeta_{X_n}(4s)&&\text{if $X_n\sim \GUE$}\ ,  \\ \\
\displaystyle\frac{1}{\Gamma(1+\alpha-s)}\;\E\,\zeta_{X_n}(s)&&\text{if $X_n\sim \LUE$}\ ,  \\ \\
\displaystyle\frac{\Gamma(1+\alpha_1+\alpha_2+2n-s)}{\Gamma(1+\alpha_2-s)}\;
\E\left(\zeta_{X_n}(s)-\zeta_{X_n}(s-1)\right)&&\text{if $X_n\sim \JUE$}\ ,
\end{array}\right.
\label{eq:spectralzets_ensembles}
\end{equation*}
when the expectations exist ($s<1/4$, $s<\alpha+1$, and $s<\alpha_2+1$ for $\GUE$, $\LUE$, and $\JUE$, respectively) and by analytic continuation otherwise (see Appendix~\ref{eq:app_Mellin}).
\par
We can now state the first result.
\begin{theorem}
\label{thm:zetafinite}
 For all $n$, $\xi_{n}(s)$ is a hypergeometric orthogonal polynomial:
\begin{equation*}
\xi_{n}(s)=
\left\{  \begin{array}{l@{\quad}cr} 
\displaystyle \frac{i^{1-n}}{\sqrt{\pi}}\;P^{(1)}_{n-1}\left(2x;\;\pi/2\right)
&\text{if $X_n\sim\GUE$}  \\ \\
\displaystyle\frac{1}{\Gamma(n)\Gamma(\alpha+n)}\;S_{n-1}\left(x^2;\; \frac{3}{2}, \frac{1}{2}, \alpha+\frac{1}{2}\right)&\text{if $X_n\sim\LUE$}  \\ \\
\displaystyle\frac{\Gamma(\alpha_1+\alpha_2+n+1)}{\Gamma(n)\Gamma(\alpha_2+n)}(-1)^{n-1}(\alpha_1+n)\\
\small\times\;\displaystyle W_{n-1}\left(x^2;\; \frac{3}{2}, \frac{1}{2}, \alpha_2+\frac{1}{2}, \frac{1}{2}-\alpha_1-\alpha_2-2n\right)&\text{if $X_n\sim\JUE$}\ ,
\end{array}\right.
\label{eq:spectralzets_ensembles2}
\end{equation*}
where $x=i(1/2-s)$. In particular, $\xi_n(s)$ satisfies the functional equation $\xi_{n}(s)=\xi_{n}(1-s)$ in the LUE and JUE cases, and $\xi_{n}(s) = (-1)^{n-1}\xi_{n}(1-s)$ for the GUE. Moreover, all its zeros lie on the critical line $\operatorname{Re}(s)=1/2$.
\end{theorem}
\begin{proof}
See Theorems~\ref{thm:GUE}, \ref{thm:LUE} and \ref{thm:JUE} below. 
\end{proof}
\par
The weights of the OP's in Theorem~\ref{thm:zetafinite} are
\be
w(s)=
\left\{  \begin{array}{l@{\quad}cr} 
\displaystyle\left|2\sqrt{\pi}\Gamma(2s)\right|^2
&\text{if $X_n\sim$ GUE}  \\ \\
\displaystyle\left|\frac{\Gamma(s)\Gamma(s+1)\Gamma(s+\alpha)}{\Gamma(2s-1)}\right|^2&\text{if $X_n\sim$ LUE}  \\ \\
\displaystyle\left|\frac{\Gamma(s)\Gamma(s+1)\Gamma(s+\alpha_2)}{\Gamma(s+\alpha_1+\alpha_2+2n)\Gamma(2s-1)}\right|^2&\text{if $X_n\sim$ JUE}\ ,
\end{array}\right.
\label{eq:spectralzets_ortho}\nonumber
\ee
For the GUE and LUE, the following orthogonality conditions hold
\be
\displaystyle\frac{1}{2\pi i}\int\limits_{\frac{1}{2}+i\R_+}\xi_{m}(s)\overline{\xi_{n}}(s)w(s)d s=h_m\;\delta_{mn},
\label{orth}
\ee
where $h_m$ is an explicit constant depending on the ensemble (see Appendix~\ref{app:neoOP}).

For an illustration of the zeros on the critical line, see Fig~\ref{fig:LUEzeta2} and~\ref{fig:JUEzeta}. For the reader's convenience, the relation between moments of the unitary ensembles and hypergeometric OP's is summarised in Table~\ref{tab:sum}.
\begin{rmk}
The orthogonality in the JUE is slightly different to the GUE and LUE cases. First of all, the fourth parameter of the Wilson polynomial is negative. An orthogonality relation in this case is far from obvious and was established by Neretin \cite{Neretin02}, see also Appendix \ref{app:neoOP}. This fourth parameter also depends on $n$, therefore each $\xi_n(s)$ belongs to a distinct family of orthogonal polynomials obtained by fixing the fourth parameter. As before, this orthogonality implies that the zeros lie on the line $\mathrm{Re}(s)=1/2$. 
\end{rmk}
\begin{figure}
\centering
\includegraphics[width=.49\columnwidth]{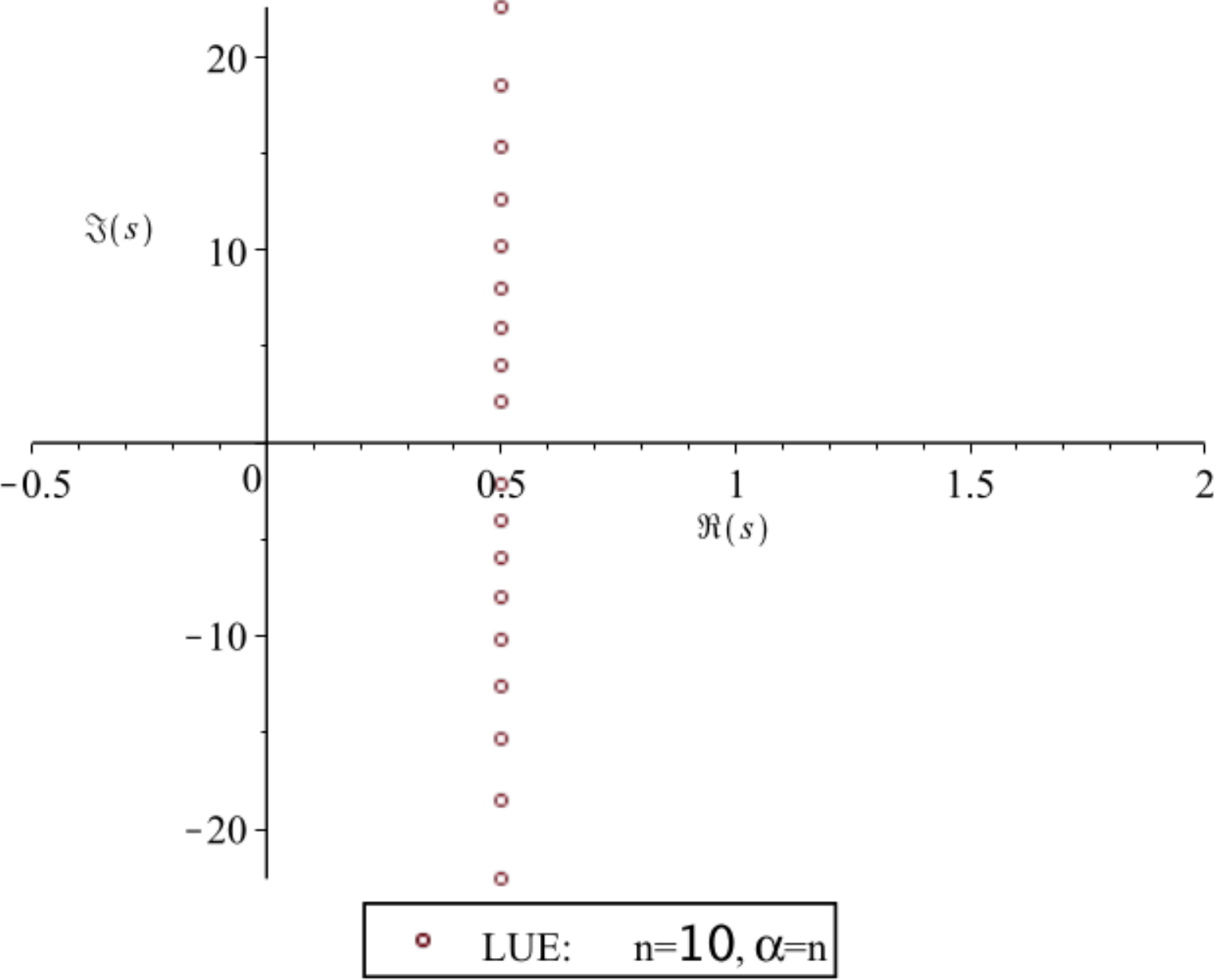}\,
\includegraphics[width=.49\columnwidth]{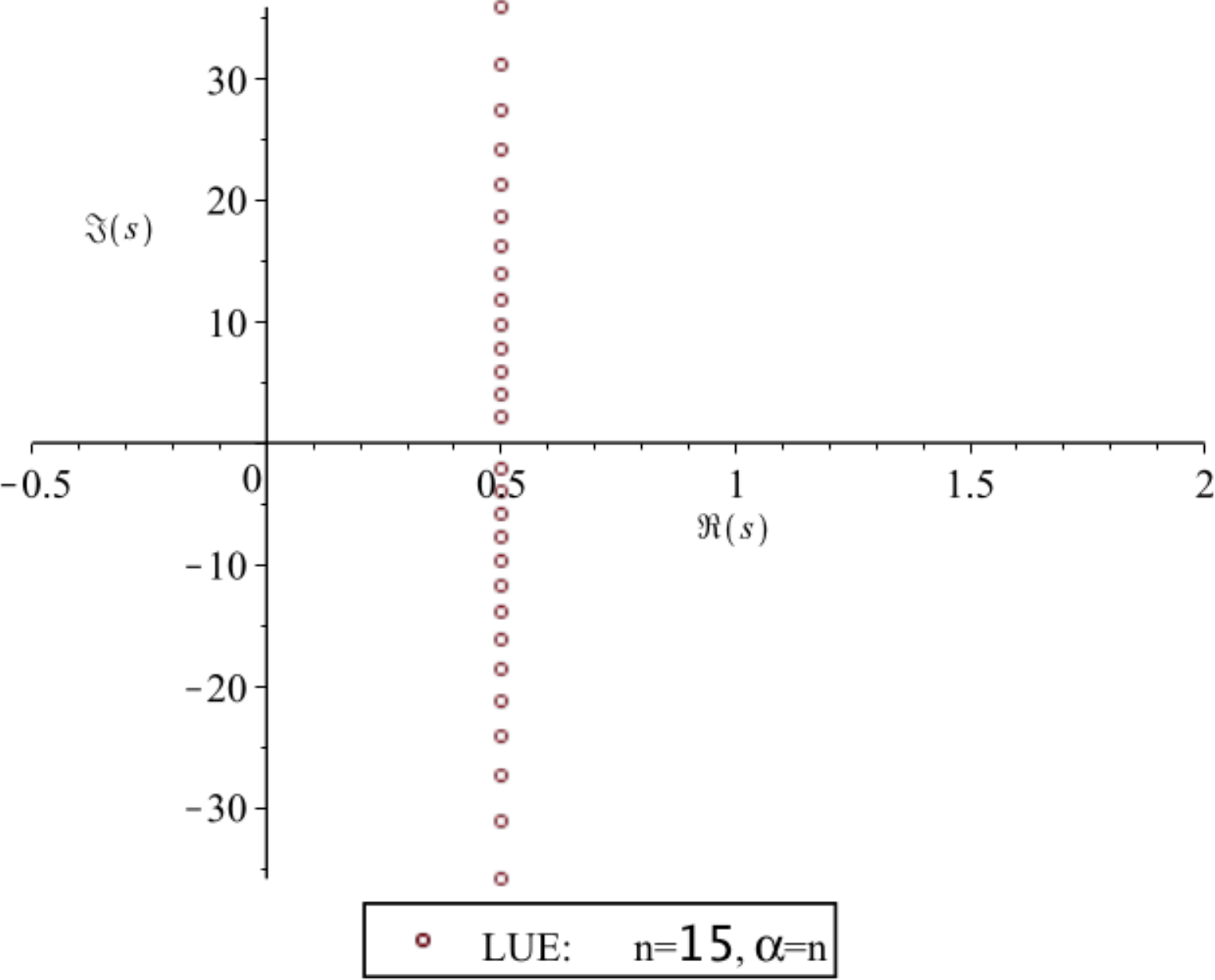}
\caption{Zeros of $\xi_{n}(s)$ for the LUE.}
\label{fig:LUEzeta2}
\end{figure}
\begin{figure}
\centering
\includegraphics[width=.49\columnwidth]{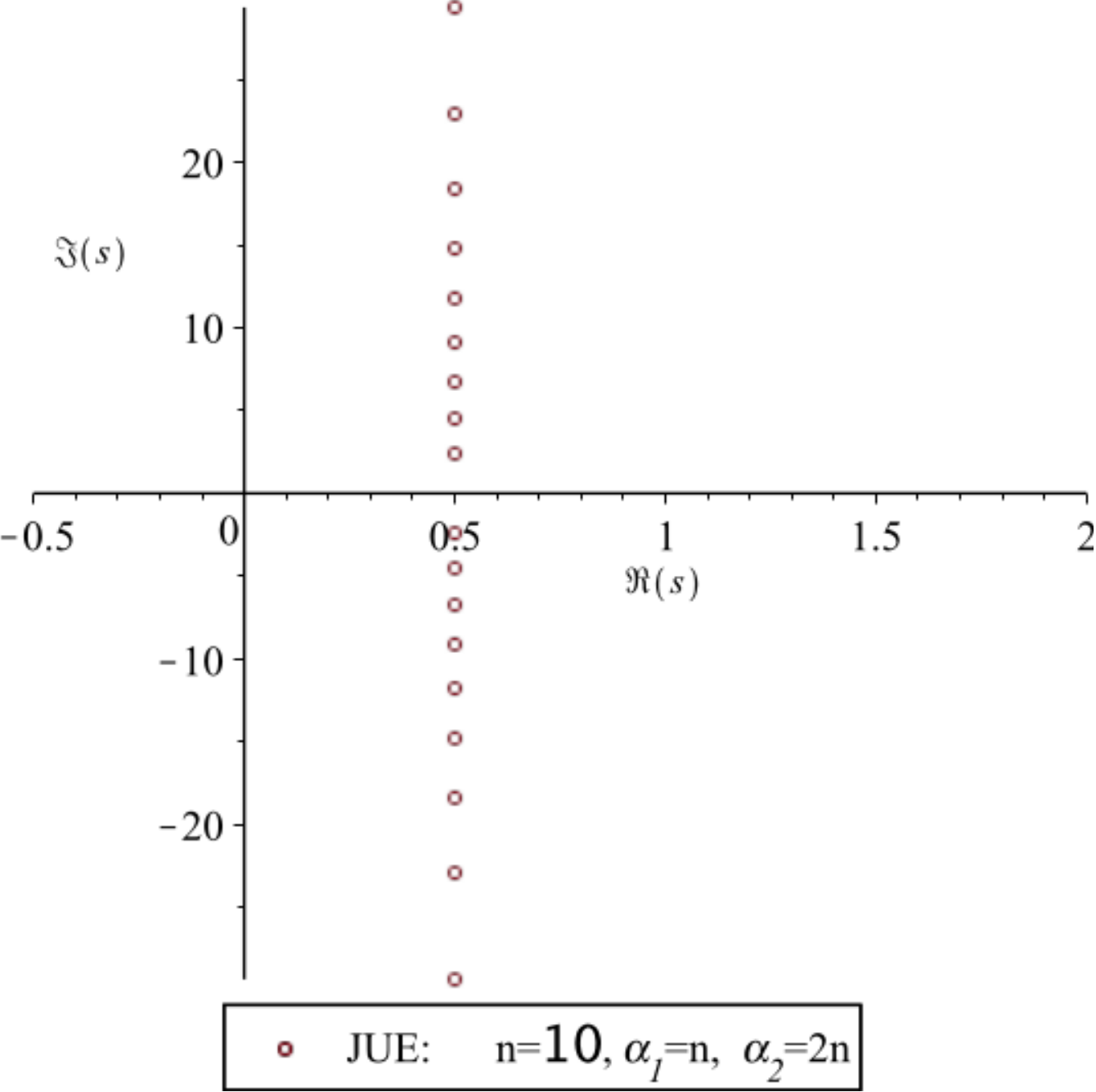}\,
\includegraphics[width=.49\columnwidth]{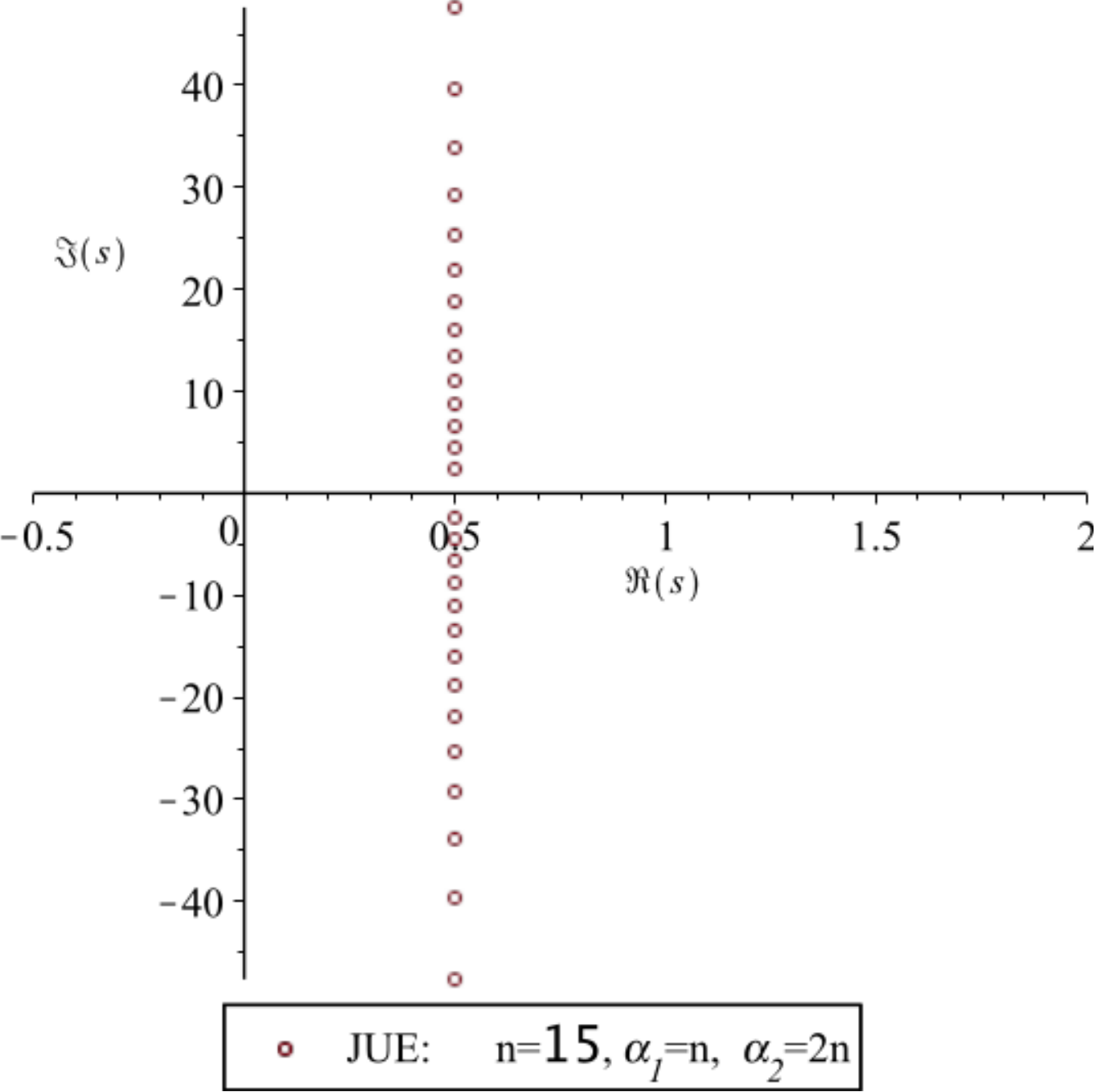}
\caption{Zeros of $\xi_{n}(s)$ for the JUE.}
\label{fig:JUEzeta}
\end{figure}
\par

\subsection{Gaussian Unitary Ensemble}
\label{sec:gaussian}
The GUE is a classical orthogonal polynomial ensemble. In particular, the correlation functions can be compactly and conveniently written in terms of Hermite polynomials. It turns out that the complex moments are essentially a Meixner-Pollaczek polynomial. The moments for GOE and GSE can be expressed using known formulae relating the one-point correlation functions of the three Gaussian ensembles. 
\par
Let $X_n$ be a GUE random matrix of dimension $n$. Define $Q_{k}^{\C}(n)=\E\Tr X_n^{2k}$ for all $k\in\C$ for which the expectation exists. It is known~\cite{Harer86} that, for $k\in\N$,
\be
Q_{k}^{\C}(n)
=(2k-1)!! \sum_{i=1}^{n} 2^{i-1} \binom{n}{i} \binom{k}{i-1} ,
\label{eq:HarerGUE}
\ee
where 
$(2k-1)!!=2^k \Gamma(1/2 + k)/\sqrt{\pi}$ (this is equal to $(2k-1)(2k-3)\cdots 1 $ for $k\geq1$ integer).
For each $k$, the moment $Q_{k}^{\C}(n)$ is a polynomial in $n$ with positive integer coefficients:
\begin{align*}
Q_0^{\C}(n) &=n \\
Q_1^{\C}(n) &=n^2 \\
Q_2^{\C}(n) &=2n^3+n\\
Q_3^{\C}(n) &=5n^4+10n^2. \\
Q_4^{\C}(n) &= 14n^5+70n^3+21n .
\end{align*}
This fact is the well-known genus expansion for Gaussian complex matrices.
\par
As observed in~\cite[Theorem 8]{Witte14}, ~\eqref{eq:HarerGUE} can be written in terms of a (terminating) hypergeometric series:
\be
\frac{Q_{k}^{\C}(n)}{n\; (2k-1)!!} =
\setlength\arraycolsep{1pt}
\; {}_2 F_1\left(\begin{matrix}-k,1-n\\2\end{matrix};2\right).
\label{eq:GUE_2F1}
\ee
From this hypergeometric representation, we see that the moment $Q_{k}^{\C}(n) $, if properly normalised, is a Meixner-Pollaczek polynomial in $i(k+1)$ of degree $n-1$. 
\begin{theorem}
\label{thm:GUE} If we write $x=i(k+1)$, then for $\operatorname{Re}(k)>-1/2$
\be
Q_{k}^{\C}(n) =  i^{1-n}\;(2k-1)!! \; P^{(1)}_{n-1}(x;\;\pi/2)
\label{eq:GUE_Meixner-Pollaczek}
\ee
In particular, $i^{n-1}Q_{k}^{\C}(n)/(2k-1)!!$ can be extended to an analytic function in $\C$ (a polynomial), invariant  up to a change of sign under reflection $k\to-2-k$, with complex zeros on the vertical line $\operatorname{Re}(k)=-1$.
\end{theorem} 
\begin{proof}
Consider the polynomials
\begin{equation*}
q_r(s)= \frac{i^{-r}}{1+r} P_r^{(1)}(is;\ \pi/2).
\end{equation*}
From  the definition of Meixner-Pollaczek polynomials~\eqref{eq:Meixner-Pollaczek}  of $P^{(1)}_{n-1}(x;\;\pi/2)$ and the hypergeometric representation~\eqref{eq:GUE_2F1} we see that $ Q_{k}^{\C}(n) = n\; (2k-1)!! \; q_{k}(n)$ when $k$ is a nonnegative integer. In order to prove the general complex case, we use a procedure of analytic continuation from integer points to a complex domain via Carlson's theorem~\cite[Theorem 2.8.1]{Andrews99}. 
A standard calculation in random matrix theory shows that, for $\operatorname{Re}(k)>-1/2$,
\begin{equation*}
Q_{k}^{\C}(n) =\int_{0}^{\infty}y^{2k}\frac{e^{-y^2/2}}{\sqrt{2\pi}}\sum_{j=0}^{n-1}\left(\frac{H_j(y/\sqrt{2})}{2^j j!}\right)^2dy,
\end{equation*}
where $H_j$ denotes the Hermite polynomial~\eqref{eq:Hermite} of degree $j$. This shows that $Q_{k}^{\C}(n)$ is analytic in the half-plane $\operatorname{Re}(k)>-1/2$. If we write  $\sum_{j=0}^{n-1}({2^j j!})^{-2}\left(H_j(y/\sqrt{2})\right)^2=c_{2n-2}y^{2n-2}+\cdots+c_1y+c_0$ for some constants $c_i$, then
\begin{align*}
\left|Q_{k}^{\C}(n)\right| &=\left|\int_{0}^{\infty}y^{2k}\frac{e^{-y^2/2}}{\sqrt{2\pi}}\sum_{j=0}^{2n-2}c_jy^jdx\right|\leq\sum_{j=0}^{2n-2}\left|c_{j}\right|\int_{0}^{\infty}y^{2k+j}\frac{e^{-y^2/2}}{\sqrt{2\pi}}dy\\
\end{align*}
We use now the elementary inequality $a+by+cy^2\leq (a+b)+(b+c)y^2$  for $a,b,c,y\geq0$. Setting $d_i=|c_{2j-1}|+2|c_{2j}|+|c_{2j+1}|$ (with $c_{ -1}=c_{2n-1}=0$) we have then
\begin{align*}
\left|Q_{k}^{\C}(n)\right| &\leq\sum_{j=0}^{n-1}d_j\int_{0}^{\infty}y^{2k+2j}\frac{e^{-y^2/2}}{\sqrt{2\pi}}dy\\
&=\sum_{j=0}^{n-1}d_j\frac{2^{k+j-1}}{\sqrt{\pi}}\Gamma(k+j+1/2)\\
&=\frac{2^k}{\sqrt{\pi}}\Gamma(k+1/2)\sum_{j=0}^{n-1}d_j2^{j-1}\prod_{i=0}^{j-1}(k+1/2+i).
\end{align*}
Therefore $i^{n-1}Q_{k}^{\C}(n)/(2k-1)!!=\operatorname{O}(k^{n-1})$ as $|k|\to\infty$ with $\operatorname{Re}(k)>-1/2$. In conclusion, $Q_{k}^{\C}(n)/(2k-1)!!$ and the polynomial $P^{(1)}_{n-1}(i(k+1);\;\pi/2)$ coincide on nonnegative integers and their difference is $\operatorname{O}(e^{c|k|})$ for any $c>0$. By Carlson's theorem the two functions coincide in the whole domain $\operatorname{Re}(k)>-1/2$.

The polynomials $q_r(s)$ enjoy the symmetries 
\begin{align*}
q_{r-1}(s)&=q_{s-1}(r)\\ 
q_r(-s)&=(-1)^r q_r(s).
\end{align*} 
Therefore $Q_{k}^{\C}(n) = n\; (2k-1)!! \; q_{n-1}(k+1)=n\; (2k-1)!! \; (-1)^{n-1}q_{n-1}(-k-1)$ thus explaining the reflection symmetry $k\to-2-k$. Recall that the Meixner-Pollaczek polynomial form an orthogonal family with respect to a positive weight on the real line. Therefore, the zeros of $q_r(s)$ are purely imaginary;
they occur in conjugate pairs, with zero included if $r$ is odd. 
 This proves that all the zeros of $i^{n-1}Q_{k}^{\C}(n)/(2k-1)!!$ lie on the line $\operatorname{Re}(k)=-1$.
\end{proof}
\begin{rmk}
The polynomials $q_r(s)$ 
satisfy the difference equation
\begin{equation*} 
(s+1)q_r(s+1)=2(r+1)q_r(s)+(s-1) q_r(s-1),
\end{equation*}
and the three-term recurrence
\begin{equation*} 
(r+2) q_{r+1}(s)=2s q_r(s)+r q_{r-1}(s).
\end{equation*}
Since $q_{r-1}(s)=q_{s-1}(r)$, these are in fact equivalent.

Recalling that 
$Q_{k}^{\C}(n) = n\; (2k-1)!! \; q_{k}(n),$ 
this yields the Harer-Zagier recursion~\cite{Harer86}
\be 
(k+2) Q_{k+1}^{\C}(n) = 2n(2k+1)Q_k^{\C}(n) +k(2k+1)(2k-1)Q_{k-1}^{\C}(n),
\label{eq:HZ_GUE}
\ee
and the recursion in $n$
\be 
n Q_k^{\C}(n+1)=2(k+1)Q_k^{\C}(n)+n Q_k^{\C}(n-1).
\label{eq:rec_GUE_n}
\ee
\end{rmk}
\par
The Meixner-Pollaczek polynomials can be thought as continuous version of the Meixer polynomials~\cite{Atakishiyeva99}:
\be
P_n^{(\lambda)}(x;\;\phi)=e^{-in\phi}\frac{(2\lambda)_n}{n!}M_n(-\lambda+ix;\;2\lambda,e^{-2i\phi}).
\nonumber
\ee
In fact, the normalized moment~\eqref{eq:GUE_2F1} is a Meixner polynomial (see~\eqref{eq:Meixner}) in $n-1$ of degree $k$ or, by symmetry, a Meixner polynomial in $k$ of degree $n-1$.
The alternative form of Theorem~\ref{thm:GUE} using Meixner polynomials is the following.
\par
\begin{thmbis}{thm:GUE}
\begin{align}
Q_{k}^{\C}(n) &=  n\; (2k-1)!! \; M_{k}(n-1;2,-1)=n\; (2k-1)!! \; M_{n-1}(k;2,-1).
\label{eq:GUE_Meixner}
\end{align}
\end{thmbis}
The first polynomials are
\begin{align*}
Q_k^{\C}(1) &=\;(2k-1)!! \\
Q_k^{\C}(2) &=2(2k-1)!!\cdot\left(k+1\right) \\
Q_k^{\C}(3) &=3(2k-1)!!\cdot \frac{1}{3}\left(2k^2+4k+3\right) \\
Q_k^{\C}(4) &=4(2k-1)!!\cdot \frac{1}{3}\left(k^3+3k^2+5k+3\right)
\\
Q_k^{\C}(5) &=5(2k-1)!!\cdot \frac{1}{15}\left(2k^4+8k^3+22k^2+28k+15\right).
\end{align*}

\begin{rmk} The normalised GUE moments can be written as products of moments $(2k-1)!!$ of a standard Gaussian, times a Meixner polynomial
\be
\frac{1}{n}Q_{k}^{\C}(n) =   (2k-1)!! \; M_{k}(n-1;\;2,-1).\nonumber
\ee
It is natural to ask whether the Meixner polynomials form moment sequences of some random variables, so that one can `decompose' the GUE one-point function as multiplicative convolution of a standard Gaussian and another probability distribution (product of two independent random variables). In fact, Ismail and Stanton~\cite{Ismail97} considered the problem of \emph{orthogonal polynomials as moments}. It turns out that the Meixner polynomials are moments of translated Beta random variables
\be
(1-c)^{1-\beta}\frac{\Gamma(\beta)}{\Gamma(x+\beta)\Gamma(-x)}\int_c^1t^k\;(1-t)^{x+\beta-1}(t-c)^{-x-1}dt=c^nM_k(x;\beta,c),
\label{eq:Ismail-Stanton}\nonumber
\ee 
for $\operatorname{Re}(x)<0$ and $\operatorname{Re}(x+\beta)>0$.
Note however that, in our setting, the Meixner polynomials have nonnegative argument $x = n - 1$, 
so that this representation of the one-point function as a `convolution' is purely formal.
\end{rmk}
As a corollary of Theorem~\ref{thm:GUE} we have the following two identities.
\begin{cor} Reflection formula:
\be
\frac{1}{n}\frac{1}{(2k-1)!!}Q_k^{\C}(n)=\frac{1}{k+1}\frac{1}{(2(n-1)-1)!!}Q_{n-1}^{\C}(k+1).
\label{eq:GUE_reflection}
\ee
Convolution formula:
\begin{multline}
\sum_{j=0}^k\frac{(j+1)}{(2j-1)!!}\frac{(k-j+1)}{(2(k-j)-1)!!}\frac{Q_j^{\C}(n)}{n}\frac{Q_{k-j}^{\C}(n)}{n}\\
=\frac{1}{4}\frac{1}{(2(k+2)-1)!!}\left(\frac{Q_{k+2}^{\C}(2n+1)}{2n+1}(k+2-2n)+\frac{Q_{k+2}^{\C}(2n-1)}{2n-1}(k+2+2n)\right).
\label{eq:GUE_convolution}
\end{multline}
\begin{proof}
The reflection formula follows from the hypergeometric representation in~\eqref{eq:GUE_2F1}. To prove~\eqref{eq:GUE_convolution} we start from the convolution property of the Meixner-Pollaczek polynomials~\cite{Al-Salam76}
\be
P^{(\lambda+\mu)}_k(x+y;\;\phi)=\sum_{j=0}^kP^{(\lambda)}_j(x;\;\phi)P^{(\mu)}_{k-j}(y;\;\phi).
\nonumber
\ee
It follows that
\be
\sum_{j=0}^k\frac{(i/2)^j\sqrt{\pi}(j+1)}{\Gamma(j+1/2)}Q_j^{\C}(n)\frac{(i/2)^{k-j}\sqrt{\pi}(k-j+1)}{\Gamma(k-j+1/2)}Q_{k-j}^{\C}(n)=P^{(2)}_k(2in;\;\pi/2).
\nonumber
\ee
To complete the proof we use the Forward Shift Operator for the Meixner-Pollaczek polynomials\cite[Sec. 9, Eq. (9.7.6)]{Koekoek10}, the reflection formula~\eqref{eq:GUE_reflection}, and the recursion~\eqref{eq:rec_GUE_n}.
\end{proof}
\end{cor}

\begin{rmk} The reflection formula~\eqref{eq:GUE_reflection} relates expectations of power of traces when the role of $k$ and $n$ is interchanged. 
We remark that these are not the only quantities invariant under this type of reflection. The other main examples are moments of characteristic polynomials. See the work of Mehta and Normand~\cite[Eq.~(3.15)]{Mehta01} and Forrester and Witte~\cite[Eq.~(4.43)]{Forrester01}. A generalization of such a duality to all $\beta$ was obtained in the work of Desrosiers~\cite{Des}.
\end{rmk}
\subsection{Laguerre unitary ensemble}
The moments of the Laguerre polynomial ensemble (LUE) enjoy a polynomial property, too. They are  (dual) Hahn polynomials~\eqref{eq:Hahn}-\eqref{eq:DualHahn}, or their continuous versions~\eqref{eq:ContinuousDualHahn}-\eqref{eq:ContinuousHahn}.
\label{sec:laguerre}

Let $X_n$ be a LUE random matrix with parameter $m$. Denote $\alpha=m-n\geq0$.
Set $Q_{k}^{\C}(m,n)=\mathbb{E} \operatorname{Tr} X_n^{k}$ for all $k\in\C$ for which the expectation exists.
Then $Q_0^{\C}(m,n)=n$ and, for $k\in\N$, it is known that~\cite{Hanlon92}
\be 
Q_{k}^{\C}(m,n)=\frac1{k} \sum_{i=1}^{k} (-1)^{i-1} \frac{(m-i+1)_k (n-i+1)_k}{(k-i)! (i-1)!} .
\label{eq:moments_LUE}
\ee
For any $k\in\N$, $Q_k^{\C}(m,n)$ is a symmetric polynomial in $m,n$ of degree $k+1$,
with positive integer coefficients:
\begin{align*}
Q_1^{\C}(m,n) &=mn \\
Q_2^{\C}(m,n) &=m^2n+mn^2 \\
Q_3^{\C}(m,n) &= m^3n+3m^2n^2+mn^3+mn\\
Q_4^{\C}(m,n) &= m^4n+6m^3n^2+6m^2n^3+mn^4+5m^2n+5mn^2.
\end{align*}
In fact, for each positive integer $n$, $Q_k^{\C}(m,n)$ is a polynomial in $k$ of degree $2(n-1)$. After some manipulations, the moments~\eqref{eq:moments_LUE} can be expressed in terms of a hypergeometric function:
\be 
\frac{Q_k^{\C}(m,n)}{(k+\alpha)!} = \frac{mn}{(1+\alpha)!} 
\setlength\arraycolsep{1pt}
\; {}_3 F_2\left(\begin{matrix}1-k,2+k,1-n \\2,2+\alpha \end{matrix}; 1\right) .
\label{eq:LUE_hyper}
\ee
This formula can be extended for $k\in\mathbb{C}$ and satisfies $Q_0^{\C}(m,n)=n$. 
\begin{theorem}
\label{thm:LUE} If we write $k=ix-1/2$, then for $\operatorname{Re}(k)>-\alpha-1$,
\be 
Q_k^{\C}(m,n) =  \frac{(k+\alpha)!}{(n-1)!\; (m-1)!} S_{n-1}\left(x^2;\; \frac{3}{2}, \frac{1}{2}, \alpha+\frac{1}{2}\right),
\label{eq:LUE_continuousDualHahn}
\ee
where $S_{n-1}$ denotes the continuous dual Hahn polynomial of degree $n-1$.
In particular this shows that $Q_k^{\C}(m,n) /(k+\alpha)!$ can be extended to a polynomial invariant under the reflection 
$k \to -1-k$ (reciprocity law) and, moreover, its complex zeros lie on the critical line $\operatorname{Re} (k)=-1/2$.
\end{theorem}
\begin{proof} Comparing~\eqref{eq:LUE_hyper} with the hypergeometric representation of continuous dual Hahn polynomials~\eqref{eq:ContinuousDualHahn} we get the result for $k$ integer. 
The extension to complex $k$ is again an application of Carlson's theorem.
\end{proof}
An alternative formulation in terms of Hahn and dual Hahn polynomials~\eqref{eq:Hahn}-\eqref{eq:DualHahn} is as follows.
\par
\begin{thmbis}{thm:LUE}
\begin{align}
Q_k^{\C}(m,n)  &=  mn (2+\alpha)_{k-1} R_{n-1}((k-1)(k+2);\; 1,1,-2-\alpha)\\
&= mn (2+\alpha)_{k-1} Q_{k-1}(n-1;\; 1,1,-2-\alpha) .
\end{align}
\end{thmbis}

\begin{rmk}
The difference equations / three-term recurrence relations for these polynomials (see~\cite[Sections 9.5 and 9.6]{Koekoek10})
yield the Haagerup-Thorbj\o rnsen recursion~\cite{Haagerup03,Cunden16b}
\be 
(k+2) Q_{k+1}^{\C}(m,n)  = (2k+1)(2n+\alpha) Q_k^{\C}(m,n)  + (k-1)(k^2-\alpha^2) Q_{k-1}^{\C}(m,n) .
\label{eq:HT}
\ee
\par
If $k$ is a positive integer, and we treat $\alpha$ as a parameter,
then $Q_k^{\C}(m,n) $ is a polynomial in $n$ of degree $k+1$.  Moreover, we can write
\be 
Q_k^{\C}(m,n)  = \frac{i^{1-k}}{k} n(n+\alpha) p_{k-1}(in;1,1-\alpha,1,1+\alpha),
\label{eq:LUE_continuousHahn}
\ee
where $p_{k-1}$ denotes the continuous Hahn polynomial~\eqref{eq:ContinuousHahn} of degree $k-1$.
\par
Note the alternative formula:
\begin{equation*}
Q_k^{\C}(m,n) = k! \binom{n+k-1}{n-1} \binom{m+k-1}{m-1} 
\setlength\arraycolsep{1pt}
\; {}_3 F_2\left(\begin{matrix}1-k,1-n,1-m \\1-k-n,1-k-m \end{matrix}; 1\right) .
\end{equation*}
\par
If $\alpha$ is fixed and we write $Q_k^{\C}(m,n)=a_n R_n$, 
where $a_n=n(n+\alpha)$, then
\begin{equation*} 
(k-1)(k+2) R_n = a_{n+1} R_{n+1} - (a_{n+1}+a_{n-1}) R_n + a_{n-1} R_{n-1}.
\end{equation*}
\end{rmk}
It is very natural to consider $m$ dependent on $n$. An interesting situation (from the point of view of the large-$n$ limit) is the case $m=cn$ with $c>0$ and fixed. The next result shows that the moments, as functions of $n$, are polynomials with all zeros on a vertical line in the complex plane.
\begin{lem}
\label{lem:Wishart}
Let $m=c n$ and $k$ a positive integer. Then, $n^{-(k+1)}Q^{\C}_k(cn,n)$  is a polynomial in $n^{-2}$ of degree  $\lfloor (k-1)/2 \rfloor$ with positive coefficients.
\end{lem} 

\begin{proof} From~\cite[Corollary 2.4]{Hanlon92} we have
\be
\frac{1}{n^{k+1}}Q^{\C}_k(cn,n)= \sum_{\sigma\in S_k}c^{\#(\sigma)}n^{\#(\sigma)+\#(\gamma_k\sigma^{-1})-(k+1)}.
\label{eq:moments_+ve2}
\ee 
For any permutation $\sigma \in S_{k}$, one has
	\begin{equation*}
	(-1)^{k-\#(\sigma)} = \mathrm{sgn}(\sigma)
	\end{equation*}
	Hence
	\begin{equation*}
	(-1)^{k-\#(\sigma)+k-\#(\gamma_k\sigma^{-1})} = \mathrm{sgn}(\sigma)\mathrm{sgn}(\gamma_k\sigma^{-1}) = \mathrm{sgn}(\gamma_{k}) = (-1)^{k+1}
	\end{equation*}
	and so
	\begin{equation*}
	(-1)^{\#(\sigma)+\#(\gamma_k\sigma^{-1})} = (-1)^{k+1}
	\end{equation*}
	Hence $\#(\sigma)+\#(\gamma_k\sigma^{-1})-(k+1)$ is even and~\eqref{eq:moments_+ve2} is a polynomial in $n^{-2}$.
\end{proof}

\begin{theorem} Fix $c>0$. The zeros of the polynomials $Q_k^{\C}(cn,n)$ as a function of $n$
are purely imaginary and satisfy the interlacing property.  
\begin{proof} Let $q_k(n)=Q_k^{\C}(cn,n)/n$.  Then for each $k$, $q_k(n)$ is a polynomial of
degree $k$, with positive coefficients, and only powers $n^k, n^{k-2},\dots$ (see Lemma~\ref{lem:Wishart}).
It follows that if we define $p_k(x)=i^k q_k(-ix)$, then $p_k(x)$ is a polynomial of degree $k$,
with alternating signs, and satisfies the (Haagerup-Thorbj\o rnsen) recursion 
\be
(k+1)p_k(x) = (c+1)(2k-1)xp_{k-1}(x)-(k-2)((k-1)^2+(c-1)^2x^2)p_{k-2}(x).
\ee  
It now follows from~\cite[Corollary 2.4]{Liu07} that $\{p_k(x)\}$ is a `Sturm sequence' of polynomials. Hence the $p_k$'s have only real zeros, and they satisfy the interlacing property.
\end{proof}
\end{theorem}
\subsection{Jacobi  unitary ensemble}
\label{sub:JUE}
Let $X_n$ be a JUE random matrix of size $n$ with parameters $\alpha_1=m_1-n$, $\alpha_2=m_2-n$.  It turns out that the suitable statistics in this ensemble are \emph{differences of consecutive moments} $\Delta Q_{k}^{\C}(\alpha_1,\alpha_2,n)$, defined as 
\begin{align*}
Q_k^{\C}(\alpha_1,\alpha_2,n)&=\mathbb{E} \operatorname{Tr} X_n^{k}\\
\Delta Q_{k}^{\C}(\alpha_1,\alpha_2,n)&=Q_k^{\C}(\alpha_1,\alpha_2,n)-Q_{k+1}^{\C}(\alpha_1,\alpha_2,n)
\end{align*}
for all $k\in\C$ for which the expectations exist.
\par
\begin{theorem}
\label{thm:JUE}
In terms of Wilson polynomials, writing $k=ix-1/2$, for $\operatorname{Re}(k)>-\alpha_2-1$,
\begin{multline} 
\Delta Q_k^{\C}(\alpha_1,\alpha_2,n)=\frac{(k+\alpha_2)!}{(k+\alpha_1+\alpha_2+2n)!} \frac{(\alpha_1+n)\; (\alpha_1+\alpha_2+n)!}{(n-1)!\; (\alpha_2+n-1)!}
\\
(-1)^{n-1}
W_{n-1}\left(x^2;\; \frac{3}{2}, \frac{1}{2}, \alpha_2+\frac{1}{2}, \frac{1}{2}-\alpha_1-\alpha_2-2n\right).
\end{multline}
This shows that $\Delta Q_k^{\C}(\alpha_1,\alpha_2,n) ((k+\alpha_1+\alpha_2)!/(k+\alpha_1)!)$ can be extended to a  polynomial invariant under the reflection $k \to -1-k$ (reciprocity law) and, moreover, its complex zeros lie on the critical line $\operatorname{Re} (k)=-1/2$.
\end{theorem}
In this case, our strategy is to look for a three-term recursion for $\Delta Q_k^{\C}(\alpha_1,\alpha_2,n)$ when $k$ is an integer. In fact, 
adapting a method due to Ledoux~\cite[Eq.~(30)-(31)]{Ledoux04}, we find the following recurrence relation for the JUE which is the analogue of the Harer-Zagier and Haagerup-Thorbj\o rnsen recursions.
\begin{prop}[Three term recurrence relation for JUE] Let $k\in\Z$. Then,
\begin{align}
R_k\Delta Q_{k+1}^{\C}(\alpha_1,\alpha_2,n)+S_k\Delta Q_{k}^{\C}(\alpha_1,\alpha_2,n)+T_k\Delta Q_{k-1}^{\C}(\alpha_1,\alpha_2,n)=0,
\label{eq:rec_JUE}
\end{align}
with `initial conditions'
\begin{align}
\Delta Q_{0}^{\C}(\alpha_1,\alpha_2,n)&=\frac{n(\alpha_1+n)}{\alpha_1+\alpha_2+2n}\\
\Delta Q_{1}^{\C}(\alpha_1,\alpha_2,n)&=\frac{n(\alpha_1+\alpha_2+n)(\alpha_1+n)(\alpha_2+n)}{(\alpha_1+\alpha_2+2n-1)(\alpha_1+\alpha_2+2n)(\alpha_1+\alpha_2+2n+1)}.
\end{align}
The coefficient $R_k$, $S_k$, and  $T_k$ are given by
\begin{align*}
R_k(\alpha_1,\alpha_2,n)&=(k+2)((\alpha_1+\alpha_2+2n)^2-(k+1)^2),\\
S_k(\alpha_1,\alpha_2,n)&=-(2k+1)(2n(\alpha_1+\alpha_2+n)+\alpha_2^2+\alpha_1\alpha_2-k(k+1)),\\
T_k(\alpha_1,\alpha_2,n)&=(k-1)(\alpha_2^2-k^2).
\end{align*}
\end{prop}
\begin{proof}[Proof of Theorem~\ref{thm:JUE}] The proof is immediate when $k$ is a nonnegative integer by observing that~\eqref{eq:rec_JUE} is the discrete S-L problem for Wilson polynomials, and by checking the initial conditions. For $k$ complex we can use the same method of Theorem~\ref{thm:GUE}.
\end{proof}
\begin{rmk}
Using the hypergeometric representation of Wilson polynomials~\eqref{eq:Wilson} we have the explicit formula
\begin{multline*}
\Delta Q_k^{\C}(\alpha_1,\alpha_2,n) = \Delta Q_1^{\C}(\alpha_1,\alpha_2,n) \frac{(k+\alpha_2)!}{(1+\alpha_2)!} \frac{(1+\alpha_1+\alpha_2+2n)!}{(k+\alpha_1+\alpha_2+2n)!}\\
\times\setlength\arraycolsep{1pt}
\; {}_4 F_3\left(\begin{matrix}1-k,2+k,1-n,1-n-\alpha_1 \\2,2+\alpha_2,2-\alpha_1-\alpha_2-2n \end{matrix}; 1\right) .
\label{eq:Jacobi4F3}
\end{multline*}
To our knowledge, this hypergeometric representation is new.
\end{rmk}

\begin{rmk}
Ledoux~\cite{Ledoux04} obtained a fourth order recursion for moment differences of the Jacobi ensemble, but the ensemble he considers is shifted compared to ours. In our notation, the moments $L(k)$ considered in~\cite{Ledoux04} can be written as
$$
L(k)= \int_0^1 (2x-1)^k  \rho^{(2)}_{n}(x) dx = \sum_{j=0}^k \binom{k}{j} 2^j (-1)^{k-j}\int_{0}^{1}x^j\rho^{(2)}_{n}(x) 
$$ 
for which it is shown that $L(k)-L(k+2)$ satisfies a fourth order recursion. Using $(2x-1)^k - (2x-1)^{k+2} = 4(2x-1)^{k}(x-x^{2})$ we obtain
$$
L(k)-L(k+2) = \sum_{j=0}^{k}  \binom{k}{j}  2^{j+2} (-1)^{k-j} \Delta Q_{j+1}^{\C}(\alpha_1,\alpha_2,n).
$$
It follows that Ledoux's moment differences can be expressed as a linear combination of hypergeometric functions.
\end{rmk}

\par 
The difference of moments $\Delta Q^{\C}_k(\alpha_1,\alpha_2,n)$ can alternatively be written in terms of Racah polynomials.
\par
\begin{thmbis}{thm:JUE}
The JUE difference of moments is ($\alpha_1+\alpha_2+n\notin\Z$)
\begin{multline} 
\Delta Q_k^{\C}(\alpha_1,\alpha_2,n)=(-1)^{n-1}n(n+\alpha_1)(n+\alpha_2)(n+\alpha_1+\alpha_2)\\
\small\times\frac{\sin\left(\pi\left(\alpha_1+\alpha_2+2n-1\right)\right)}{\sin\left(\pi\left(\alpha_1+\alpha_2+n\right)\right)}
\frac{\left(2+\alpha_2\right)_{k-1}}{\left(\alpha_1+\alpha_2+2n-1\right)_{k+2}}\\
\small\times R_{n-1}\left(\left(k-1\right)\left(k+2\right);\;1,-\alpha_1-2n,1-\alpha_1-\alpha_2-2n,1+\alpha_1+\alpha_2+2n\right).
\end{multline}
\end{thmbis}
\par
Using the recurrence relation~\eqref{eq:rec_JUE}, one can verify the following identity between positive and negative moments (this is the analogue of~\eqref{eq:dualityLUE}).
\begin{prop}[Reciprocity law for JUE]
\be
\Delta Q_{-(k+1)}^{\C}(\alpha_1,\alpha_2,n)=\left(\prod_{j=-k}^k\frac{\alpha_1+\alpha_2+2n-j}{\alpha_2-j}\right)\Delta Q_{k}^{\C}(\alpha_1,\alpha_2,n).
\ee
\end{prop}
For instance, ($\alpha_1=m_1-n$, $\alpha_2=m_2-n$):
\begin{align*}
\Delta Q_{0}^{\C}(\alpha_1,\alpha_2,n)&=\frac{nm_1}{\alpha_1+\alpha_2+2n}\\
\Delta Q_{-1}^{\C}(\alpha_1,\alpha_2,n)&=\frac{nm_1}{\alpha_2}\\
\Delta Q_{1}^{\C}(\alpha_1,\alpha_2,n)&=\frac{nm_1m_2(m_1+m_2-n)}{(\alpha_1+\alpha_2+2n)\left((\alpha_1+\alpha_2+2n)^2-1\right)}\\
\Delta Q_{-2}^{\C}(\alpha_1,\alpha_2,n)&=\frac{nm_1m_2(m_1+m_2-n)}{\alpha_2\left(\alpha_2^2-1\right)}\\
\Delta Q_{2}^{\C}(\alpha_1,\alpha_2,n)&=\frac{nm_1m_2(m_1+m_2-n)\left(m_2(m_1+m_2-n)+nm_1-2\right)}{(\alpha_1+\alpha_2+2n)\left((\alpha_1+\alpha_2+2n)^2-1\right)\left((\alpha_1+\alpha_2+2n)^2-4\right)}\\
\Delta Q_{-3}^{\C}(\alpha_1,\alpha_2,n)&=\frac{nm_1m_2(m_1+m_2-n)\left(m_2(m_1+m_2-n)+nm_1-2\right)}{\alpha_2\left(\alpha_2^2-1\right)\left(\alpha_2^2-4\right)}\\
\end{align*}

\begin{table}
\begin{tabular}{ ccc }
 \hline
{\bf Matrix ensembles} & {\bf Correlation functions } & {\bf Moments } \\
&(classical OP's) &(hypergeometric OP's) \\
\hline\hline
 GUE & Hermite& Meixner-Pollaczek\\
  LUE & Laguerre& continuous dual Hahn\\
   JUE & Jacobi& Wilson\\
 \hline
\end{tabular}
\caption{Relation between the correlation functions (in terms of classical OP's) and the moments (given by hypergeometric OP's) of the classical unitary ensembles.}
\label{tab:sum}
\end{table}

\subsection{Generating functions}
\label{sub:GF}
It is sometimes convenient to define the moments of random matrices in terms of their generating function.
The first example of such a generating function was constructed by Harer and Zagier for the GUE of fixed size $n$ (see Eq.~\eqref{eq:GF_GUEn} below). This convergent series is a rational function. As emphasised by Morozov and Shakirov~\cite{Morozov09}, from the point of view of random matrices and enumeration problems, this is a highly non-trivial result: a generating function for moments at all genera appears to be rational. The generating function of covariances of the GUE computed in~\cite{Morozov09} turns out to be again an elementary function. The generating function of higher order cumulants of the GUE have been studied recently by Dubrovin and Yang~\cite{Dubrovin16} who expressed them in terms of traces of $2\times2$ matrix-valued series. 
\par
One of the advantages of the representation of the moments in terms of hypergeometric OP's discussed in the present work, is that  we can write explicit formulae for the generating functions of the moments of GUE and LUE for fixed $n$ \emph{and/or} $k$. Remarkably, these closed expressions are elementary functions.
\begin{prop}
\label{prop:GF} Let $Q^{\C}_k(n)$ and $Q^{\C}_k(m,n)$ be the moments of the GUE and LUE, respectively. Then
\begin{itemize}
\item[GUE:]
\be
\sum_{k=0}^{\infty}\frac{Q^{\C}_k(n)}{(2k-1)!!}\,t^k=\frac{1}{2t}\left(\left(\frac{1+t}{1-t}\right)^n-1\right)
\label{eq:GF_GUEn}
\ee
\be
\sum_{n=1}^{\infty}\frac{Q^{\C}_k(n)}{(2k-1)!!}\,z^n=\frac{z}{1-z^2}\left(\frac{1+z}{1-z}\right)^{k+1}
\label{eq:GF_GUEk}
\ee
\be
\sum_{n\geq1,k\geq0}\frac{Q^{\C}_k(n)}{(2k-1)!!}\,t^k z^n=\frac{z}{1-z}\frac{1}{(1-t)-z(1+t)}.
\label{eq:GF_GUEnk}
\ee
\item[LUE:]
\be
\sum_{k=1}^{\infty}\frac{Q^{\C}_k(m,n)}{(k-1)!}\,\frac{t^k}{k!}=te^t
L_{n-1}^{(1)}(-t)L_{m-1}^{(1)}(-t).
\label{eq:GF_LUEn}
\ee
\be
\sum_{n=1}^{\infty}Q^{\C}_k(m,n)\frac{(m-1)!}{(k+\alpha)!}\,\frac{z^n}{n!}=\frac{z}{(1-z)^{\alpha+k+1}}
\setlength\arraycolsep{1pt}
\; {}_2 F_1\left(\begin{matrix}-k,1-k\\2\end{matrix};z\right)
\label{eq:GF_LUEk}
\ee
\be
\sum_{n,m,k\geq1}\frac{Q^{\C}_k(m,n)}{(k-1)!}\,\frac{z_1^nz_2^mt^k}{k!}=\frac{z_1z_2t}{(1-z_1)^2(1-z_2)^2}\,
\exp\left(\frac{3z_1z_2-2(z_1+z_2)+1}{(1-z_1)(1-z_2)}t\right).
\label{eq:GF_LUEmnk}
\ee
\end{itemize}
\begin{proof}
The sum~\eqref{eq:GF_GUEn} can be computed from~\eqref{eq:GUE_Meixner} using the formula of the generating function of Meixner polynomials~\cite[Eq. (9.10.11)]{Koekoek10}:
\be
\sum_{k=1}^{\infty}\frac{k}{n}\frac{Q^{\C}_k(n)}{(2k-1)!!}t^{k-1}=\left(\frac{1+t}{1-t}\right)^{n-1}.\nonumber
\ee
The formula for generating function~\eqref{eq:GF_GUEn} follows from the identity
\be
\frac{n}{t}\frac{d}{dt}\sum_{k=0}^{\infty}\frac{Q^{\C}_k(n)}{(2k-1)!!}\,t^k=\sum_{k=1}^{\infty}\frac{k}{n}\frac{Q^{\C}_k(n)}{(2k-1)!!}t^{k-1}.\nonumber
\ee
\par
The generating function~\eqref{eq:GF_GUEk} for fixed $k$, is a direct consequence of the representation of the moments in terms of Meixner-Pollaczek polynomials~\eqref{eq:GUE_Meixner-Pollaczek} and their generating function~\cite[Eq. (9.7.11)]{Koekoek10}.
\par 
Finally, the joint generating function~\eqref{eq:GF_GUEnk} is the resummation in $n$ of~\eqref{eq:GF_GUEn} (or the resummation in $k$ of~\eqref{eq:GF_GUEk}).
\par
For the LUE we use the generating series of continuous (dual) Hahn polynomials. 
From the representation of the LUE moments as continuous Hahn~\eqref{eq:LUE_continuousHahn}, using~\cite[Eq. (9.5.11)]{Koekoek10} we get
\begin{align}
\sum_{k=1}^{\infty}\frac{Q^{\C}_k(m,n)}{nm\,(k-1)!}\,\frac{t^{k-1}}{k!}&=
\setlength\arraycolsep{1pt}
 {}_1 F_1\left(\begin{matrix}1-n\\2\end{matrix};-t\right)
\setlength\arraycolsep{1pt}
 {}_1 F_1\left(\begin{matrix}1+m\\2\end{matrix};t\right)\nonumber\\
&=
\setlength\arraycolsep{1pt}
 {}_1 F_1\left(\begin{matrix}1-n\\2\end{matrix};-t\right)
\setlength\arraycolsep{1pt}
 {}_1 F_1\left(\begin{matrix}1-m\\2\end{matrix};-t\right)e^t.\nonumber
\end{align}
Note that the hypergeometric functions on the right-hand side are terminating series. In fact, they are  Laguerre polynomials~\cite[Eq. (9.12.1)]{Koekoek10}, thus proving~\eqref{eq:GF_LUEn}. 
\par
For~\eqref{eq:GF_LUEk} we use the representation in terms of continuous dual Hahn polynomials~\eqref{eq:LUE_continuousDualHahn} and  the formula of the generating series~\cite[Eq. (9.3.11)]{Koekoek10}. We have
\be
\sum_{n=1}^{\infty}Q^{\C}_k(m,n)\frac{(m-1)!}{(k+\alpha)!}\,\frac{z^n}{n!}=\frac{z}{(1-z)^{\alpha-k}}
\setlength\arraycolsep{1pt}
\; {}_2 F_1\left(\begin{matrix}k+2,k+1\\2\end{matrix};z\right),\nonumber
\ee
which is equal to~\eqref{eq:GF_LUEk} by Euler's tranformation. Note that the hypergeometric series is terminating. In fact, one could also write it in terms of Jacobi polynomials.  
\par
The joint generating series in~\eqref{eq:GF_LUEmnk} is a resummation of~\eqref{eq:GF_LUEn} 
over $n$ and $m$ using the known formula for the generating function of Laguerre polynomials~\cite[Eq. (9.12.10)]{Koekoek10}.
\end{proof}
\end{prop}
\begin{rmk}
The series~\eqref{eq:GF_GUEn} was computed by Harer and Zagier~\cite{Harer86} using different methods. 
It is surprising that, although the three-term recurrence in $k$ and the generating function were known, nobody recognized the moments of the GUE as Meixner polynomials. The generating function of the GUE for fixed $k$ (Eq.~\eqref{eq:GF_GUEk}) does not seem to appear in the previous literature. The joint series~\eqref{eq:GF_GUEnk} appears in the work of Morozov and Shakirov~\cite{Morozov09} who stressed the nontrivial fact that it is a rational function in both variables.
\par
The generating functions~\eqref{eq:GF_LUEn}-\eqref{eq:GF_LUEk}-\eqref{eq:GF_LUEmnk} for the LUE seem to be new. It is remarkable, again, that these series sum to elementary functions.
\end{rmk}
\section{Large-$n$ asymptotics of the spectral zeta functions}
\label{sec:asymptotics}
It is a classical result that, after rescaling, the one-point function $\rho_n^{(\beta)}(x)$ of the random matrix ensembles considered in this paper weakly converges to a compactly supported probability measure, as $n$ goes to infinity. The limit $\rho_{\infty}(x)$ is known as \emph{equilibrium measure} of the ensemble and does not depend on the Dyson index $\beta$. In formulae, for all $k\in\N$,
\be
\lim_{n\to\infty}\frac{1}{n^{k+1}}\E\Tr X_n^{k}=\int x^{k}\rho_{\infty}(x)d x.
\nonumber
\ee
This suggests to define the limit zeta function $\zeta_{\infty}(s)$ as the analytic continuation of $\int|x|^{-s}\rho_{\infty}(x)d x$. The limit zeta functions for the classical ensembles turn out to be meromorphic functions. For the LUE and JUE, $\zeta_{\infty}(s)$ has infinitely many nontrivial zeros, and they all lie on a critical line.

We discuss the three classical ensembles separately. For notational convenience we consider the GUE, LUE and JUE, although the results hold true for any $\beta$-ensemble. 

\subsection{Gaussian ensemble} The equilibrium measure is given by the semicircular law
\be
\rho_{\infty}(x)=\frac{1}{2\pi}\sqrt{4-x^2}\;1_{x\in(-2,2)}.
\ee
After a suitable rescaling, in the large-$n$ limit, the integer moments converge to the Catalan numbers:
\be
\lim_{n\to\infty}\frac{1}{n^{k+1}}Q^{\C}_k(n)=\int x^{2k}\rho_{\infty}(x)d x=\frac{1}{k+1}\binom{2k}{k}.
\ee
This formula can be analytically continued and suggests to define the limit GUE zeta function as the analytic continuation of $\int|x|^{-s}\rho_{\infty}(x)d x$:
\be
\zeta_{\infty}(s)=\frac{2^{-s}\;\Gamma\left(\frac{1-s}{2}\right)}{\sqrt{\pi}\;\Gamma\left(2-\frac{s}{2}\right)}.
\ee
This function has alternating simple poles and zeros on the positive integers, with no other zeros in the rest of the complex plane. The large-$n$ limit is more interesting for matrices in Laguerre and Jacobi ensembles. 

\subsection{Laguerre ensemble} Let $X_n$ be in the Laguerre unitary ensemble. Set $\alpha=m-n=(c-1)n$, with $c\geq1$. Define the  equilibrium measure
\be 
\rho_{\infty}(x)=\frac{1}{2\pi x}\sqrt{(x_+-x)(x-x_-)}\;1_{x\in(x_-,x_+)}
\label{eq:equi_Laguerre}
\ee
where $x_{\pm}=(1\pm\sqrt{c})^2\geq 0$ (this is the celebrated Marchenko-Pastur distribution). 
Then, 
\be
\lim_{n\to\infty}\frac{1}{n^{k+1}}Q_k^{\C}((c-1)n,n)=\int_{\R}x^k\rho_{\infty}(x)d x,
\ee
Define the limit LUE zeta function as 
\be
\zeta_{\infty}(s)=\int_{\R}|x|^{-s}\rho_{\infty}(x)\de x
\label{eq:LUEinfty}
\ee
with $\rho_{\infty}(x)$ in~\eqref{eq:equi_Laguerre}. From the finite-$n$ functional equation $\xi_{n}(s)=\xi_{n}(1-s)$, we see that a good definition for the limit $n\to\infty$ is 
\begin{equation*}
\xi_{\infty}(s)=(x_-x_+)^{s/2}\zeta_{\infty}(s).
\end{equation*}
\begin{prop}\label{prop:LUEthermo} Assume $c>1$. Then, the functional equation $\xi_{\infty}(s)=\xi_{\infty}(1-s)$ holds, and the zeros of the $\zeta_{\infty}(s)$ all lie on the critical line $\operatorname{Re}(s)=1/2$.
\end{prop}
\begin{proof} To prove the functional equation it suffices to use the change of coordinates $y=(x_+x_-)/x$ in the integral~\eqref{eq:LUEinfty}, and notice that $x_-x_+=(c-1)^2$.  Alternatively, a calculation using Euler's integral formula and Pfaff transformation, reveals that
\begin{align}
\zeta_{\infty}(s)&=\frac{1}{16}\frac{(x_+-x_-)^2}{(x_-)^{s+1}}\;  {}_2 F_1\left(\begin{matrix}3/2,s+1 \\3 \end{matrix}; \frac{x_--x_+}{x_-}\right)\nonumber\\
&=c(1-c)^{-s-1}\; {}_2 F_1\left(\begin{matrix}\frac{3}{4}+\frac{1}{2}\left(s-\frac{1}{2}\right),\frac{3}{4}-\frac{1}{2}\left(s-\frac{1}{2}\right) \\2 \end{matrix}; -\frac{4c}{(1-c)^2}\right) .
\label{eq:Hyper_MP}
\end{align}
Then apply Euler's transformation formula to show the functional equation. 

To show that the zeros of $\xi_{\infty}(s)$ are on the critical line we use an argument based on Sturm-Liouville theory (we borrowed this argument from a similar problem in a paper by Biane~\cite{Biane09}). 

First, observe that the integral~\eqref{eq:LUEinfty} can not vanish for $s$ real. We want to show that the zeros of 
\[
{}_2 F_1\left(\begin{matrix}\frac{3}{4}+\mu,\frac{3}{4}-\mu \\2 \end{matrix}; -\frac{4c}{(1-c)^2}\right),
\]
where $2\mu=s-1/2$, lie on the imaginary axis.
The function $y(z)={}_2 F_1(3/4+\mu,3/4-\mu; 2; -z)$ satisfies the hypergeometric equation
\be
z(1+z)y'' + (2+5z/2) y' + (9/16-\mu^2) y = 0.
\nonumber
\ee
Thus, if $\mu$ is such that $y(4c/(1-c)^2)=0$, then $y(x)$ is a solution to the Sturm-Liouville problem 
\be
( p(x) y'(x) )'+ q(x) y(x) = \mu^2 w(x) y(x)
\nonumber
\ee
on $[4c/(1-c)^2,\infty)$, with Dirichlet boundary conditions,
where $p(x)=x^2(1+x)^{1/2}$, $q(x) = (9/16) w(x)$ and $w(x)=x(1+x)^{-1/2}$.
It then follows from the Sturm-Liouville theory that the eigenvalues $\mu^2$ are real  which can only happen if $\mu$ is real or purely imaginary. Since we have excluded the real case, we conclude that $2\mu=s-1/2$ is purely imaginary: the zeros of $\zeta_{\infty}(s)$ all lie on the critical line $\operatorname{Re}(s)=1/2$.
\end{proof}

\subsection{Jacobi ensemble} Let $X_n$ be in the JUE.
Set $\alpha_{1,2}=(c_{1,2}-1)n$, with $c_{1,2}\geq1$. Then, the  equilibrium measure is
\be 
\rho_{\infty}(x)=\frac{c_1+c_2}{2\pi x(1-x)}\sqrt{(x_+-x)(x-x_-)}\;1_{x\in(x_-,x_+)}
\label{eq:equi_Jacobi}
\ee
where $x_{\pm}=\left(\frac{\sqrt{c_1}\pm\sqrt{c_2(c_1+c_2-1)}}{c_1+c_2}\right)^2$, and
\be
\lim_{n\to\infty}\frac{1}{n^{k+1}}Q_k^{\C}((c_1-1)n,(c_2-1)n,n)=\int_{\R}x^k(1-x)\rho_{\infty}(x)\de x
\ee
Define the limit JUE zeta function as 
\be
\zeta_{\infty}(s)=\int_{\R}|x|^{-s}\rho_{\infty}(x)\de x
\label{eq:JUEinfty}
\ee
with $\rho_{\infty}(x)$ in~\eqref{eq:equi_Jacobi}. Again, the finite-$n$ functional equation $\xi_{n}(s)=\xi_{n}(1-s)$, suggests the definition of 
\begin{equation*}
\xi_{\infty}(s)=(x_-x_+)^{s/2}\left(\zeta_{\infty}(s)-\zeta_{\infty}(s-1)\right).
\end{equation*}
\begin{prop}\label{prop:JUEthermo} Assume $c_{1,2}>1$. Then, the functional equation $\xi_{\infty}(s)=\xi_{\infty}(1-s)$ holds, and the complex zeros of $\xi_{\infty}(s)$ all lie on the critical line $\operatorname{Re}(s)=1/2$.
\end{prop}
\begin{proof} Using Euler's integral formula, we have
\be
\zeta_{\infty}(s)=\frac{c_1+c_2}{16}\frac{(x_+-x_-)^2}{(x_-)^{s+1}}\;  {}_2 F_1\left(\begin{matrix}3/2,s+1 \\3 \end{matrix}; \frac{x_--x_+}{x_-}\right).
\label{eq:Hyper_JUE}
\ee
The proof of Proposition~\ref{prop:LUEthermo}  is easily adapted.
\end{proof}

\section{Beyond random matrices: Wronskians of orthogonal polynomials}
\label{sec:extension}

\subsection{Mellin transform of orthogonal polynomials}
Bump and Ng~\cite{Bump86} and Bump, Choi, Kurlberg and Vaaler~\cite{Bump00} made the remarkable discovery that the Mellin transforms of Hermite and Laguerre functions have zeros on the critical line $\operatorname{Re}(s)=1/2$.  Their proof is based on the observation that the Mellin transform preserves orthogonality. Hence, Mellin transforms of orthogonal polynomials are themselves orthogonal with respect to some inner product. Later~\cite{Coffey07,Coffey15} it was noticed that those orthogonal functions are hypergeometric OP's (multiplied by some nonnegative integrable weight).
\par
Consider, for concreteness, the Hermite polynomials $H_n(x)$ and the normalised Hermite wavefunctions $\phi_n(x)=(2^nn!\sqrt{\pi})^{-1/2}e^{-x^2/2}H_n(x)$.
The following proposition follows from the result of Bump \textit{et al.}~\cite{Bump86,Bump00}.
\begin{prop}[Mellin transform of Hermite functions] 
\label{prop:Bump}
For all integers $n\geq 0$
\begin{align*}
\phi_{2n}^*(s)&=i^n2^{n-1+\frac{s}{2}}\frac{n!}{\sqrt{(2n)!\sqrt{\pi}}}\Gamma\left(\frac{s}{2}\right)P^{\left(\frac{1}{4}\right)}_n\left(-i\frac{(s-1/2)}{2};\;\frac{\pi}{2}\right)\\
\phi_{2n+1}^*(s)&=i^n2^{n+\frac{s}{2}}\frac{n!}{\sqrt{(2n+1)!\sqrt{\pi}}}\Gamma\left(\frac{s+1}{2}\right)P^{\left(\frac{3}{4}\right)}_n\left(-i\frac{(s-1/2)}{2};\;\frac{\pi}{2}\right).
\end{align*}
\end{prop}
\par
The averaged spectral zeta function of unitary invariant ensembles of random matrices can be interpreted as
Mellin transform of Wronskians of adjacent Hermite, Laguerre or Jacobi wavefunctions. Given our results,
it is natural to ask whether more general Wronskians have the property that their Mellin transforms can be written in terms of hypergeometric OP's.
\subsection{Mellin transforms of products and Wronskians of  classical orthogonal polynomials}
\label{sub:extension}
In this section we will use repeatedly the properties~\eqref{eq:Mellin_d}-\eqref{eq:Mellin_m}-\eqref{eq:Mellin_md} of the Mellin transform.
\par
The Wronskian of smooth functions $f_1(x)\dots, f_m(x)$ is defined as
\be
\operatorname{Wr}(f_1(x),\dots,f_m(x))=\det\left(f^{(i-1)}_j(x)\right)_{i,j=1}^m.
\nonumber
\ee
Note the homogeneity property
\be
\operatorname{Wr}(g(x)f_1(x),\dots,g(x)f_m(x))=(g(x))^m\operatorname{Wr}(f_1(x),\dots,f_m(x))
\nonumber
\ee
\par
Consider the Hermite, Laguerre and Jacobi polynomials defined in~\eqref{eq:Hermite}-\eqref{eq:Laguerre}-\eqref{eq:Jacobi}, and the associated normalised wavefunctions
\be
\phi_{n}(x)=
\left\{  \begin{array}{l@{\quad}cr} 
\displaystyle (2^nn!\sqrt{\pi})^{-1/2}e^{-x^2/2}H_n(x)
&\text{(Hermite)}  \\ \\
\displaystyle \sqrt{\frac{n!}{\Gamma(n+\alpha+1)}} x^{\alpha/2}e^{-x/2} L_n^{(\alpha)}(x)\,\chi_{\R_{+}}(x)
&\text{(Laguerre)} \\ \\
\displaystyle \sqrt{\frac{n!(\alpha_1+\alpha_2+2n+1)\Gamma(\alpha_1+\alpha_2+n+1)}{\Gamma(\alpha_1+n+1)\Gamma(\alpha_2+n+1)}}&\text{}\\ 
\,\small\times\, (1-x)^{\alpha_1/2}x^{\alpha_2/2}P_n^{(\alpha_1,\alpha_2)}(x)\,\chi_{[0,1]}(x)
&\text{(Jacobi)}.
\end{array}\right.
\label{eq:wavefunctions}
\ee
The wavefunctions are orthonormal 
\begin{equation*}
 \int_I\phi_n(x)\phi_m(x)dx=\delta_{n,m},
\end{equation*}
where $I=\R,\R_+$, and $[0,1]$ for Hermite, Laguerre, and Jacobi, respectively.
We set 
\begin{align*}
\omega_{n,\ell}(x)&=\phi_n(x)\phi_{n+\ell}(x)\\
W_{n,\ell}(x)&=\operatorname{Wr}(\phi_n(x),\phi_{n+\ell}(x)).
\end{align*}
\begin{theorem} 
\label{thm:Wronsk_GUE}
Let $\phi_n(x)$ be Hermite wavefunctions~\eqref{eq:wavefunctions}. Then,
\begin{itemize}
\item[i)] the Mellin transform of the products is
\be
\omega_{n,\ell}^*(s) = i^n2^{\frac{\ell}{2}-s}\sqrt{\frac{n!}{(n+\ell)!}}\frac{\Gamma(s)}{\Gamma\left(\frac{s-\ell+1}{2}\right)}\,P_n^{\left(\frac{\ell+1}{2}\right)}\left(-\frac{is}{2};\frac{\pi}{2}\right);
\label{eq:Mproduct_Hermite}
\ee
\item[ii)] the Mellin transform of the Wronskians is
\be
W^*_{n,\ell}(s-1) =\frac{2\ell}{s-1}\omega^*_{n,\ell}(s).
\label{eq:Mwronskian_Hermite}
\ee
\end{itemize}
\end{theorem}
\begin{proof} Part i): Given the three-term recurrence of the Meixner-Pollaczek polynomials\cite[Sec. 9, Eq. (9.7.3)]{Koekoek10}, it is sufficient to show that $\omega_{n,\ell}^*(s) $ satisfies the recurrence
\be
\sqrt{(n+1)(n+\ell+1)}\omega_{n+1,\ell}^*(s)-s\omega_{n,\ell}^*(s) -\sqrt{n(n+\ell)}\omega_{n-1,\ell}^*(s)=0.
\label{eq:rec_prodH}
\ee
Using the three-term recurrence of the Hermite functions 
\be
\sqrt{n+1}\phi_{n+1}(x)-\sqrt{2}x\phi_n(x)+\sqrt{n}\phi_{n-1}(x)=0
\ee 
we have
\begin{multline*}
\sqrt{(n+1)(n+\ell+1)}\omega_{n+1,\ell}(x)=2x^2\omega_{n,\ell}(x)+\sqrt{n(n+\ell)}\omega_{n-1,\ell}(x)\\
-\sqrt{2}x\left(\sqrt{n+\ell}\omega_{n,\ell-1}(x)
+\sqrt{n}\omega_{n-1,\ell+1}(x)\right).
\end{multline*}
We now take the Mellin transform of both sides. Using the equation 
\be
\phi'_n(x)+x\phi_n(x)-\sqrt{2n}\phi_{n-1}(x)=0,
\label{eq:diff_rec_Hermite}
\ee 
we get
\begin{equation*}
\sqrt{(n+1)(n+\ell+1)}\omega_{n+1,\ell}^*(s)=-\MM\left[x \omega_{n,\ell}'(x);s\right]+\sqrt{n(n+\ell)}\omega_{n-1,\ell}^*(s).
\end{equation*}
The fact that $\MM\left[x\omega'_{n,\ell}(x);s\right] = -s\omega_{n,\ell}^*(s)$ follows from integration by parts (or property \eqref{eq:Mellin_md} of the Mellin transform, with $m=1$). We have proved~\eqref{eq:rec_prodH}. To complete the proof of~\eqref{eq:Mproduct_Hermite} we compute the initial conditions
\begin{equation*}
\omega_{0,\ell}^*(s)=\frac{2^{\frac{\ell}{2}-s}}{\sqrt{\ell!}}\frac{\Gamma(s)}{\Gamma\left(\frac{s-\ell+1}{2}\right)}, \quad\text{and}\quad \omega_{1,\ell}^*(s)=\frac{2^{\frac{\ell}{2}-s}}{\sqrt{(\ell+1)!}}\frac{\Gamma(s)}{\Gamma\left(\frac{s-\ell+1}{2}\right)}\,s.
\end{equation*}
\par
Part ii): Note that 
\begin{equation*}
\label{Wr2}
W_{n,\ell}(x)= \sqrt{2(n+\ell)} \omega_{n,\ell-1}(x) -\sqrt{2n}\omega_{n-1,\ell+1}(x) 
\end{equation*} 
where we have used the identity~\eqref{eq:diff_rec_Hermite}. 
Taking the Mellin transform of both sides, 
substituting the identity~\eqref{eq:Mproduct_Hermite},  and using  the Forward Shift Operator for the  Meixner-Pollaczek polynomials\cite[Sec. 9, Eq. (9.7.6)]{Koekoek10} we complete the proof of~\eqref{eq:Mwronskian_Hermite}.
\end{proof}
\par
The following analogue of Proposition~\ref{prop:Bump} for the Laguerre wavefunction is essentially  due to Bump \textit{et al.}~\cite{Bump00} and Coffey~\cite{Coffey07}.
\begin{prop}
\label{prop:Coffey}
Let $\phi_n(x)$ be Laguerre wavefunctions~\eqref{eq:wavefunctions}. Then
\begin{equation*}
\phi_n^*(s)=(-i)^{n}2^{s+\alpha/2}\sqrt{\frac{n!}{\Gamma(n+\alpha+1)}}\Gamma(s+\alpha/2)P_{n}^{(\frac{1+\alpha}{2})}\left(\frac{1}{i}(s-1/2);\;\frac{\pi}{2}\right).
\end{equation*}
\end{prop}
\begin{theorem} 
\label{thm:Wronsk_LUE}
Let $\phi_n(x)$ be Laguerre wavefunctions. Then,
\begin{itemize}
\item[i)] the Mellin transform of the products is
\begin{multline}
\omega_{n,\ell}^*(s) = (-1)^{\ell}\frac{\Gamma(s)\Gamma(s+\alpha)}{\Gamma(s-\ell)\sqrt{n!(n+\ell)!\Gamma(n+\alpha+1)\Gamma(n+\ell+\alpha+1)}}\\
\small\times S_n\left(-\left(s-1/2\right)^2;\ell+1/2,1/2,\alpha+1/2\right);
\label{eq:Mproduct_Laguerre}
\end{multline}
\item[ii)] the Mellin transform of the Wronskians is
\be
W^*_{n,\ell}(s-1) =\frac{\ell}{s-1}\omega_{n,\ell}^*(s).
\label{eq:Mwronskian_Laguerre}
\ee
\end{itemize}
\end{theorem}
\begin{proof} Given the three-term recurrence of the continuous dual Hahn polynomials\cite[Sec. 9, Eq. (9.3.4)]{Koekoek10}, we need to show that $\omega_{n,\ell}^*(s) $ satisfies 
\begin{multline}
\sqrt{A_n(n+1)(n+\alpha+1)}\,\omega_{n+1,\ell}^*(s)
=(A_n+C_n-\ell(\ell+1)+s(s-1))\omega_{n,\ell}^*(s)\\
 -\sqrt{C_n(n+\ell)(n+\ell+\alpha)}\omega_{n-1,\ell}^*(s),
\label{eq:rec_prodL}
\end{multline}
with $A_n=(n+\ell+1)(n+\ell+\alpha+1)$ and $C_n=n(n+\alpha)$.
Using the relations
\begin{align}
\sqrt{(n+1)(n+\alpha+1)}\phi_{n+1}(x)&=(2n+\alpha+1-x)\phi_n(x)-\sqrt{n(n+\alpha)}\phi_{n-1}(x),
\label{eq:Laguerre_rec1}\\
x\phi'_n(x)&=\frac{1}{2}(2n+\alpha-x)\phi_n(x)-\sqrt{n(n+\alpha)}\phi_{n-1}(x),
\label{eq:Laguerre_rec2}
\end{align}
we have the identity
\begin{multline*}
\sqrt{A_n(n+1)(n+\alpha+1)}\,\omega_{n+1,\ell}(x)=(A_n+C_n-\ell(\ell+1))\omega_{n,\ell}(x)+\left(x^2\omega_{n,\ell}'(x)\right)'\\
 -\sqrt{C_n(n+\ell)(n+\ell+\alpha)}\omega_{n-1,\ell}(x).
\end{multline*}
Taking the Mellin transform of both sides and using~\eqref{eq:Mellin_d}-\eqref{eq:Mellin_md} we get~\eqref{eq:rec_prodL}. The initial conditions are
\begin{align*}
\omega_{0,\ell}^*(s)&=(-1)^\ell \frac{\Gamma (s ) \Gamma (s+\alpha )}{\Gamma (s-\ell) \sqrt{\Gamma (\ell+1)\Gamma (\alpha+1)\Gamma (\ell+\alpha+1)}}, \\
 \omega_{1,\ell}^*(s)&=(-1)^\ell \frac{\Gamma (s ) \Gamma (s+\alpha )\left(s^2-s+(\alpha+1)(\ell+1)\right)}{\Gamma (s-\ell) \sqrt{\Gamma (\ell+2)\Gamma (\alpha+2)\Gamma (\ell+\alpha+2)}}.
\end{align*}
These complete the proof of~\eqref{eq:Mproduct_Laguerre}. Similarly to the Hermite case, Eq.~\eqref{eq:Mwronskian_Laguerre} can be proved by using~\eqref{eq:Laguerre_rec1}-\eqref{eq:Laguerre_rec2}, and the elementary properties of the Mellin transform~\eqref{eq:Mellin_d}-\eqref{eq:Mellin_md}.
\end{proof}

Using the very same method one can show that the Mellin transform of products and Wronskians of two Jacobi wavefunctions is essentially a Wilson polynomial. The proof follows the same lines as the Hermite and Laguerre cases and is omitted.
\begin{theorem} 
\label{thm:Wronsk_JUE}
Let $\phi_n(x)$ be Jacobi wavefunctions~\eqref{eq:wavefunctions}. Then,
\begin{itemize}
\item[i)] the Mellin transform of the products is
\begin{align}
\omega_{n,\ell}^*(s) &=\sqrt{\frac{(\alpha_1+\alpha_2+2n+1)\Gamma(\alpha_1+\alpha_2+n+1)}{n!\Gamma(\alpha_1+n+1)\Gamma(\alpha_2+n+1)}}\nonumber\\
&\small\times\sqrt{\frac{(\alpha_1+\alpha_2+2(n+\ell)+1)\Gamma(\alpha_1+\alpha_2+n+\ell+1)}{(n+\ell)!\Gamma(\alpha_1+n+\ell+1)\Gamma(\alpha_2+n+\ell+1)}}\nonumber\\
&\small\times(-1)^n\Gamma(\alpha_1+n+\ell+1)\frac{\Gamma(s)\Gamma(\alpha_2+s)}{\Gamma(s-\ell)\Gamma(\alpha_1+\alpha_2+2n+\ell+1)}\nonumber\\
 &\small\times W_n\left(-\left(s-\frac{1}{2}\right)^2;\;\ell+\frac{1}{2},\frac{1}{2},\alpha_2+\frac{1}{2},-\alpha_1-\alpha_2-2n-\ell-\frac{1}{2}\right),
 \label{eq:Mproduct_Jacobi}
\end{align}
\item[ii)] the difference of Mellin transforms of the Wronskians is
\be
\label{eq:Mwronskian_Jacobi}
W^*_{n,\ell}(s)-W^*_{n,\ell}(s-1) =\frac{\ell}{s-1}\omega_{n,\ell}^*(s).
\ee
\end{itemize}
\end{theorem}
\par
We can also calculate the Mellin transform of a single Jacobi wavefunction in terms of a continuous Hahn polynomial. The interesting case turns out to be for weights with slightly shifted parameters.
\begin{theorem}
Consider the functions
\be
\tilde{\phi}_{n}(x) = x^{\frac{\alpha_2-1}{2}}(1-x)^{\frac{\alpha_1-1}{2}}P^{(\alpha_1,\alpha_2)}_{n}(x)\chi_{[0,1]}(x) \label{Jacobiwavemodweight}
\ee
Then the Mellin transform can be written in terms of continuous Hahn polynomials:
\be
\begin{split}
&\tilde{\phi}^{*}_{n}(s) = \frac{\Gamma\left(s+\frac{\alpha_2-1}{2}+n\right)}{\Gamma\left(s+\frac{\alpha_1+\alpha_2}{2}+n\right)}\Gamma\left(\frac{\alpha_1+1}{2}\right)(-i)^{n}\\
& \times p_{n}(-i(s-1),(\alpha_2+1)/2,-(\alpha_1+\alpha_2)/2-n,(\alpha_2+1)/2,-(\alpha_1+\alpha_2)/2-n) \label{conthahnjacobi}.
\end{split}
\ee
These polynomials have zeros on the vertical line $\mathrm{Re}(s)=1$ and satisfy an orthogonality condition. They are invariant under the reflection $s \to 2-s$ (up to a change of sign if $n$ is odd).
\end{theorem}
\begin{proof}
To identify the continuous Hahn polynomial, one employs the standard expansion of Jacobi polynomials and calculates the Mellin transform integrating term by term. The result is a hypergeometric sum which can be identified with definition \ref{eq:ContinuousDualHahn}, leading to \eqref{conthahnjacobi}. The difficulty in establishing the conclusion is that these polynomials are only known to be orthogonal when the parameters have positive real part. 

We proceed by expressing the right-hand side of \eqref{conthahnjacobi} in terms of Wilson polynomials and apply a result of Neretin \cite{Neretin02}. We claim that for generic parameters $a,b \in \mathbb{C}$, we have
\begin{align}
p_{2n}(x,a,b,a,b) &\propto W_{n}(x^{2},0,1/2,a,b) \label{wiltohahn1}\\
p_{2n+1}(x,a,b,a,b) &\propto xW_{n}(x^{2},1,1/2,a,b) \label{wiltohahn2}.
\end{align}
up to a constant independent of $x$. If $a$ and $b$ have positive real part, this follows by writing the orthogonality condition for the Wilson polynomials with the parameters given in \eqref{wiltohahn1} or \eqref{wiltohahn2}. Use of the duplication formula for the Gamma function shows that the weight reduces to $|\Gamma(a+ix)|^{2}\,|\Gamma(b+ix)|^{2}$ which is the weight function for continuous Hahn polynomials. If $a$ or $b$ have negative real part the identity follows by analytic continuation. Now inserting \eqref{wiltohahn1} and \eqref{wiltohahn2} into \eqref{conthahnjacobi} and applying  the orthogonality \eqref{neretin} completes the proof.
\end{proof}

Going back to random matrices, we obtain exact ordinary differential equations for the one-point functions of the classical ensembles.
Denote by $\rho_n^{(2)}(x)$ the one-point function of the ensembles GUE, LUE, or JUE. As already discussed in the Introducion, $\rho^{(2)}_n(x)$ can be represented as the Wronskian of two adjacent wavefunctions. 
\par
The following proposition is a corollary of the previous theorems on Mellin transforms of Wronskians. For GUE and LUE we recover a result obtained and used earlier by G\"otze and Tikhomirov~\cite[Lemma 2.1 and Lemma 3.1]{Gotze05}. For the JUE a similar result does not seem to have been published.

\begin{prop} 
\label{prop:ODE}
The one-point correlation function $\rho_n^{(2)}(x)$ satisfies the  differential equation $D\rho_n^{(2)}(x)=0$, where
\begin{align*}
Dy=
\left\{  \begin{array}{l@{\quad}cr} 
\displaystyle y'''+(4n-x^2)y'+xy
&(\GUE)  \\ \\
\displaystyle x^2y'''+4xy''+(x-a)(b-x)y'+\left(\frac{1+b}{2}-\frac{\alpha^2}{x}\right)y,
&(\LUE) \\ \\
\displaystyle  x^2((1-x)^3y)''''+x((4+x)(1-x)^2y)''+(1-x) A(x)y'+B(x)y,&(\JUE)
\end{array}\right.
\label{eq:ODE1point}
\end{align*}
where 
\begin{equation*}
a=m+n-\sqrt{4mn+2}, \quad b=m+n+\sqrt{4mn+2},
\end{equation*}
and 
\begin{equation*}
A(x)=-\left((\alpha_1+\alpha_2)^2-2\right)x^2 -4 n x(x-1) (\alpha_1+\alpha_2+n)+2 \alpha_2x (\alpha_1+\alpha_2)-\alpha_2^2+2,
\end{equation*}
\begin{multline*} 
B(x)=x^2 ((\alpha_1+2 n)^2-4)-6 n x (\alpha_1+n)+2 x-2\\
+\alpha_2(2 x-1) (x-1) (\alpha_1+2 n) +2 n (\alpha_1+n)-\frac{\alpha_2^2}{x} (1-x)^3.
\end{multline*}
\end{prop} 
\begin{proof} We present the proof for the GUE. 
The one-point function $\rho^{(2)}_n(x)$ of the GUE is proportional to the Wronskian of consecutive Hermite functions $W_{n-1,1}(x/\sqrt{2})=\operatorname{Wr}(\phi_{n-1}(x/\sqrt{2}),\phi_{n}(x/\sqrt{2}))$. To prove the Theorem for the GUE it is therefore sufficient to show that the function $W_{n,1}(x)$ satisfies the third-order differential equation 
\be
W_{n,1}'''(x)+4(2(n+1)-x^2)W_{n,1}'(x)+4xW_{n,1}'''(x)=0.
\label{eq:ODE_Wronsk_H}
\nonumber
\ee
The Harer-Zagier recurrence relation~\eqref{eq:HZ_GUE} for the moments of the GUE is in fact a difference equation for the Mellin transform $W_{n,1}^*(s)$
\be
s(s+1)(s+2)W_{n,1}^*(s)+8(s+2)(n+1)W_{n,1}^*(s+2)
-4(s+5)W_{n,1}^*(s+4)=0,
\label{eq:diff_Mellin_GUE}
\ee
Using the properties~\eqref{eq:Mellin_d}-\eqref{eq:Mellin_md} of the Mellin transform we get the claim.

For the LUE and JUE the proof follows the same steps starting from the recurrence relations~\eqref{eq:HT} and~\eqref{eq:rec_JUE}.
\end{proof}

\begin{rmk} It is natural to ask whether Wronskians of nonadjacent wavefunctions satisfy similar differential equations. The Mellin transforms of those Wronskians are essentially hypergeometric OP's (see Theorems~\ref{thm:Wronsk_GUE}, \ref{thm:Wronsk_LUE}, and \ref{thm:Wronsk_JUE}). Hence, they satisfy a discrete Sturm-Liouville problem. Upon inversion of the Mellin transform this discrete problem correspond to a differential equation.

We discuss, for concreteness the case of Wronskian of Hermite wavefunctions $W_{n,\ell}(x)=\operatorname{Wr}(\phi_{n}(x),\phi_{n+\ell}(x))$. Its Mellin transform is given in~\eqref{eq:Mproduct_Hermite}-\eqref{eq:Mwronskian_Hermite}. From the difference equation of Meixner-Pollaczek polynomials 
\begin{multline}
(\ell-s-2)P_n^{\left(\frac{\ell+1}{2}\right)}\left(-\frac{i(s + 1)}{2};\frac{\pi}{2}\right)-2(2n+s+1)P_n^{\left(\frac{\ell+1}{2}\right)}\left(-\frac{i(s + 3)}{2};\frac{\pi}{2}\right)\\
+(\ell+s+4)P_n^{\left(\frac{\ell+1}{2}\right)}\left(-\frac{i(s + 5)}{2};\frac{\pi}{2}\right)=0,
\nonumber
\end{multline}
the Mellin transform of the Wronskian satisfies the difference equation
\begin{multline}
s(s+1)(s+2)(s+3)(\ell-s-2)W_{n,\ell}^*(s)
-4(s+2)(s+3)(2n+\ell+1)(s-\ell+2)W_{n,\ell}^*(s+2)\\
+4((s+4)^2-\ell^2)(s-\ell+2)W_{n,\ell}^*(s+4)=0.
\label{eq:functional_Wronsk_H}
\end{multline}
This implies that, for generic $n$ and $\ell$, $W_{n,\ell}(x)$ satisfies a fifth-order differential equation. When when $\ell=1$, formula~\eqref{eq:functional_Wronsk_H} simplifies as~\eqref{eq:diff_Mellin_GUE},
which corresponds to the third-order equation of Proposition~\ref{prop:ODE}.
\end{rmk}
\subsection{Convolution of hypergeometric OP's}
If $\phi_n(x)$ are the Hermite, Laguerre, or Jacobi functions and $\omega_{n,\ell}(x)=\phi_n(x)\phi_{n+\ell}(x)$, by the convolution property~\eqref{eq:convolution_theorem} of the Mellin transform, we have
\begin{equation}
\frac{1}{2\pi i}\int_{c -i\infty}^{c + i\infty} \phi_n^*(s - u)\phi_{n+\ell}^*(u)du=\omega_{n,\ell}^*(s),
\label{eq:convolution_Mellin}
\end{equation}
with $c$ in the fundamental strip of convergence of the Mellin transform. Note that $\phi_n(x)$ is in $L^2(\R_{+})$ so that the fundamental strip always contains the line $\frac{1}{2}+i\R$. Given that $ \phi_n^*$, $\phi_{n+\ell}^*$ and $\omega_{n,\ell}^*$ have hypergeometric OP's factors, the above formula is a `convolution formula' for hypergeoemetric OP's. Note that this is different from the usual (discrete) convolutions formulas of orthogonal polynomials.
\par
When $\phi_n(x)$ is a Hermite wavefunction, $\phi_n^*(s)$ has a Meixner-Pollaczeck polynomial factor whose parameter depends on the parity of $n$. Using the explicit expressions in Proposition~\ref{prop:Bump} and Theorem~\ref{thm:Wronsk_GUE} we can write the special cases of~\eqref{eq:convolution_Mellin}:
\begin{multline}
\frac{1}{2\pi i}\int_{\frac{1}{2} -i\infty}^{\frac{1}{2} + i\infty}\Gamma\left(\frac{s-u}{2}\right)\Gamma\left(\frac{u}{2}\right)P^{\left(\frac{1}{4}\right)}_{m}\left(\frac{s-u-1/2}{2i};\;\frac{\pi}{2}\right)P^{\left(\frac{1}{4}\right)}_{n}\left(\frac{u-1/2}{2i};\;\frac{\pi}{2}\right)du\\
=2^{\ell-m-n-\frac{3}{2}s+2}i^{2r-m-n}\frac{\sqrt{\pi}(2r)!}{m!n!}
\frac{\Gamma(s)}{\Gamma\left(\frac{s+1}{2}-\ell\right)}\,P_{2r}^{\left(\ell+\frac{1}{2}\right)}\left(\frac{s}{2i};\frac{\pi}{2}\right)
\end{multline}
\begin{multline}
\frac{1}{2\pi i}\int_{\frac{1}{2}  -i\infty}^{\frac{1}{2}  + i\infty}\Gamma\left(\frac{s-u+1}{2}\right)\Gamma\left(\frac{u+1}{2}\right)P^{\left(\frac{3}{4}\right)}_{m}\left(\frac{s-u-1/2}{2i};\;\frac{\pi}{2}\right)P^{\left(\frac{3}{4}\right)}_{n}\left(\frac{u-1/2}{2i};\;\frac{\pi}{2}\right)du\\
=2^{\ell-m-n-\frac{3}{2}s}i^{2r-m-n+1}\frac{\sqrt{\pi}(2r+1)!}{m!n!}
\frac{\Gamma(s)}{\Gamma\left(\frac{s+1}{2}-\ell\right)}\,P_{2r+1}^{\left(\ell+\frac{1}{2}\right)}\left(\frac{s}{2i};\frac{\pi}{2}\right)
\end{multline}
\begin{multline}
\frac{1}{2\pi i}\int_{\frac{1}{2}  -i\infty}^{\frac{1}{2}  + i\infty}\Gamma\left(\frac{s-u+1}{2}\right)\Gamma\left(\frac{u}{2}\right)P^{\left(\frac{3}{4}\right)}_{m}\left(\frac{s-u-1/2}{2i};\;\frac{\pi}{2}\right)P^{\left(\frac{1}{4}\right)}_{n}\left(\frac{u-1/2}{2i};\;\frac{\pi}{2}\right)du\\
=2^{\ell-m-n-\frac{3}{2}s+\frac{3}{2}}i^{2r-m-n}\frac{\sqrt{\pi}(2r)!}{m!n!}
\frac{\Gamma(s)}{\Gamma\left(\frac{s}{2}-\ell\right)}\,P_{2r}^{\left(\ell+1\right)}\left(\frac{s}{2i};\frac{\pi}{2}\right)
\end{multline}
with $r=\min(m,n)$ and $\ell=|m-n|$.

In a similar way, from Proposition~\ref{prop:Coffey} and Theorem~\ref{thm:Wronsk_LUE}, the convolution property~\eqref{eq:convolution_Mellin} gives the identity
\begin{multline}
\frac{1}{2\pi}\int\limits_{\frac{1}{2}-i\infty}^{\frac{1}{2}+i\infty}\Gamma\left(s-u+\frac{\alpha}{2}\right)\Gamma\left(u+\frac{\alpha}{2}\right)P^{\left(\frac{\alpha+1}{2}\right)}_{m}\left(\frac{s-u-1/2}{i};\;\frac{\pi}{2}\right)P^{\left(\frac{\alpha+1}{2}\right)}_{n}\left(\frac{s-1/2}{i};\;\frac{\pi}{2}\right)du\\
=\frac{(-1)^{mn+\ell}}{m!n!}\frac{\Gamma(s)\Gamma(s+\alpha)}{2^{s+\alpha}\Gamma(s-\ell)} S_r\left(-\left(s-\frac{1}{2}\right)^2;\ell+\frac{1}{2},\frac{1}{2},\alpha+\frac{1}{2}\right)
\end{multline}
with $r=\min(m,n)$ and $\ell=|m-n|$.

\section{Higher order cumulants}
\label{sec:further}
It is tempting to look for reciprocity formulae for cumulants of higher order (covariances, etc.). 
Write the moments as $Q_k^{\C}(m,n)=\E\operatorname{Tr} X_n^k$, second order moments as $Q_{k,l}^{\C}(m,n)=\E\operatorname{Tr} X_n^k\operatorname{Tr}X_n^l$,
and covariances as $C_{k,l}^{\C}=Q_{k,l}^{\C}(m,n)-Q_k^{\C}(m,n)Q_l^{\C}(m,n)$.

Positive and negative moments of LUE matrices satisfy recursion relations known as \emph{loop equations}. The following lemma can be proved using standard methods in random matrix theory (see, e.g.~\cite{Eynard15} for similar loop equations for positive moments of the GUE).\begin{lem}[Loop equations for positive and negative moments of LUE] For $k_1,\dots,k_v\in\N$, the positive and negative moments satisfy the relations (loop equations):
\begin{multline}
\sum_{\ell=0}^{k_1-1}\mathbb{E}\operatorname{Tr}X_n^{\ell}\operatorname{Tr}X_n^{k_1-\ell+1}\operatorname{Tr}X_n^{k_2}\cdots\operatorname{Tr}X_n^{k_v}+\sum_{j=2}^vk_j\mathbb{E}\operatorname{Tr}X_n^{k_1+k_j-1}\prod_{\substack{i=2\\ i\neq j}}^{v}\operatorname{Tr}X_n^{k_i}\\
=\mathbb{E}\operatorname{Tr}X_n^{k_1}\operatorname{Tr}X_n^{k_2}\cdots\operatorname{Tr}X_n^{k_v}
-\alpha\mathbb{E}\operatorname{Tr}X_n^{k_1-1}\operatorname{Tr}X_n^{k_2}\cdots\operatorname{Tr}X_n^{k_v},
\label{eq:looppos}
\end{multline}
\begin{multline}
\sum_{\ell=0}^{k_1-1}\mathbb{E}\operatorname{Tr}X_n^{-\ell-1}\operatorname{Tr}X_n^{-k_1+\ell}\operatorname{Tr}X_n^{-k_2}\cdots\operatorname{Tr}X_n^{-k_v}+\sum_{j=2}^vk_j\mathbb{E}\operatorname{Tr}X_n^{-k_1-k_j-1}\! \prod_{\substack{i=2\\ i\neq j}}^{v}\operatorname{Tr}X_n^{-k_i}\\
=-\mathbb{E}\operatorname{Tr}X_n^{-k_1}\operatorname{Tr}X_n^{-k_2}\cdots\operatorname{Tr}X_n^{-k_v}
+\alpha\mathbb{E}\operatorname{Tr}X_n^{-k_1-1}\operatorname{Tr}X_n^{-k_2}\cdots\operatorname{Tr}X_n^{-k_v},
\label{eq:loopneg}
\end{multline}
provided they exist.
\end{lem}
\begin{prop}[Reflection symmetry for LUE covariances]
\be
C_{-k,-1}^{\C} = \frac{C_{k,1}^{\C}}{\alpha^2 (\alpha^2-1)\cdots (\alpha^2-k^2)}
\label{eq:reciprocity_covLUE}
\ee
\end{prop}
\begin{proof} Note that $Q_0^{\C}(m,n)=n$, $Q_1^{\C}(m,n)=mn$ and $Q_{-1}^{\C}(m,n)=n/\alpha$.
The proof of~\eqref{eq:reciprocity_covLUE} uses the loop equations~\eqref{eq:looppos}-\eqref{eq:loopneg} together with the reciprocity law~\eqref{eq:dualityLUE} for moments, as follows.
By the loop equations,
\[
Q_{k,1}^{\C}(m,n)= m Q_{k,0}^{\C}(m,n) + k Q_{k}^{\C}(m,n) = Q_{1}^{\C}(m,n)Q_{k}^{\C}(m,n) + k Q_{k}^{\C}(m,n)
\]
and
\[
\alpha Q_{-k,-1}^{\C}(m,n) = Q_{-k,0}^{\C}(m,n) + k Q_{-k-1}^{\C}(m,n) = m Q_{-k}^{\C}(m,n) + k Q_{-k-1}^{\C}(m,n).
\]
The first gives 
\[
C_{k,1}^{\C} = k Q_{k}^{\C}(m,n).
\] 
Using $Q^{\C}_{-1}(m,n)=n/\alpha$, the second gives 
\[
Q_{-k,-1}^{\C}(m,n) = Q_{-1}^{\C}(m,n) Q_{-k}^{\C}(m,n) + k Q_{-k-1}^{\C}(m,n) / \alpha,
\]
that is $C_{-k,-1} ^{\C}= k Q_{-k-1}^{\C}(m,n) / \alpha$.
Now~\eqref{eq:dualityLUE} gives
\[
C_{-k,-1}^{\C} = \frac{k Q_{k}^{\C}(m,n)}{ \alpha^2 (\alpha^2-1)\cdots (\alpha^2-k^2)}
= \frac{C_{k,1}^{\C}}{ \alpha^2 (\alpha^2-1) \cdots (\alpha^2-k^2)}
\]
as required.
\end{proof}
There exists also a precise reflection symmetry for the covariances of one-cut $\beta$-ensembles at leading order in $n$. Suppose that the eigenvalues $x_1,\dots,x_n$ of $X_n$ have a joint probability density proportional to
\be
\prod_{1 \leq j < k \leq n}|x_{k}-x_{j}|^{\beta}\prod_{i=1}^n e^{-n\beta V(x_i)}dx_{i}\quad (\beta>0)
\ee
on the real line. Special cases of one-cut $\beta$-ensembles are the Laguerre and Jacobi ensembles defined by the weights~\eqref{eq:weights} with $\alpha=(c-1)n$ and $\alpha_{1,2}=(c_{1,2}-1)n$, with $c,c_{1},c_{2}>0$ (see discussion in Section~\ref{sec:asymptotics}).
Denote the covariances by 
\be
C_{k,l}=\E\operatorname{Tr} X_n^k\operatorname{Tr}X_n^l-\E\operatorname{Tr} X_n^k\E\operatorname{Tr}X_n^l.
\ee
The  two-point connected correlator 
\be
G_2(z,w)=\E\operatorname{Tr}\frac{1}{z-X_n}\operatorname{Tr}\frac{1}{w-X_n}-\E\operatorname{Tr}\frac{1}{z-X_n}\E\operatorname{Tr}\frac{1}{w-X_n}.
\ee
is the generating function of covariances of positive and negative moments 
\begin{align}
G_2(z,w)&=\sum_{k,l=0}^{\infty}C_{k,l}z^{-(k+1)}w^{-(l+1)} &&\text{as $z,w\to \infty$}\label{eq:Ggen1}\\
&=\sum_{k,l=0}^{\infty}C_{-(k+1),-(l+1)}z^{k}w^{l} &&\text{as $z,w\to 0$}.
\label{eq:Ggen2}
\end{align}

For one-cut $\beta$-ensembles the large $n$ limit of $G_2(z,w)$ exists and depends only on the edges of the cut~\cite{Ambjorn90,BrezinZee93,Beenakker94,Cunden15,Cunden16}, see Eq.~\eqref{eq:G20} below. (On the other hand, for multicut ensembles the asymptotics of $G_2(z,w)$ is more delicate due to the presence of oscillating terms~\cite{Albeverio01,BrezinDeo98}). In the Laguerre ensemble set $\alpha=m-n=(c-1)n$, with $c>1$. The edges $x_{\pm}$ of the cut are strictly positive, see~\eqref{eq:equi_Laguerre}. In the Jacobi ensemble set $\alpha_{1,2}=(c_{1,2}-1)$, with $c_{2}>1$, so that the cut $[x_-,x_+]$ is contained in the interval $(0,1]$, see~\eqref{eq:equi_Jacobi}.
\begin{theorem}[Covariances of Laguerre and Jacobi ensembles at leading order in $n$] Let $X_n$ be in the Laguerre (resp. Jacobi) ensemble with $c>1$ (resp. $c_2>1$). Denote the edges of the cut by $0<x_-<x_+$.  Then, for all $\beta>0$,
\be
\lim_{n\to\infty}C_{-k,-l}=\left(\frac{1}{x_-x_+}\right)^{k+l}\lim_{n\to\infty}C_{k,l}.
\ee
More explicitly:
\begin{align}
\lim_{n\to\infty}C_{-k,-l}&=\left(\frac{1}{c-1}\right)^{2(k+l)}\lim_{n\to\infty}C_{k,l}&\text{(Laguerre)}\\
\lim_{n\to\infty}C_{-k,-l}&=\left(\frac{c_1+c_2}{c_2-1}\right)^{2(k+l)}\lim_{n\to\infty}C_{k,l}&\text{(Jacobi)}.
\end{align}
\begin{proof} By the one-cut property, the limit
\be
G_{2,0}(z,w)=\lim_{n\to\infty}G(z,w),
\ee
 is given by the explicit formula
\be
G_{2,0}(z,w)=\frac{1}{\beta}\frac{1}{(z-w)^2}\left(\frac{zw-(x_-+x_+)(z+w)/2+x_-x_+}{\sqrt{(z-x_-)(z-x_+)(w-x_-)(w-x_+)}}-1\right).
\label{eq:G20}
\ee
Moreover, since the cut does not contain zero, the negative covariances $\lim_{n\to\infty} C_{-k,-l}$ exist.
From~\eqref{eq:G20} it is easy to verify the following functional equation
\be
G_{2,0}\left(\frac{x_-x_+}{z},\frac{x_-x_+}{w}\right)=\left(\frac{x_-x_+}{zw}\right)^2G_{2,0}(z,w).
\ee
Using~\eqref{eq:Ggen1}-\eqref{eq:Ggen2}, the claim follows. 
\end{proof}
\end{theorem}

\section{Orthogonal and Symplectic ensembles}
\label{sec:Orth_Sympl}
We will now discuss analogous results for the orthogonal and symplectic ensembles of random matrices, corresponding to averages over the density \eqref{eq:integrals} with $\beta=1$ or $\beta=4$ respectively. Our aim will be to isolate the polynomial factors of the moments (Mellin transforms), again for the Gaussian, Laguerre and Jacobi ensembles. These ensembles are characterized by their joint eigenvalue distribution as in \eqref{eq:integrals} with corresponding weight functions \eqref{eq:weights}. 
We will consider various expectation values of power of traces 
with respect to \eqref{eq:integrals} with $\beta=1$ and $4$. 
We will use the shorthand GSE and GOE to mean Gaussian Symplectic Ensemble, GSE ($\beta=4$), Gaussian Orthogonal Ensemble, GOE ($\beta=1$) and similarly for the Laguerre (LSE / LOE) and Jacobi (JSE / JOE) cases. As in the complex case, we will denote by $\alpha=m-n$ in the Laguerre case, and $\alpha_1=m_1-n$ and $\alpha_2=m_{2}-n$ in the Jacobi case, which we treat as fixed $n$-independent parameters.

Moments of real and quaternionic Gaussian ensembles have already received some attention in the literature, however much less is known compared with the complex case $\beta=2$. One of the first explicit formulas was derived by Goulden and Jackson \cite{Goulden-Jackson} who were motivated by the fact that, for $\beta=1$,  moments of Gaussian matrices describe the genus expansion of \textit{non-orientable} surfaces. An important development was achieved in the work of Ledoux \cite{Ledoux09} who discovered recursion relations for moments of the GOE and GSE (which can be viewed as real and quaternionic analogues of the Harer-Zagier recurrence relations). This was extended to the Laguerre ensemble in \cite{Cunden16}. Results holding for complex moments  were obtained in \cite{Simm11} .

\subsection{Recurrence relations and hypergeometric representations}
\label{sub:GFbeta1}
We define
\begin{align*}
Q_k^{\R}(n)&=\E\Tr X_n^{2k} &\text{if $X_n\sim\GOE$}\\
Q_k^{\Hh}(n)&=\E\Tr X_n^{2k} &\text{if $X_n\sim\GSE$}\\
Q_k^{\R}(m,n)&=\E\Tr X_n^{k}&\text{if $X_n\sim\LOE$}\\
Q_k^{\Hh}(m,n)&=\E\Tr X_n^{k}&\text{if $X_n\sim\LSE$}\\
Q_k^{\R}(\alpha_1,\alpha_2,n)&=\E\Tr X_n^{k}&\text{if $X_n\sim\JOE$}\\
Q_k^{\Hh}(\alpha_1,\alpha_2,n)&=\E\Tr X_n^{k} &\text{if $X_n\sim\JSE$}\\
\Delta Q_k^{\R}(\alpha_1,\alpha_2,n)&=Q_k^{\R}(\alpha_1,\alpha_2,n)-Q_{k+1}^{\R}(\alpha_1,\alpha_2,n)\\
\Delta Q_k^{\Hh}(\alpha_1,\alpha_2,n)&=Q_k^{\Hh}(\alpha_1,\alpha_2,n)-Q_{k+1}^{\Hh}(\alpha_1,\alpha_2,n)
\end{align*}
for all $k\in\C$ for which the expectations exist.
Moments of the classical orthogonal ensembles satisfy recursions similar to those of the unitary ensembles. To our knowledge this was first noticed by Ledoux for the GOE~\cite{Ledoux09} and extended to the LOE in~\cite{Cunden16b}. 
The first question is whether moments of orthogonal / symplectic ensembles enjoy reflection symmetries and have orthogonal polynomial factors as in the unitary case. This is not the case as can be ascertained from the following observation. The Harer-Zagier recursion for moments $\Q^{\C}_k(n)$ of the GUE is a three terms recursion in $k$ which can be interpreted as the discrete S-L problem of some families of hypergeometric (Meixner / Meixner-Pollaczek) polynomials. 
Moments of the classical orthogonal ensembles satisfy recursion formulae too. 
For the GOE, Ledoux~\cite{Ledoux09} discovered that $Q^{\R}_k(n)$ satisfy a five term recurrence relation which cannot be interpreted as a
S-L equation (a second order difference equation). In fact, Ledoux also found an alternative inhomogeneous recursion formula for $Q^{\R}_k(n)$ coupled with the moments of the GUE that is somewhat more convenient for some application. An analogue of this coupled recursion was obtained later for the LOE~\cite{Cunden16b}. 

These results suggest that \emph{suitable combinations of moments, rather than the moments themselves, have nice hypergeometric polynomial factors} similar to the unitary cases. 

Moments of the symplectic ensembles can be analysed similarly given the duality relations between moments of GSE, LSE and JSE of size $n$ and the (formal) moments of the GOE, LOE and JSE of size $-2n$
\begin{align}
Q_k^{\Hh}(n)&=(-1)^{k+1}2^{-1}Q_k^{\R}(-2n)
\label{eq:dualityG}\\
Q_k^{\Hh}(m,n)&=(-1)^{k+1}2^{-1}Q_k^{\R}(-2m,-2n)
\label{eq:dualityL}\\
Q_k^{\Hh}(\alpha_1,\alpha_2,n)&=-2^{-1} Q_k^{\R}(-2\alpha_1,-2\alpha_2,-2n).
\label{eq:dualityJ}
\end{align}
For the GOE/GSE the duality was put forward by Mulase and Waldron in terms diagrammatic expansion of Gaussian integrals~\cite{Mulase03}. See also~\cite{Ledoux09,Bryc09}.
This duality between orthogonal and symplectic  Laguerre ensembles appeared in the paper of Hanlon, Stanley and Stembridge~\cite[Corollary 4.2]{Hanlon92}.  The duality in the Jacobi ensembles has been observed by Forrester, Rahman and Witte~\cite[Eq. (4.15)]{Forrester17}. See also~\cite{Dumitriu15} and ~\cite[Appendix B]{Fyodorov16}.

\begin{theorem}
\label{thm:GOE1} The combinations of GOE and GSE moments
\begin{align}
S_k^{\R}(n)&=Q^{\R}_{k+1}(n)-(4n-2)Q^{\R}_{k}(n)-8k(2k-1)Q^{\R}_{k-1}(n)
\label{eq:S_k_R}\\
S_k^{\Hh}(n)&=2Q^{\Hh}_{k+1}(n)-(16n+4)Q^{\Hh}_{k}(n)-16k(2k-1)Q^{\Hh}_{k-1}(n)
\end{align}
have Meixner polynomial factors:
\begin{align}
S_k^{\R}(n)&=-3n(n-1)\;(2k-1)!!\;M_{n-2}(k;3,-1)
\label{eq:GOE_Meixner1}\\
&=-3n(n-1)\;(2k-1)!!\;M_{k}(n-2;3,-1)
\label{eq:GOE_Meixner2}\\
S_k^{\Hh}(n)&=-6n(2n+1)\;(2k-1)!!\;M_{2n-1}(k;3,-1)
\label{eq:GSE_Meixner1}\\
&=-6n(2n+1)\;(2k-1)!!\;M_{k}(2n-1;3,-1).
\label{eq:GSE_Meixner2}
\end{align}
In particular, for any integer $n$, $S_{k}^{\R}(n)/(2k-1)!!$ and $S_{k}^{\Hh}(n)/(2k-1)!!$ are Meixner-Pollaczek polynomials in $x=-i(k+3/2)$
\begin{align}
\frac{S_k^{\R}(n)}{(2k-1)!!}&=-6i^{n-2}\;P^{(3/2)}_{n-2}(x;\pi/2)
\label{eq:GOE_MP1}\\
\frac{S_k^{\Hh}(n)}{(2k-1)!!}&=-6i^{2n-1}\;P^{(3/2)}_{2n-1}(x;\pi/2)
\label{eq:GSE_MP1}
\end{align}
 invariant up to a change of sign under the reflection $k\to-3-k$, with complex zeros on the vertical line $\operatorname{Re}(k)=-3/2$.
\begin{proof} We first prove~\eqref{eq:GOE_Meixner1}. We read from Ledoux paper~\cite[Theorem 3]{Ledoux09}
\begin{align*}
S_k^{\R}(n)=Q^{\C}_{k+1}(n)-(4n-3)Q^{\C}_{k}(n)-2k(2k-1)Q^{\C}_{k-1}(n).
\end{align*}
Using the polynomial property of $Q^{\C}_{k}(n)$, this remainder can be expressed as a sum of two Meixner polynomials
\begin{align*}
S_k^{\R}(n)=3n\;(2k-1)!!\;\left(M_{n-1}(k;2,-1)-M_{n-1}(k+1;2,-1)\right)
\end{align*}
The representation~\eqref{eq:GOE_Meixner1} in terms of a single Meixner polynomial follows using the Forward Shift Operator~\cite[Eq. (9.10.6)]{Koekoek10}.
\par
To prove~\eqref{eq:GSE_Meixner1} the starting point is again a result in Ledoux paper~\cite[Theorem 5]{Ledoux09}
\begin{align*}
S_k^{\Hh}(n)=Q_{k+1}^{\R}(2n+1)-(8n+2)Q_{k}^{\R}(2n+1)-8k(2k-1)Q_{k-1}^{\R}(2n+1),
\end{align*}
 which can be written in the insightful form
\begin{align*}
S_k^{\Hh}(n)=S_k^{\R}(2n+1).
\end{align*}
The representation~\eqref{eq:GSE_Meixner1} is now a consequence of~\eqref{eq:GOE_Meixner1}. Alternatively, by the duality relation~\eqref{eq:dualityG}
\begin{align*}
S_k^{\Hh}(n)=(-1)^{k}\left(Q_{k+1}^{\R}(-2n)+(8n+2)Q_{k}^{\R}(-2n)-8k(2k-1)Q_{k-1}^{\R}(-2n)\right),
\end{align*}
that is
\begin{align*}
S_k^{\Hh}(n)=(-1)^{k}S_k^{\R}(-2n).
\end{align*}
Using the self-duality of Meixner polynomials, we write 
\begin{align*}
S_k^{\Hh}(n)=(-1)^{k+1}6n(2n+1)\;(2k-1)!!\;M_{k}(-2n-2;3,-1).
\end{align*}
Now we use the symmetry $(-1)^{k+1}M_{k}(-2n-2;3,-1)=M_{k}(2n+1;3,-1)$, and the self-duality again to conclude the proof.
\end{proof}
\end{theorem}
The S-L problem satisfied by the Meixner polynomials is a three term recursion formula for $S^{\R}_k(n)$ and $S^{\Hh}_k(n)$,
\begin{align}
(k+3)S^{\R}_{k+1}(n)=(2k+1)(2n-1)S^{\R}_{k}(n)+k(2k+1)(2k-1)S^{\R}_{k-1}(n)\\
(k+3)S^{\Hh}_{k+1}(n)=(2k+1)(4n+1)S^{\Hh}_{k}(n)+k(2k+1)(2k-1)S^{\Hh}_{k-1}(n)
\end{align}
These recursions, which are very similar to the Harer-Zagier formula, become five term recurrences for the moments $Q_{k}^{\R}(n)$ (this is Ledoux recursion~\cite[Theorem 2]{Ledoux09}) and $Q_{k}^{\Hh}(n)$. 
\begin{cor} 
\begin{align}
(k+1)Q_k^{\R}(n)&=(4k-1)(2n-1)Q_{k-1}^{\R}(n)\nonumber\\
&+(2k-3)(10k^2-9k-8n^2+8n)Q_{k-2}^{\R}(n)\nonumber\\
&-5(2k-3)(2k-4)(2k-5)(2n-1)Q_{k-3}^{\R}(n)\nonumber\\
&-2(2k-3)(2k-4)(2k-5)(2k-6)(2k-7)Q_{k-4}^{\R}(n).
\end{align}
\begin{align}
(k+1)Q_k^{\Hh}(n)&=(4k-1)(4n+1)Q_{k-1}^{\Hh}(n)\nonumber\\
&+(2k-3)(10k^2-9k-32n^2-16n)Q_{k-2}^{\Hh}(n)\nonumber\\
&-5(2k-3)(2k-4)(2k-5)(4n+1)Q_{k-3}^{\Hh}(n)\nonumber\\
&-2(2k-3)(2k-4)(2k-5)(2k-6)(2k-7)Q_{k-4}^{\Hh}(n).
\end{align}
\end{cor}
\begin{theorem}
\label{thm:LOE1} 
Set $\alpha=m-n$, and consider the following combinations of moments for the LOE and the LSE. 
\begin{multline}
S_k^{\R}(m,n)=Q^{\R}_{k+1}(m,n)-2(m+n-1)Q^{\R}_k(m,n)\\-(1-\alpha^2+4k(k-1))Q^{\R}_{k-1}(m,n)
\end{multline}
\begin{multline}
S_k^{\Hh}(m,n)=2Q^{\Hh}_{k+1}(m,n)-(8m+8n+4)Q^{\Hh}_k(m,n)\\-(2-8\alpha^2+8k(k-1))Q^{\Hh}_{k-1}(m,n)
\end{multline}
Then, $S_k^{\R}(m,n$ and $S_k^{\Hh}(m,n)$ can be written in terms of dual Hahn polynomials (as functions of $k$) or Hahn polynomials (as functions of $n$):
\begin{align}
\frac{S_k^{\R}(m,n)}{(k+\alpha)!}&=-\frac{6}{(n-2)!(m-2)!}\;S_{n-2}\left(x^2;\;\frac{5}{2},\frac{1}{2},\alpha+\frac{1}{2}\right)
\label{eq:LOE_CDH} \\
&=-\frac{3nm(n-1)(m-1)}{(\alpha+2)!}\;Q_{k-2}(n-2;\;2,2,-3-\alpha)
\label{eq:LOE_CDH2}\\
\frac{S_k^{\Hh}(m,n)}{(k+2\alpha)!}&=-\frac{24nm}{(2n)!(2m)!}
\; S_{2n-1}\left(x^2, \frac{5}{2}, \frac{1}{2}, 
  2 \alpha + \frac{1}{2}\right)
  \label{eq:LSE_CDH} \\
&=-\frac{12nm(2n+1)(2m+1)}{(2\alpha+2)!}\;Q_{k-2}(2n-1;\,2,2,-3-2\alpha)
\label{eq:LSE_CDH2} 
\end{align}
where $k=ix-1/2$. In particular this shows that the polynomials $S_k^{\R}(m,n) /(k+\alpha)!$ and $S_k^{\Hh}(m,n) /(k+2\alpha)!$  are invariant under the reflection $k \to -1-k$ (reciprocity law) and, moreover, their  zeros lie on the critical line $\operatorname{Re} (k)=-1/2$.
\begin{proof} By the inhomogeneous recursion for moments of the LOE~\cite[Theorem 3.5]{Cunden16b}, $S_k^{\R}(m,n)$ is a combination of moments of the LUE 
\begin{align*}
S_k^{\R}(m,n)=\frac{3}{k-1}((m+n-k-1)Q^{\C}_k(m-1,n-1)-Q^{\C}_{k+1}(m-1,n-1)),
 \label{eq:mom_LOE_LUE}
\end{align*}
that can be expressed in terms of continuous dual Hahn polynomials 
\begin{multline*}
S_k^{\R}(m,n)=\frac{3(k+\alpha)!}{(k-1)(m-2)!(n-2)!}\times\\
\biggl((m+n-k-1)S_{n-2}\left(x^2;\;\frac{3}{2},\frac{1}{2},\alpha+\frac{1}{2}\right)
-(m-n+k+1)S_{n-2}\left((x-i)^2;\;\frac{3}{2},\frac{1}{2},\alpha+\frac{1}{2}\right)\biggr)
\end{multline*}
The final result~\eqref{eq:LOE_CDH} can be obtained by using the Forward Shift Operator of the continuous dual Hahn polynomials~\cite[Eq. (9.3.7)]{Koekoek10}.
From the hypergeometric representation of Hahn polynomials we get~\eqref{eq:LOE_CDH2}. 

For the symplectic moments, note that, by duality~\eqref{eq:dualityL} between LOE and LSE moments,
\begin{equation*}
S_k^{\Hh}(m,n)=(-1)^kS_k^{\R}(-2m,-2n).
\end{equation*}
This proves~\eqref{eq:LSE_CDH2}. We use the identity $(3-2\alpha)_{k-2}Q_{k-2}(-2n-2;\;2,2,-3+2\alpha)=(-1)^k(3+2\alpha)_{k-2}Q_{k-2}(2n-1;\;2,2,-3-2\alpha)$ and find~\eqref{eq:LSE_CDH2}.
Now, this can be written as a polynomial in $k$ which can be cast as~\eqref{eq:LSE_CDH}.
\end{proof}
\end{theorem}
\begin{rmk} As in the Gaussian case, there is a reflection formula under the transformation $2m+1\to -2m$ and $2n+1\to -2n$
\be
S^{\R}(2m+1,2n+1)=(-1)^{k}S^{\R}(-2m,-2n).
\ee
\end{rmk}
The S-L problem of continuous dual Hahn polynomials corresponds to the three term recursions (in $k$)
\begin{align}
(k+3)S^{\R}_{k+1}(m,n)&=(2k+1)(2n+\alpha-1)S^{\R}_{k}(m,n)+(k-2)(k^2-\alpha^2)S^{\R}_{k-1}(m,n)\\
(k+3)S^{\Hh}_{k+1}(m,n)&=(2k+1)(4n+2\alpha+1)S^{\Hh}_{k}(m,n)+(k-2)(k^2-4\alpha^2)S^{\Hh}_{k-1}(m,n)
\end{align}
which are similar to the Haagerup-Thorbj\o rnsen formula for the moments of the LUE. Writing $S^{\R}_{k}(m,n)$ and $S^{\Hh}_{k}(m,n)$ in terms of the moments we obtain the following five term recursions. To our knowledge these recursion formulae are new.
\begin{cor}
\begin{align}
(k+1)Q_k^{\R}(m,n)&=AQ_{k-1}^{\R}(m,n)+BQ_{k-2}^{\R}(m,n)+CQ_{k-3}^{\R}(m,n)+DQ_{k-4}^{\R}(m,n).
\end{align}
with
\begin{align}
A&=(4k-1)(m+n-1)\nonumber\\
B&=-(4k-6)(m+n-1)^2+(k-4) \left((k-2)^2-\alpha ^2\right)+(k+1) \left((2k-3)^2-\alpha ^2\right)\nonumber\\
C&=-(m+n-1)\left((2k-3)\left((2k-5)^2-\alpha^2\right)-(2k-8)\left((k-2)^2-\alpha^2\right)\right)\nonumber\\
D&=-(k-4)\left((k-2)^2-\alpha^2\right)\left((2k-7)^2-\alpha^2\right).\nonumber
\end{align}
and
\begin{align}
(k+1)Q_k^{\Hh}(m,n)&=AQ_{k-1}^{\Hh}(m,n)+BQ_{k-2}^{\Hh}(m,n)+CQ_{k-3}^{\Hh}(m,n)+DQ_{k-4}^{\Hh}(m,n).
\end{align}
with
\begin{align}
A&=(4k-1)(2m+2n+1)\nonumber\\
B&=-(4k-6)(2m+2n+1)^2+(k-4) \left((k-2)^2-4\alpha ^2\right)+(k+1) \left((2k-3)^2-4\alpha ^2\right)\nonumber\\
C&=-(2m+2n+1)\left((2k-3)\left((2k-5)^2-4\alpha^2\right)-(2k-8)\left((k-2)^2-4\alpha^2\right)\right)\nonumber\\
D&=-(k-4)\left((k-2)^2-4\alpha^2\right)\left((2k-7)^2-4\alpha^2\right).\nonumber
\end{align}
\end{cor}
Using methods similar to those in~\cite{Cunden16b} it is possible to write a recursion formula for the moments of the JOE. Denote
\begin{multline*}
S_k^{\R}(\alpha_1,\alpha_2,n)=(2k+4 - \alpha_1 - \alpha_2 - 2 n) (\alpha_1 + \alpha_2 + 2 ( n + k + 1))\Delta Q_{k+1}^{\R}(\alpha_1,\alpha_2,n)\\
+2 (\alpha_1 \alpha_2-\alpha_1 - \alpha_2  + \alpha_2^2 - 4 k (1 + k) - 2 n + 
   2 (\alpha_1 + \alpha_2) n + 2 n^2)\Delta Q_{k}^{\R}(\alpha_1,\alpha_2,n)\\
-  (\alpha_2^2 - (1 - 2 k)^2) \Delta Q_{k-1}^{\R}(\alpha_1,\alpha_2,n).
\end{multline*}
\begin{prop}[Recursion for moments of the JOE] Set $\alpha_1=m_1-n$ and $\alpha_2=m_2-n$. 
Then the differences of adjacent moments $\Delta Q_{k}^{\R}(\alpha_1,\alpha_2,n)$ of the JOE satisfy the following inhomogeneous three term recursion
\begin{multline}
S_k^{\R}(\alpha_1,\alpha_2,n)=\\
\frac{3}{k-1}\Bigl(\left( (\alpha_1 + \alpha_2) (\alpha_2 - k-1) + 2 ( \alpha_1 + \alpha_2 - k-1) n + 2 n^2\right)\Delta Q_{k}^{\C}(\alpha_1,\alpha_2,n-1)\\
- (\alpha_1 + \alpha_2 + 2 n) (\alpha_1 + \alpha_2  + 2 n- k -3)\Delta Q_{k+1}^{\C}(\alpha_1,\alpha_2,n-1)\Bigr),
\end{multline}
\end{prop}
As in the Gaussian and Laguerre cases, the above recursion formula suggests that $S_k^{\R}(\alpha_1,\alpha_2,n)$, and not the moments themselves, have a nice polynomial property. This is the content of the next theorem whose proof goes along the same lines as the Gaussian and Laguerre cases. 
\begin{theorem} Set $k=ix-1/2$. The combination $S_k^{\R}(\alpha_1,\alpha_2,n)$ has a Wilson polynomial factor
\begin{multline}
S_k^{\R}(\alpha_1,\alpha_2,n)=(-1)^{n-1}\frac{6}{(n-2)!}\frac{(\alpha_1+n)(\alpha_1+n-1)(\alpha_1+\alpha_2+n)!}{(\alpha_2+n-2)!}\\\times\frac{(\alpha_2+k)!}{(\alpha_1+\alpha_2+k+2n-1)!} W_{n-2}\left(x^2;\;\frac{5}{2},\frac{1}{2},\alpha_2+\frac{1}{2},\frac{3}{2}-\alpha_1-\alpha_2-2n\right).
\end{multline}
In particular, $S_k^{\R}(\alpha_1,\alpha_2,n) ((\alpha_1+\alpha_2+k+2n-1)!/(\alpha_2+k)!)$ is a polynomial of degree $2(n-2)$ in $k$, invariant under the reflection $k\to-1-k$, with zeros on the vertical line $\operatorname{Re}(k)=-1/2$.
\end{theorem}

\subsection{Symplectic Ensembles} 
The goal of this section is to establish the following polynomial property for the moments of the symplectic ensembles. 
\begin{figure}
\centering
\begin{subfigure}{.5\textwidth}
  \centering
  \includegraphics[width=.8\linewidth]{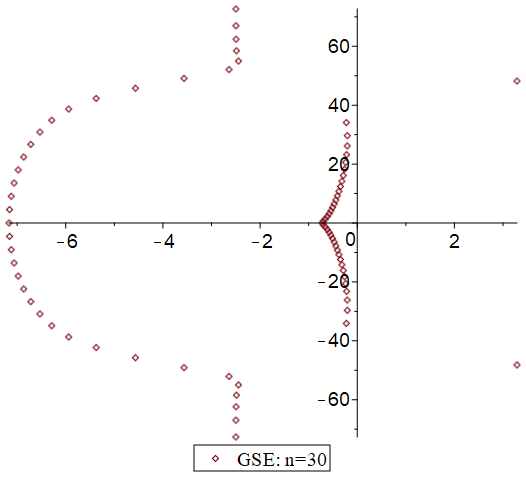}
\end{subfigure}
\begin{subfigure}{.5\textwidth}
  \centering
  \includegraphics[width=.8\linewidth]{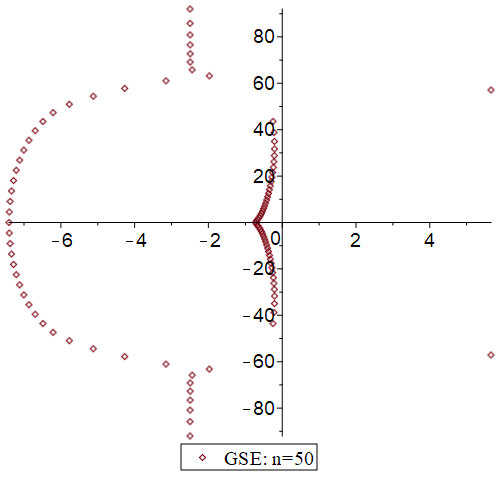}
\end{subfigure}
\caption{The $2(n-1)$ complex zeros of the symplectic polynomial $p^{\mathbb{H}}_{k}(n)$ for the GSE with $n=30$ (left) and $n=50$ (right).}
\label{fig:GSE}
\end{figure}
\begin{theorem}
The rescaled moments
\begin{equation}
p^{\mathbb{H}}_{n}(k) = 
\begin{cases} c_{n}\dfrac{1}{\Gamma(k+1/2)}Q^{\mathbb{H}}_{k}(n)& \mathrm{GSE}\\\\
c_{n,m}\dfrac{2^k}{\Gamma(2\alpha+k+2)}Q^{\mathbb{H}}_{k}(m,n)&  \mathrm{LSE}\\\\
c_{n,\alpha_1,\alpha_2}\dfrac{\Gamma(1+2\alpha_1+2\alpha_2+4n+k)}{\Gamma(2\alpha_1+k+2)}\Delta Q^{\mathbb{H}}_{k}(\alpha_1,\alpha_2,n)&  \mathrm{JSE}
\end{cases} \label{symppolys}
\end{equation}
are monic polynomials in $k$ of degree $2(n-1)$ in the GSE, and degree $4(n-1)$ in the LSE and JSE cases. The normalizing constants are
\begin{equation*}
\begin{cases}c_{n} = 2^{2(1-n)}\Gamma(2n)\sqrt{\pi}\\ 
c_{n,m} = \Gamma(2m)\Gamma(2n)\\ 
c_{n,\alpha_1,\alpha_2} = \dfrac{\Gamma(2n+2\alpha_1)\Gamma(2\alpha_2+4)\Gamma(2n)}{\Gamma(2\alpha_1+2\alpha_2+2n+2)\Gamma(2\alpha_2+2n+2)(2\alpha_2+2)}. 
\end{cases}
\end{equation*}
\label{thm:symp}
\end{theorem}
\begin{proof}
We will discuss the Gaussian case in detail, as the Laguerre and Jacobi cases follow a similar pattern. By \cite[Eq. (33)]{Simm11}, we have the explicit formula
\begin{equation}
Q^{\mathbb{H}}_{k}(n) = 2^{-k-1}Q^{\mathbb{C}}_{k}(2n)-a_{n}\sum_{j=1}^{n}\sum_{i=0}^{n-j}{k \choose i}{k \choose i+j}(n-i-j+1)_{(k-1/2)} \label{GSEmoms}
\end{equation}
where $Q^{\mathbb{C}}_{k}(2n)$ denotes the moments of the GUE (see Section~\ref{sec:gaussian}) and 
\begin{equation*}
a_{n} = \frac{\Gamma(n+1)\Gamma(n)}{\sqrt{\pi}\Gamma(2n)4^{1-n}}.
\end{equation*}
Although formula \eqref{GSEmoms} was only stated in \cite{Simm11} for $k \in \mathbb{N}$, it naturally defines a meromorphic continuation to $k \in \mathbb{C}$, as follows. As a function of $k \in \mathbb{N}$, we have that $Q^{\mathbb{C}}_{2k}(2n)/(\Gamma(k+1/2))$ is a polynomial of degree $2n-1$ in $k$ (see equation \eqref{eq:GUE_Meixner-Pollaczek}) and hence is defined for any $k \in \mathbb{C}$. It remains to study the second term in \eqref{GSEmoms}. Note that
\begin{equation*}
\frac{(n-i-j+1)_{(k-1/2)}}{\Gamma(k+1/2)}
\end{equation*} 
is a polynomial of degree $n-i-j$, while ${k \choose i}$ and ${k \choose i+j}$ are polynomials of degree $i$ and $i+j$ respectively. Hence $Q^{\mathbb{H}}_{k}(n)/\Gamma(k+1/2)$ is a finite sum of polynomials in $k$ and is therefore a polynomial. To compute degrees, notice that the highest degree term in the summand of \eqref{GSEmoms} occurs when $2i+j+n-i-j = n+i$ is maximal, namely when $i=n-1$ implying a degree of $2n-1$. That the degree of the combined polynomials (\textit{i.e.} unitary plus symplectic contribution) is really $2n-2$ is a consequence of the following cancellation. Setting $j=1$ and $i=n-1$ in the summand of \eqref{GSEmoms} and dividing by $\Gamma(k+1/2)$ gives the polynomial
\begin{equation*}
a_{n}{k \choose n-1}{k \choose n}
\end{equation*}
Then Stirling's formula gives the estimate 
\begin{equation*}
a_{n}{k \choose n-1}{k \choose n} = \frac{k^{2n-1}}{(2n-1)!4^{1-n}}+O(k^{2n-2}), \qquad k \to \infty
\end{equation*}
Similarly, consider the complex moments
\begin{equation*}
\frac{1}{2\;\Gamma(k+1/2)}Q^{\mathbb{C}}_{k}(2n) = \frac{1}{\sqrt{\pi}}\sum_{j=0}^{2n-1}2^{j-1}{k \choose j}{2n \choose j+1} \label{complexmoms}
\end{equation*}
which has the same leading coefficient (setting $j=2n-1$) as 
\begin{equation*}
\frac{1}{\sqrt{\pi}}{k \choose 2n-1}2^{2n-2} = \frac{k^{2n-1}4^{n-1}}{(2n-1)!}+O(k^{2n-2}), \qquad k \to \infty
\end{equation*}
Hence the terms of order $k^{2n-1}$ in \eqref{GSEmoms} cancel, yielding a polynomial of degree $2n-2$. To compute the normalizing factor $c_{n}$ requires studying terms of order $k^{2n-2}$. This is a straightforward but tedious task and we omit the details. The only contributions to the monomial $k^{2n-2}$ come from \eqref{complexmoms} when $j=2n-1$ and $j=2n-2$, and from the double sum in \eqref{GSEmoms} with indices $(i,j) = (n-1,1)$ and $(i,j) = (n-2,1), (n-2,2)$. Then studying the asymptotics of these five terms as $k \to \infty$ with Stirling's formula gives the result. For the Laguerre and Jacobi cases, this computation can be repeated with the formulae \cite[Eq. (89) and Eq. (98)]{Simm11} which have an identical structure to~\eqref{GSEmoms} and is therefore omitted. 
\end{proof}
Below are the first few polynomials $p^{\mathbb{H}}_{n}(k)$ for the GSE, whose zeros appear to settle onto an explicit contour in the complex plane as $n$ becomes larger (see Fig.~\ref{fig:GSE}).
\begin{equation*}
\begin{split}
p^{\mathbb{H}}_{1}(k) &= 1\\
p^{\mathbb{H}}_{2}(k) &= k^{2}+5k+3\\
p^{\mathbb{H}}_{3}(k) &= k^{4}+10k^{3}+38k^{2}+41k+\frac{45}{2}\\
p^{\mathbb{H}}_{4}(k) &= k^{6}+15k^{5}+109k^{4}+393k^{3}+637k^{2}+735k+315.
\end{split}
\end{equation*}
\subsection{Orthogonal Ensembles}
In this section we will study the Mellin transform of the one-point correlation function $\rho_n^{(\beta)}(x)$ with $\beta=1$. One can expect this case to be more complicated in general, since now \eqref{eq:integrals} contains a non-analytic term (the absolute value of the Vandermonde determinant, which happened to be a polynomial in the cases $\beta=2$ and $\beta=4$). In the case of $n$ odd we are saved by a remarkable duality principle for the Mellin transform, relating the orthogonal and symplectic ensembles. This duality involves a simple correction term which is a single hypergeometric OP. 
\par
The case of $n$ even has a different analytic structure, evident already at $n=2$. Indeed, it was known since the beginnings of random matrix theory that the parity $n$ plays an important role for ensembles with orthogonal symmetry (see \cite[Chapter 6]{Mehta} for example or more recently \cite{ForMays09}), with most authors assuming $n$ to be even for simplicity. Here it is the converse, we describe the analytic structure for $n$ odd and give an explicit analytic continuation. First we need a proposition relating the orthogonal and symplectic ensembles.

\begin{prop}
\label{prop:dualdens}
Given the notation of Section \ref{sec:def}, let $p_{n}(x)$ denote the degree $n$ monic polynomial orthogonal with respect to the weight $w_{2}(x)$ on the interval $I$. Then the one-point eigenvalue density \eqref{dens} satisfies the following duality
\begin{equation}
\rho^{(1)}_{2n+1}(x) = 2\tilde{\rho}^{(4)}_{n}(x)+\frac{w_{1}(x)p_{2n}(x)}{\int_{I}w_{1}(t)p_{2n}(t)\,dt} \label{dualitydens}
\end{equation}
where $\tilde{\rho}^{(4)}(x)$ is the $\beta=4$ eigenvalue density with respect to the modified weights
\begin{equation}
\tilde{w}_{4}(x) = \begin{cases} e^{-x^{2}/2} & \text{Hermite}\\
x^{\alpha+1}e^{-x} & \text{Laguerre}\\
x^{\alpha_1+1}(1-x)^{\alpha_2+1} & \text{Jacobi}.
\end{cases} \label{modweight}
\end{equation}
\end{prop}
\begin{rmk}
We give a complete proof of Propostition \ref{prop:dualdens} in Appendix \ref{app:orthsymp}, which is based on the skew-orthogonal polynomial formalism developed in \cite{AFNvM}. In the specific case of the GOE, formula \eqref{dualitydens} was mentioned in \cite{Ledoux09}. Actually, the statement of Proposition~\ref{prop:dualdens} is implicit in Forrester's book~\cite[(6.120)-(6.122)]{Forrester_book}. It is worth emphasizing that this duality goes beyond the one-point function and can be formulated as a duality between the \textit{correlation kernels} of $n$-odd orthogonal and symplectic ensembles. This suggests a possibly simpler route to studying correlation functions of $n$-odd orthogonal ensembles, but this lies beyond the scope of the current investigation. 
\end{rmk}
We now study the consequences of the duality \eqref{dualitydens} for the Mellin transforms of the orthogonal ensembles. 
\begin{theorem}[Duality in the $n$-odd orthogonal ensemble]
\label{thm:duality}
In the three orthogonal ensembles the following identity holds for all $k \in \mathbb{C}$ and $n \in \mathbb{N}$:
\begin{align}
Q_k^{\R}(2n+1)&=2^{k+1}Q^{\Hh}_{k}(n)+4^{k}\Gamma\left(k+1/2\right)f_{k}(n) \label{eq:result_GOE}\\\nonumber\\
Q_k^{\R}(2m+1,2n+1)&=2^{k+1}Q^{\Hh}_{k}\left(m,n\right)+2^{k}\Gamma\left(k+m-n+1/2\right)f_{k}(m, n) \label{eq:result_LOE}\\\nonumber\\
Q_k^{\R}(2\alpha_1,2\alpha_2,2n+1)&=2Q^{\Hh}_{k}\left(\alpha_1,\alpha_2,n\right)+\dfrac{\Gamma\left(k+\alpha_1+1/2\right)}{\Gamma\left(k+\alpha_1+\alpha_2+2n+1\right)}f_{k}(\alpha_1,\alpha_2,n) \label{eq:result_JOE}\
\end{align}
where
\begin{equation*}
\begin{cases}f_{k}(n) =  c_{n}P^{(1/4)}_{n}(-i(k+1/4);\pi/2)\\\\
f_{k}(m, n)= c_{n,m}P^{(m-n+1/2)}_{2n}(-ik;\pi/2)\\\\ 
f_{k}(\alpha_1,\alpha_2,n) = c_{n,\alpha_1,\alpha_2}p_{2n}(-ik; \alpha_1+\frac{1}{2},-\alpha_1-\alpha_2-2n, \alpha_1+\frac{1}{2},-\alpha_1-\alpha_2-2n). 
\end{cases}
\end{equation*}
In each case, the $f_{k}$ is a hypergeometric orthogonal polynomial from the Askey scheme: The $P^{(\lambda)}_{n}(x,\phi)$ are the Meixner-Pollaczek polynomials, while $p_{n}(x;a,b,c,d)$ are the continuous Hahn polynomials, see \eqref{eq:ContinuousHahn}. The 
normalization constants are
\begin{equation*}
\begin{cases} c_{n} = \dfrac{i^{n}n!}{\Gamma\left(n+1/2\right)} \\
c_{n,m} = \dfrac{(-1)^{n}n!}{\Gamma\left(m+1/2\right)} \\ 
c_{n,\alpha_1,\alpha_2} = \dfrac{(-1)^{n}n!\Gamma\left(\alpha_{1}+1/2\right)\Gamma\left(n+\alpha_1+\alpha_2+1\right)}{\Gamma\left(n+\alpha_1+1/2\right)\Gamma\left(n+\alpha_2+1/2\right)}. 
\end{cases}
\end{equation*}
\label{thm:ortho}
\end{theorem}
\begin{rmk}
The results~\eqref{eq:result_GOE}-\eqref{eq:result_JOE} combined with Theorem \ref{thm:symp} imply a polynomial property for the moments of the orthogonal ensemble, though not as cleanly as in the symplectic case. It is not possible to normalize $Q^{\mathbb{R}}_{k}(2n+1)$ and obtain a polynomial in $k$, unlike in the symplectic and unitary cases (\textit{e.g.} the second term in \eqref{eq:result_GOE} is always exponentially larger in $k$ than the first). 
\end{rmk}
\par
The content of Theorems~\ref{thm:GOE1} and~\ref{thm:LOE1} is that the combinations $S_{k}^{\R}(n)$ and $S_{k}^{\R}(m,n)$ of moments of the orthogonal ensembles (with fixed $n$) have hypergeometric orthogonal polynomial factors.
Putting together the dualities in the $n$-odd orthogonal ensembles of Theorem~\ref{thm:duality} and the classical duality between symplectic moments and formal orthogonal moments~\eqref{eq:dualityG}-\eqref{eq:dualityL}-\eqref{eq:dualityJ}, we find that the combinations of moments (with fixed $k$)
\begin{align*}
T_{k}^{\R}(n)&=Q_k^{\R}(2n+1)+(-2)^kQ_k^{\R}(-2n)\\
T_{k}^{\R}(m,n)&=Q_k^{\R}(2m+1,2n+1)+(-2)^kQ_k^{\R}(-2m,-2n)\\
T_{k}^{\R}(\alpha_1,\alpha_2,n)&=Q_k^{\R}(2\alpha_1,2\alpha_2,2n+1)+Q_k^{\R}(-2\alpha_1,-2\alpha_2,-2n)
\end{align*}
do have hypergeometric polynomial factors.
\begin{cor}
\begin{align}
T_{k}^{\R}(n)&=(2k-1)!!\;M_n\left(k;\;1/2,-1\right)\\
T_{k}^{\R}(m,n)
&=\frac{2^{m+n+k}}{\sqrt{\pi}}\frac{\Gamma(k+\alpha+1/2)}{(2n-1)!!(2m-1)!!}\; S_n((ik)^2;\;1/2,0,\alpha+1/2)
\end{align}
\begin{align}
T_{k}^{\R}(\alpha_1,\alpha_2,n)&=\dfrac{\Gamma\left(k+\alpha_1+1/2\right)\Gamma\left(\alpha_{1}+1/2\right)\Gamma\left(n+\alpha_1+\alpha_2+1\right)}{\Gamma\left(k+\alpha_1+\alpha_2+2n+1\right)\Gamma\left(n+\alpha_1+1/2\right)\Gamma\left(n+\alpha_2+1/2\right)}\nonumber\\
&\frac{(2(n+\alpha_2)-1)!!}{(2\alpha_2-1)!!(2n-1)!!}
(-1)^{n}W_n\left((ik)^2;\;\frac{1}{2},0, \alpha_1+\frac{1}{2},-\alpha_1-\alpha_2-2n\right).
\end{align}
\end{cor}
\begin{proof}[Proof of Theorem \ref{thm:ortho}]
We multiply both sides of identity \eqref{dualitydens} by $x^{k}$ (or $|x|^{2k}$ for the GOE) and integrate over $I$. By the correspondence \eqref{tracedens} this gives
\begin{equation}
Q^{\mathbb{R}}_{k}(2n+1) = 2\tilde{Q}^{\mathbb{H}}_{k}(n)+\psi_{n}(k) \label{OEresult}
\end{equation}
where $\tilde{Q}^{\mathbb{H}}_{k}(n)$ are moments defined with respect to the modified weights \eqref{modweight}. Such moments are easily expressed in terms of the usual $Q^{\mathbb{H}}_{k}(n)$ by multiplying by $2^{k}$ (Hermite and Laguerre case) or by dividing the parameters $\alpha_{1}$ and $\alpha_2$ by $2$ (Laguerre and Jacobi cases). This gives the first terms in \eqref{eq:result_GOE}-\eqref{eq:result_JOE}. 

The correction $\psi_{n}(k)$ is a weighted Mellin transform of the corresponding orthogonal polynomial. In the JOE and LOE this takes the form
\begin{equation} 
\psi_{n}(k) = \frac{\int_{I}x^{k}w_{1}(x)p_{2n}(x)\,dx}{\int_{I}w_{1}(x)p_{2n}(x)\,dx}. \label{mnk}
\end{equation}
while for the GOE $x^{k}$ is replaced with $|x|^{2k}$. The integral~\eqref{mnk} can be computed explicitly by expanding $p_{2n}(x)$ as a sum and integrating term by term. This expansion turns out to be a terminating hypergeometric series which can be identified as one of the hypergeometric polynomials appearing in the claimed result. In fact, for the Gaussian and Laguerre cases, precisely this calculation is carried out in a different context in \cite{Coffey07}, so let us just explain the Jacobi case. Then the monic polynomials $p_{2n}(x)$ are proportional to the usual Jacobi polynomials which can be written down explicitly (see \textit{e.g.} \cite[Eq. 4.32]{Sze}). Integrating term by term in \eqref{mnk} gives
\begin{equation*}
\begin{split}
&\psi_{n}(k) = \frac{n!\Gamma(\frac{\alpha_{1}+\alpha_{2}}{2}+1+n)}{\Gamma(\frac{\alpha_{1}+1}{2}+n)\Gamma(\frac{\alpha_{2}+1}{2}+n)\Gamma(\frac{\alpha_{1}+\alpha_{2}}{2}+2n+k+1)}\\
&\times \sum_{i=0}^{2n}(-1)^{i}\binom{2n+\alpha_{2}}{2n-i}\binom{2n+\alpha_{1}}{i}\Gamma\left(\frac{\alpha_{2}+1}{2}+i+k\right)\Gamma\left(\frac{\alpha_{1}+1}{2}+2n-i\right).
\end{split}
\end{equation*}
The sum can be matched with a hypergeometric function and we obtain
\begin{equation*}
\psi_{n}(k)  = \frac{\eta_{n,\alpha_{1},\alpha_{2}}\Gamma\left(\frac{\alpha_{2}+1}{2}+k\right)}{\Gamma\left(1+\frac{\alpha_{1}+\alpha_{2}}{2}+k+2n\right)}\,{}_3 F_2\left(\begin{matrix}k+(\alpha_{2}+1)/2,-2n,-\alpha_{1}-2n \\\alpha_{2}+1,-(\alpha_{1}-1)/2-2n \end{matrix};1\right) \label{3F2Jacobi}
\end{equation*}
where
\begin{equation*}
\eta_{n,\alpha_{1},\alpha_{2}} = \frac{n!\Gamma(1+\frac{\alpha_{1}+\alpha_{2}}{2}+n)\Gamma(1+\alpha_{2}+2n)\Gamma(\frac{1+\alpha_{1}}{2}+2n)}{(2n)!\Gamma(1+\alpha_{2})\Gamma(\frac{1+\alpha_{1}}{2}+n)\Gamma(\frac{1+\alpha_{2}}{2}+n)}.
\end{equation*}
Finally, comparing \eqref{3F2Jacobi} with the definition of the continuous Hahn polynomial in \eqref{eq:ContinuousHahn} gives the result.
\end{proof}
\begin{rmk}
In the Gaussian and Laguerre cases, the evaluation of the integral \eqref{mnk} already appeared in the literature on special functions, see the work of Bump et al. \cite{Bump00}, Coffey et al.\cite{Coffey07, Coffey15}, though no connection to random matrix theory is made. These works show that the quantity \eqref{mnk} satisfies a functional equation and Riemann hypothesis with critical line $\mathrm{Re}(k)=-1/2$ (Hermite polynomials) and $\mathrm{Re}(k)=0$ (Laguerre polynomials) in our notation. We believe it is new that precisely these Mellin transforms should appear in the context of random matrices. The last and most complicated case of Jacobi appears to be absent from the literature. This turns out to be a continuous Hahn polynomial $p_{2n}(-ik;a,b,c,d)$ with $a=c>0$ and $b=d<0$. The analogous properties in this case are most easily proved by noticing that the continuous Hahn polynomial can be represented in terms of the Wilson polynomial with a negative fourth parameter. Explicitly, one has
\begin{multline*}
p_{2n}(-ik; \alpha_1+1/2,-\alpha_1-\alpha_2-2n, \alpha_1+1/2,-\alpha_1-\alpha_2-2n)=\\
\frac{(2\alpha_2+1)(2\alpha_2+3)\cdots(2\alpha_2+2n-1)}{(2n-1)!!\; n!}
W_n\left((ik)^2;\;\frac{1}{2},0, \alpha_1+\frac{1}{2},-\alpha_1-\alpha_2-2n\right)
\end{multline*}
This identity demonstrates that the Mellin transforms satisfy a symmetry on the line $\mathrm{Re}(k)=0$ (the polynomials are invariant under $k \to -k$). Furthermore, by the orthogonality property \eqref{neretin}, we can deduce that the zeros all lie on the imaginary axis (this does not seem to be obvious from the Hahn polynomial representation).
\end{rmk}
We now study the orthogonal ensemble with $n$ even. In this case the analytic structure of the Mellin transform seems to be more complicated and remains somewhat mysterious to us. For this reason we restrict ourselves to the Gaussian case, though analogous results for Laguerre and Jacobi could easily be derived. We are able to prove an analytic continuation of $Q^{\mathbb{R}}_{k}(2n)$ to an entire function of $k$ as in the previous sections, but with a more complicated structure. We first consider the simplest case $n=2$ where this structure already appears. Directly integrating $|x|^{2k}$ against the density \eqref{eq:integrals} with $\beta=1$ and a Gaussian weight gives
\begin{equation}
Q_{k}^{\mathbb{R}}(2) = \frac{2^{k}\Gamma(k+1/2)}{\sqrt{\pi}}+\frac{1}{\sqrt{2}}\int_{0}^{\infty}x^{2k}x\,\mathrm{erf}(x/2)e^{-x^{2}/4}\,dx. \label{q2}
\end{equation}
Clearly, the first term above has a similar structure to that already observed in the GSE and GUE. But the second term, which is a weighted Mellin transform of the error function, is different. It is analytic in the half-plane $\mathrm{Re}(2k)>-1$ and standard properties of Mellin transforms show that it extends as to an analytic function in the entire complex plane except for simple poles when $2k+1 \in \{-2,-4,-6,\ldots\}$. Since these simple poles are eliminated on dividing by $\Gamma(k+1/2)$, this gives an entire function of $k$. In fact this analytic continuation can be given in terms of a hypergeometric function:
\begin{align*}
Q_{k}^{\mathbb{R}}(2) &= 2^{2k+3/2}  {}_2 F_1\left(\begin{matrix}1/2,k+3/2 \\3/2 \end{matrix}; -1\right)\Gamma(k+3/2)/\sqrt{\pi}+\frac{2^{k}\Gamma(k+1/2)}{\sqrt{\pi}}\\
&=(2k-1)!!\;\left((2k+1)M_{k}(-1;\;3/2,1/2)+1\right).
\end{align*}
This hypergeometric function reduces to a polynomial whenever $k$ is a positive integer. But its analytic continuation to $k \in \mathbb{C}$ appears more complicated than in the previously considered cases. Indeed one has the asymptotics (see \textit{e.g.} \cite{Tem02}):
\begin{equation*}
\begin{split}
&{}_2 F_1\left(\begin{matrix}1/2,k+3/2 \\3/2 \end{matrix}; -1\right) \sim \frac{1}{2}\sqrt{\frac{\pi}{k}}, \qquad k \to +\infty,\\
&{}_2 F_1\left(\begin{matrix}1/2,k+3/2 \\3/2 \end{matrix}; -1\right) \sim -k2^{-3/2-k}, \qquad k \to -\infty.
\end{split}
\end{equation*}

For say $n=4,6,8,\ldots$ and so on, this structure persists and follows a similar pattern. As for the GSE, the results of \cite{Simm11} are again useful here, providing a general formula for the GOE moments:
\begin{equation}
Q^{\mathbb{R}}_{k}(2n) = Q^{\mathbb{C}}_{k}(2n-1)-\sum_{j=0}^{n-1}\sum_{i=0}^{n-j-1}{k \choose i}{k \choose i+j}\frac{(n-i-j)_{(k+1/2)}}{(n-j)_{(1/2)}}+A^{\mathbb{R}}_{k}(2n). \label{goeformula}
\end{equation}
where
\begin{equation*}
A^{\mathbb{R}}_{k}(2n) := c_{n}\int_{0}^{\infty}x^{2k}e^{-x^{2}/2}H_{2n-1}(x)\mathrm{erf}(x/\sqrt{2})\,dx.\label{akn}
\end{equation*}
The first two terms in \eqref{goeformula} have a simple analytic structure, similar to that found in the symplectic case. The term $A^{\mathbb{R}}_{k}(2n)$ is the generalization to larger $n$ of the second term in \eqref{q2}.
\begin{prop}
For any positive integer $n$, the ratio $Q^{\mathbb{R}}_{k}(2n)/\Gamma(k+1/2)$ has an analytic continuation to an entire function of $k$.
\end{prop}
\begin{proof}
It is clear that the first two terms in \eqref{goeformula} yield a polynomial in $k$ after dividing by $\Gamma(k+1/2)$. The third term \eqref{akn} is the Mellin transform of the function $\phi_{n}(x) = e^{-x^{2}/2}H_{2n-1}(x)\mathrm{erf}(x/\sqrt{2})$ which is analytic in the right-half plane $\mathrm{Re}(2k)>-3$. To extend to the left-half plane $\mathrm{Re}(2k) \leq -3$ it suffices to notice that $\phi_n(x)$ has an asymptotic expansion near $x=0$ with respect to the sequence $\{x^{2k}\}_{k\geq1}$.
Therefore $A^{\mathbb{R}}_{2k}(2n)$ has an meromorphic continuation into the left-half plane except for simple poles when $2k+1 =-2,-4,-6,\ldots$. Precisely these poles are eliminated after dividing through by the factor $\Gamma(k+1/2)$.
\end{proof}

\section*{Acknowledgements}
The research of FDC and NO'C is supported by ERC Advanced Grant 669306. The research of FDC is partially supported by the Italian National Group of Mathematical Physics (GNFM-INdAM). FM acknowledges support from EPSRC Grant No. EP/L010305/1. NS acknowledges support from a Leverhulme Trust Early Career Fellowship ECF-2014-309. We would  like to thank Philippe Biane for helpful conversations at an earlier stage of this work, in particular for drawing our attention to the papers of Bump \textit{et al.}~\cite{Bump86,Bump00}.
We are grateful to Peter Forrester and Brian Winn for valuable remarks on the first version of the paper. We also thank the referees who indicated the papers of Mehta and Normand~\cite{Mehta01}, and Forrester and Witte~\cite{Forrester01} on duality relations for moments of characteristic polynomials of random matrices, and made us aware that the duality in Proposition~\ref{prop:dualdens} already appears in Forrester's book~\cite[Eq.~(6.120)-(6.122)]{Forrester_book}.

\appendix

\section{Orthogonal and symplectic ensembles: duality}
\label{app:orthsymp}
The purpose of this section is to prove Proposition~\ref{prop:dualdens} in the three classical ensembles. The proof is based on explicit results for the eigenvalue density of orthogonal and symplectic ensembles obtained by Adler et al.~\cite{AFNvM}. To begin with, we introduce the notation and relevant results obtained in \cite{AFNvM}. There, the notation $e^{-2V(x)}$ is equivalent to our $w_{2}(x)$ as in \eqref{eq:weights}. 

We begin by denoting by $p_{n}(x)$ the unique degree $n$ monic polynomial orthogonal with respect to $w_{2}(x)$ and we set
\begin{equation}
h_{n} = \int_{I}w_{2}(x)p_{n}(x)^{2}\,dx.
\end{equation}
An important quantity in the theory is the ratio
\begin{equation}
2V'(x) = \frac{g(x)}{f(x)}
\end{equation}
where $f$ and $g$ are polynomials of minimal degree, with $f \geq 0$. For the classical weights, this implies
\begin{equation}
f(x) = \begin{cases} 1 & \text{GOE}\\
x & \text{LOE}\\
x(1-x) & \text{JOE}
\end{cases}
\label{fpoly}
\end{equation}
Then we define \textit{modified potentials}
\begin{equation}
V_{1}(x) = V(x)+\frac{1}{2}\log f(x), \qquad V_{4}(x) = V(x)-\frac{1}{2}\log f(x)
\end{equation}
and eigenvalue densities $\tilde{\rho}_{n}^{(1)}(x)$ with respect to the weight $e^{-V_{1}(x)}$ for $\beta=1$ and $\tilde{\rho}_{n}^{(4)}(x)$ with respect to the weight $e^{-2V_{4}(x)}$ for $\beta=4$. We have that $e^{-2V_{4}(x)} = \tilde{w}_{4}(x)$ are precisely the modified weights \eqref{modweight}. On the other hand, $e^{-V_{1}(x)} = w_{1}(x)$ and so $\tilde{\rho}^{(1)}_{n}(x) = \rho^{(1)}_{n}(x)$.

The first result we need is \cite[Eq.~(4.18)]{AFNvM} which writes the density in the $n$-odd orthogonal ensemble as
\begin{equation}
\begin{split}
\rho^{(1)}_{2n+1}(x) &= \tilde{\rho}^{(1)}_{2n}(x)-\gamma_{2n-2}\tilde{s}_{2n-2}\frac{e^{-V_{1}(x)}}{\tilde{s}_{2n}}\left(\tilde{\phi}_{2n}(x)p_{2n-1}(x)-p_{2n}(x)\tilde{\phi}_{2n-1}(x)\right)\\
&+\frac{e^{-V_{1}(x)}p_{2n}(x)}{2\tilde{s}_{2n}}, \label{afndens}
\end{split}
\end{equation}
where we define
\begin{align}
\tilde{s}_{n} &= \frac{1}{2}\int_{I}e^{-V_{1}(x)}p_{n}(x)\,dx\\
\tilde{\phi}_{j}(x) &= \frac{1}{2}\int_{I}e^{-V_{1}(y)}\operatorname{sgn}(x-y)p_{j}(y)\,dy \label{phitilde}\\
\gamma_{n}h_{n} &= \begin{cases} 1 & \GOE,\\
\frac{1}{2} & \LOE,\\
\frac{1}{2}(2n+\alpha_1+\alpha_2+2) & \JOE.\end{cases} \label{gamn}\\
h_{n} &= \begin{cases} 
n!\sqrt{\pi}2^{-n} & \GOE\\
n!\Gamma(\alpha+n+1) & \LOE\\
\frac{\Gamma(\alpha_{1}+n+1)\Gamma(\alpha_{2}+n+1)\Gamma(n+1)\Gamma(\alpha_{1}+\alpha_{2}+n+1)}{\Gamma(\alpha_{1}+\alpha_{2}+2n+1)\Gamma(\alpha_{1}+\alpha_{2}+2n+2)} & \JOE
\end{cases} \label{hn}
\end{align}
We also recall the classical identity (see \textit{e.g.} \cite[Chap. 5]{Mehta})
\begin{equation}
\rho^{(2)}_{n}(x) = e^{-2V(x)}\sum_{j=0}^{n-1}\frac{p_{j}(x)^{2}}{h_{j}}. \label{unitarydensapp}
\end{equation}

The integrals $\tilde{s}_{n}$ happen to be known for all positive integers $n$.
\begin{lem}
For all three classical weights and any positive integer $n$, we have $\tilde{s}_{2n-1}=0$ and
\begin{equation}
\tilde{s}_{2n} = \begin{cases} 
\sqrt{2\pi}\frac{(2n)!}{4^{n}n!} & \GOE\\
2^{\frac{\alpha-1}{2}}\Gamma\left(n+\frac{\alpha_1+1}{2}\right) & \LOE\\
2^{\alpha_{1}+\alpha_{2}}16^{n}\frac{\Gamma\left(n+\frac{1}{2}\right)\Gamma\left(n+\frac{\alpha_{1}+\alpha_{2}+1}{2}\right)\Gamma\left(\frac{\alpha_{1}+1}{2}+n\right)\Gamma\left(\frac{\alpha_{2}+1}{2}+n\right)}{\pi \Gamma\left(\alpha_{1}+\alpha_{2}+4n+1\right)} & \JOE.
\end{cases}
\end{equation}
\label{lem:sn}
\end{lem}
\begin{proof}
For the GOE case the fact that $\tilde{s}_{2n-1}=0$ follows from symmetry and the formula for $\tilde{s}_{2n}$ is in \cite[Sec. 4]{AFNvM}. In the LOE and JOE cases these facts are less obvious, but were derived by Nagao and Forrester in \cite[A.2 and A.7]{NF95} based on evaluations in terms of hypergeometric functions.
\end{proof} 

Note that the expression in the second line of \eqref{afndens} corrects a typo in \cite[Eq. 4.18]{AFNvM}. The formula for $\rho^{(1)}_{2n}(x)$ in \eqref{afndens} is given in \cite[Eq. 4.12]{AFNvM} as 
\begin{equation}
\rho^{(1)}_{2n}(x) = \rho^{(2)}_{2n-1}(x)+\gamma_{2n-2}e^{-V_{1}(x)}p_{2n-1}(x)\phi_{2n-2}(x)
\end{equation}
On the other hand, formula \cite[Eq. 4.27]{AFNvM} writes the density in the symplectic ensemble as
\begin{equation}
\rho^{(4)}_{n}(x) = \frac{1}{2}\rho^{(2)}_{2n}(x)-\frac{1}{2}\gamma_{2n-1}e^{-2V(x)+V_{4}(x)}p_{2n}(x)\int_{I}e^{-2V(y)+V_{4}(y)}p_{2n-1}(y)\mathbbm{1}_{y>x}\,dy \label{sympdens}
\end{equation}
But by definition
\begin{equation}
e^{-2V(x)+V_{4}(x)} = e^{-V(x)-\frac{1}{2}\log f(x)} = e^{-V_{1}(x)}
\end{equation}
So
\begin{equation}
\tilde{\rho}^{(4)}_{n}(x) = \frac{1}{2}\rho^{(2)}_{2n}(x)-\frac{1}{2}\gamma_{2n-1}e^{-V_{1}(y)}p_{2n}(x)\int_{I}e^{-V_{1}(y)}p_{2n-1}(y)\mathbbm{1}_{y>x}\,dy
\end{equation}
Now we must demonstrate the relation \eqref{dualitydens}. 

\begin{proof}[Proof of Proposition \ref{prop:dualdens}]
The point of the proof is that \eqref{afndens} can be simplified considerably. There are two calculations required in the proof, which in the GOE case were carried out in \cite[Eqs 4.19 and 4.13]{AFNvM} respectively. The first claim is that the following identity holds in all three cases:
\begin{equation}
\gamma_{2n-2}\frac{\tilde{s}_{2n-2}}{\tilde{s}_{2n}} = \gamma_{2n-1}
\end{equation}
This can be verified by direct computation using the above explicit formulae for $\tilde{s}_{2n}$ and $\gamma_{n}$. Then formula \eqref{afndens} becomes
\begin{equation}
\begin{split}
\rho^{(1)}_{2n+1}(x) &=\rho^{(2)}_{2n-1}(x)+\gamma_{2n-2}e^{-V_{1}(x)}p_{2n-1}(x)\tilde{\phi}_{2n-2}(x)\\
&-\gamma_{2n-1}e^{-V_{1}(x)}\left(\tilde{\phi}_{2n}(x)p_{2n-1}(x)-p_{2n}(x)\tilde{\phi}_{2n-1}(x)\right)+\frac{e^{-V_{1}(x)}p_{2n}(x)}{2\tilde{s}_{2n}}, \label{afndens2}
\end{split}
\end{equation}
The second claim is the following identity
\begin{equation}
e^{-V_{1}(y)}p_{2n-1}(y)\left(\gamma_{2n-2}\tilde{\phi}_{2n-2}(x)-\gamma_{2n-1}\tilde{\phi}_{2n}(x)\right) = e^{-V_{1}(y)-V_{1}(x)}f(x)\frac{p_{2n-1}(x)p_{2n-1}(y)}{h_{2n-1}} \label{phident}
\end{equation}
where $f(x)$ is given by \eqref{fpoly}. Setting $x=y$ and inserting it into \eqref{afndens2} gives the simplification
\begin{equation}
\rho^{(1)}_{2n+1}(x) =\rho^{(2)}_{2n}(x)+\gamma_{2n-1}e^{-V_{1}(x)}p_{2n}(x)\tilde{\phi}_{2n-1}(x)+\frac{e^{-V_{1}(x)}p_{2n}(x)}{2\tilde{s}_{2n}}, \label{afndens3}
\end{equation}
where we used the explicit form of the $\beta=2$ density \eqref{unitarydensapp} and that
\begin{equation}
e^{-2V_{1}(x)+\log f(x)} = e^{-2V(x)} = w_{2}(x).
\end{equation}
This now looks very similar to \eqref{sympdens}. Indeed, the proof is complete if we can check that $\tilde{\phi}_{2n-1}(x) = -\int_{I}w_{1}(y)p_{2n-1}(y)\mathbbm{1}_{y>x}\,dy$. Comparing with the definition \eqref{phitilde} we see that 
\begin{equation}
\tilde{\phi}_{2n-1}(x) = -\int_{I}w_{1}(y)p_{2n-1}(y)\mathbbm{1}_{y>x}\,dy+\frac{1}{2}\int_{I}w_{1}(y)p_{2n-1}(y)\,dy, \label{phind}
\end{equation}
but the second integral is $\tilde{s}_{2n-1}$ which is zero by Lemma \ref{lem:sn}. This immediately implies
\begin{equation}
\rho^{(1)}_{2n+1}(x) =2\tilde{\rho}^{(4)}_{n}(x)+\frac{e^{-V_{1}(x)}p_{2n}(x)}{2\tilde{s}_{2n}}
\end{equation}
as given in the statement of Proposition \ref{prop:dualdens}. It remains to check identity \eqref{phident}. Cancelling $e^{-V_{1}(y)}p_{2n-1}(y)$, it is equivalent to checking
\begin{equation}
\gamma_{2n-2}\tilde{\phi}_{2n-2}(x)-\gamma_{2n-1}\tilde{\phi}_{2n}(x) =f(x)e^{-V_{1}(x)}\frac{p_{2n-1}(x)}{h_{2n-1}}. \label{phident2}
\end{equation}
Differentiating both sides of \eqref{phident2} with respect to $x$ reduces the claim to
\begin{equation}
h_{2n-1}\gamma_{2n-2}p_{2n-2}(x)-h_{2n-1}\gamma_{2n-1}p_{2n}(x) = \frac{1}{w_{1}(x)}\frac{d}{dx}\left(f(x)w_{1}(x)p_{2n-1}(x)\right) \label{polyident}
\end{equation}
Now using standard differential identities for the classical orthogonal polynomials and some routine calculation shows that \eqref{polyident} is a consequence of the three term recurrence relation. 
\end{proof}
\section{Mellin transform}
\label{app:Mellin}
We summarise here some properties of the Mellin transform (and its extension).
The Mellin transform of  $f(x)$ is defined by the integral
\begin{equation}
\MM\left[f(x);s\right] = \int_0^\infty f(x)x^{s-1}dx,
\label{eq:app_Mellin}
\end{equation}
when it exists. We set $f^*(s) = \MM\left[f(x);s\right]$. 
\par
In general, the integral~\eqref{eq:app_Mellin} converges and defines a holomorphic function $f^*(s)$ only in a vertical strip $D$ of the complex plane. It turns out that, in the frequently occurring case where $f(x)$ is of rapid decay at infinity and has an asymptotic expansion $f(x)\sim \sum_{j=0}^{\infty}a_jx^{b_j}$ as $x\to0^+$ (as in all instances in this paper), the Mellin transform $f^*(s)$ has a meromorphic continuation to the whole complex plane with simple poles of residue $a_j$ at $s=-b_j$. For more details on meromorphic extensions of Mellin transforms see~\cite{Zagier06}.
\par
If the integral~\eqref{eq:app_Mellin} converges in the strip $D$, then the following relations hold:
\begin{align}
\MM\left[f^{(m)}(x);s\right] &= (-1)^m(s - m)_m f^*(s - m)&s-m&\in D\label{eq:Mellin_d}\\
\MM\left[x^mf(x);s\right] &= f^*(s +m)&s+m&\in D\label{eq:Mellin_m}\\
\MM\left[x^mf^{(m)}(x);s\right] &=(-1)^m(s)_mf^*(s)&s&\in D.\label{eq:Mellin_md}
\end{align} 
\par
Suppose that $f(s)$ and $g(s)$ have Mellin transforms $f^*(s)$ and $g^*(s)$, respectively, analytic in a vertical strip $D$ in the complex plane.  Take any $c\in D$.  Then
\begin{equation}
\label{eq:convolution_theorem}
\mathcal{M}\left[f(x)g(x);s\right] = \frac{1}{2\pi i}\int_{c -i\infty}^{c + i\infty} f^*(s - u)g^*(u)du. 
\end{equation} 
whenever the Mellin transfom of $(fg)(x)$ exists.
\section{Hypergeometric orthogonal polynomials}
\label{app:neoOP}
We report a few basic properties of some families of hypergeometric OP's. 
\subsection{Wilson}
The Wilson polynomials
are solutions of the discrete Sturm-Liouville problem~\cite[Section 9.1]{Koekoek10}
\be
B(x)y(x+i)-\left[B(x)+D(x)\right]y(x)+D(x)y(x-i)=n(n+a+b+c+d-1)y(x),
\label{eq:Wilson_SL}
\ee
where
\[
y(x)=W_n(x^2;a,b,c,d)
\]
and
\[
\left\{  \begin{array}{l@{\quad}cr} 
\displaystyle 
B(x)=\frac{(a-ix)(b-ix)(c-ix)(d-ix)}{2ix(2ix-1)}\\\\
\displaystyle
D(x)=\frac{(a+ix)(b+ix)(c+ix)(d+ix)}{2ix(2ix+1)}.
\end{array}\right.
\]
In this paper we have considered the less conventional situation when $a,b,c,1-d>0$. For this range of the parameters, Neretin~\cite[Section 3.3]{Neretin02} found the orthogonality relation
\begin{align}
&\frac{1}{2\pi}\int\limits_{\R_+}\left|\frac{\Gamma(a+ix)\Gamma(b+ix)\Gamma(c+ix)}{\Gamma(1-d+ix)\Gamma(2ix)}\right|^2 W_m(x^2;\;a,b,c,d)W_n(x^2;\;a,b,c,d)\de x\nonumber\\
&=\frac{a+b+c+d-1}{a+b+c+d+2n-1}
\frac{(a+b)_n(a+c)_n(a+d)_n(b+c)_n(b+d)_n(c+d)_n}{(a+b+c+d-1)_n}\nonumber\\
&\small\times\frac{\Gamma(a+b)\Gamma(a+c)\Gamma(b+c)\Gamma(1-a-b-c-d)}{\Gamma(1-a-d)\Gamma(1-b-d)\Gamma(1-c-d)}n!\;
\delta_{mn}, \label{neretin}
\end{align}
for $n,m<1-a-b-c-d$. 
\subsection{Continuous dual Hahn}
The continuous dual Hahn polynomials $S_n(x^2;a,b,c)$ can be found from the Wilson polynomials by dividing by $(a + d)_n$ and letting $d\to\infty$.
If $a$, $b$ and $c$ are positive, then~\cite[Section 9.3]{Koekoek10}
\begin{align}
\frac{1}{2\pi}\int\limits_{\R_+}\left|\frac{\Gamma(a+ix)\Gamma(b+ix)\Gamma(c+ix)}{\Gamma(2ix)}\right|^2 S_m(x^2;\;a,b,c)S_n(x^2;\;a,b,c)\de x\nonumber\\
=\Gamma(n+a+b)\Gamma(n+a+c)\Gamma(n+b+c)n!\;
\delta_{mn}.
\end{align}
The continuous dual Hahn polynomials are solution of the discrete Sturm-Liouville problem
\be
B(x)y(x+i)-\left[B(x)+D(x)\right]y(x)+D(x)y(x-i)=ny(x),
\ee
where
\[
y(x)=S_n(x^2;a,b,c)
\]
and
\[
\left\{  \begin{array}{l@{\quad}cr} 
\displaystyle 
B(x)=\frac{(a-ix)(b-ix)(c-ix)}{2ix(2ix-1)}\\\\
\displaystyle
D(x)=\frac{(a+ix)(b+ix)(c+ix)}{2ix(2ix+1)}.
\end{array}\right.
\]

\subsection{Meixner-Pollaczek}

The Meixner-Pollaczek polynomials
satisfy the orthogonality relation~\cite[Section 9.7]{Koekoek10} 
\begin{align}
\frac{1}{2\pi}\int_{\R}e^{(2\phi-\pi)x}|\Gamma(\lambda+ix)|^2P_m^{(\lambda)}(x;\phi)P_n^{(\lambda)}(x;\phi)\de x\nonumber\\
=\frac{\Gamma(n+2\lambda)}{(2\sin\phi)^{2\lambda}n!}\;
\delta_{mn},\quad\lambda>0\quad\text{and}\quad0<\phi<\pi,
\end{align}
and the Sturm-Liouville equation (set $y(x)=P_n^{(\lambda)}(x;\phi)$):
\be
e^{i\phi}(\lambda -i x)y(x+i)+2ix\cos\phi\; y(x)-e^{-i\phi}(\lambda +i x)y(x-i)=2i(n+\lambda)\sin\phi\;y(x).
\ee
\subsection{Meixner polynomials}
The Meixner polynomials
satisfy the orthogonality relation~\cite[Section 9.7]{Koekoek10} 
\begin{align}
\sum_{x=0}^{\infty}\frac{(\beta)_x}{x!}c^xM_m(x;\;\beta,c)M_n(x;\;\beta,c)=
\frac{n!}{(\beta)_nc^n(1-c)^{\beta}}\delta_{mn},\quad \text{$\beta>0$ and $0<c<1$}.
\end{align}

\end{document}